\documentclass[11pt,letterpaper]{article}

\usepackage{setspace}

\usepackage{bm,aliascnt,amsthm,amsmath,amsfonts,rotate,datetime,amsmath,ifthen,eurosym,wrapfig,cite,url,cite,amsfonts,amssymb,ifthen,color,wrapfig,rotate,lmodern,aliascnt,datetime,graphicx,algorithmic,algorithm,enumerate,enumitem,todonotes,fancybox,bm,tabularx,colortbl,xspace,graphicx,algorithmic,algorithm}

\usepackage{soul} 
 
\usepackage[]{datetime}
\usepackage{todonotes}
\usepackage{subcaption}
\usepackage{enumerate}
\usepackage{cite}
\usepackage{enumitem,xspace}   

\usepackage[letterpaper, top=1in, bottom=1in, left=1in, right=1in, includefoot]{geometry}

\newcommand{\stf}[1]{\textcolor{black}{#1}}

\usepackage[pdftex,
backref=page,
plainpages = false,
pdfpagelabels,
hyperfootnotes=false,
bookmarks=true,
bookmarksopen = true,
bookmarksnumbered = true,
breaklinks = true,
hyperfigures,
pagebackref,
urlcolor = magenta,
urlcolor = black!30!green,
anchorcolor = green,
hyperindex = true,
colorlinks = true,
linkcolor = black!30!blue,
citecolor = black!30!green
]{hyperref}

\newcounter{func}

\newcommand{\funref}[1]{\hyperref[#1]{f_{\ref*{#1}}}} 

\newcounter{con}

\newcommand{\conref}[1]{\hyperref[#1]{c_{\ref*{#1}}}} 

\usepackage[polutonikogreek,english,]{babel}
\usepackage[utf8]{inputenc}

\usepackage[T1]{fontenc} 
\usepackage{alphabeta}

\newtheorem{observation}{Observation}
\newtheorem{proposition}{Proposition}
\newtheorem{lemma}{Lemma}
\newtheorem{corollary}{Corollary}
\newtheorem{claim}{Claim}
\newtheorem{theorem}{Theorem}
\newtheorem{definition}{Definition}

\usepackage{color}
\definecolor{MidnightBlack}{rgb}{0.1,0.1,.34}
\definecolor{MidnightBlue}{rgb}{0.1,0.1,0.44}
\definecolor{Black}{rgb}{0,0, 0}
\definecolor{Blue}{rgb}{0, 0 ,1}
\definecolor{Red}{rgb}{1, 0 ,0}
\definecolor{White}{rgb}{1, 1, 1}
\definecolor{Grey}{rgb}{.6, .6, .6}
\definecolor{Mygreen}{rgb}{.0, .7, .0}
\definecolor{Yellow}{rgb}{.55,.55,0}
\definecolor{Mustard}{rgb}{1.0, 0.86, 0.35}
\definecolor{applegreen}{rgb}{0.55, 0.71, 0.0}
\definecolor{darkturquoise}{rgb}{0.0, 0.81, 0.82}
\definecolor{celestialblue}{rgb}{0.29, 0.59, 0.82}
\definecolor{green_yellow}{rgb}{0.68, 1.0, 0.18}
\definecolor{crimsonglory}{rgb}{0.75, 0.0, 0.2}
\definecolor{darkmagenta}{rgb}{0.30, 0.0, 0.30}
\definecolor{internationalorange}{rgb}{1.0, 0.31, 0.0}
\definecolor{darkorange}{rgb}{1.0, 0.55, 0.0}
\definecolor{ao}{rgb}{0.0, 0.5, 0.0}
\definecolor{awesome}{rgb}{1.0, 0.13, 0.32}
\definecolor{dark-red}{rgb}{0.4,0.15,0.15}
\definecolor{dark-blue}{rgb}{0.15,0.15,0.4}
\definecolor{medium-blue}{rgb}{0,0,0.5}
\definecolor{gray}{rgb}{0.5,0.5,0.5}
\definecolor{color-Ig}{rgb}{0.15,0.7,0.15}

\newcommand{\hh}{\end{document}}

\usepackage{float}
\usepackage{tikz,pgf}
\usetikzlibrary{backgrounds}
\usetikzlibrary{calc,3d}
\usetikzlibrary{intersections,decorations.pathmorphing,shapes,decorations.pathreplacing,fit,fadings} 

\usetikzlibrary{arrows.meta}

\makeatletter

\pgfdeclarearrow{
  name = ipe _linear,
  defaults = {
    length = +1bp,
    width  = +.666bp,
    line width = +0pt 1,
  },
  setup code = {
    \pgfarrowssetbackend{0pt}
    \pgfarrowssetvisualbackend{
      \pgfarrowlength\advance\pgf@x by-.5\pgfarrowlinewidth}
    \pgfarrowssetlineend{\pgfarrowlength}
    \ifpgfarrowreversed
      \pgfarrowssetlineend{\pgfarrowlength\advance\pgf@x by-.5\pgfarrowlinewidth}
    \fi
    \pgfarrowssettipend{\pgfarrowlength}
    \pgfarrowshullpoint{\pgfarrowlength}{0pt}
    \pgfarrowsupperhullpoint{0pt}{.5\pgfarrowwidth}
    \pgfarrowssavethe\pgfarrowlinewidth
    \pgfarrowssavethe\pgfarrowlength
    \pgfarrowssavethe\pgfarrowwidth
  },
  drawing code = {
    \pgfsetdash{}{+0pt}
    \ifdim\pgfarrowlinewidth=\pgflinewidth\else\pgfsetlinewidth{+\pgfarrowlinewidth}\fi
    \pgfpathmoveto{\pgfqpoint{0pt}{.5\pgfarrowwidth}}
    \pgfpathlineto{\pgfqpoint{\pgfarrowlength}{0pt}}
    \pgfpathlineto{\pgfqpoint{0pt}{-.5\pgfarrowwidth}}
    \pgfusepathqstroke
  },
  parameters = {
    \the\pgfarrowlinewidth,%
    \the\pgfarrowlength,%
    \the\pgfarrowwidth,%
  },
}

\pgfdeclarearrow{
  name = ipe _pointed,
  defaults = {
    length = +1bp,
    width  = +.666bp,
    inset  = +.2bp,
    line width = +0pt 1,
  },
  setup code = {
    \pgfarrowssetbackend{0pt}
    \pgfarrowssetvisualbackend{\pgfarrowinset}
    \pgfarrowssetlineend{\pgfarrowinset}
    \ifpgfarrowreversed
      \pgfarrowssetlineend{\pgfarrowlength}
    \fi
    \pgfarrowssettipend{\pgfarrowlength}
    \pgfarrowshullpoint{\pgfarrowlength}{0pt}
    \pgfarrowsupperhullpoint{0pt}{.5\pgfarrowwidth}
    \pgfarrowshullpoint{\pgfarrowinset}{0pt}
    \pgfarrowssavethe\pgfarrowinset
    \pgfarrowssavethe\pgfarrowlinewidth
    \pgfarrowssavethe\pgfarrowlength
    \pgfarrowssavethe\pgfarrowwidth
  },
  drawing code = {
    \pgfsetdash{}{+0pt}
    \ifdim\pgfarrowlinewidth=\pgflinewidth\else\pgfsetlinewidth{+\pgfarrowlinewidth}\fi
    \pgfpathmoveto{\pgfqpoint{\pgfarrowlength}{0pt}}
    \pgfpathlineto{\pgfqpoint{0pt}{.5\pgfarrowwidth}}
    \pgfpathlineto{\pgfqpoint{\pgfarrowinset}{0pt}}
    \pgfpathlineto{\pgfqpoint{0pt}{-.5\pgfarrowwidth}}
    \pgfpathclose
    \ifpgfarrowopen
      \pgfusepathqstroke
    \else
      \ifdim\pgfarrowlinewidth>0pt\pgfusepathqfillstroke\else\pgfusepathqfill\fi
    \fi
  },
  parameters = {
    \the\pgfarrowlinewidth,%
    \the\pgfarrowlength,%
    \the\pgfarrowwidth,%
    \the\pgfarrowinset,%
    \ifpgfarrowopen o\fi%
  },
}

\newdimen\ipeminipagewidth

\tikzstyle{ipe import} = [
  x=1bp, y=1bp,
  ipe node stretch/.store in=\ipenodestretch,
  ipe stretch normal/.style={ipe node stretch=1},
  ipe stretch normal,
  ipe node/.style={
    anchor=base west, inner sep=0, outer sep=0, scale=\ipenodestretch
  },
  ipe mark scale/.store in=\ipemarkscale,
  ipe mark tiny/.style={ipe mark scale=1.1},
  ipe mark small/.style={ipe mark scale=2},
  ipe mark normal/.style={ipe mark scale=3},
  ipe mark large/.style={ipe mark scale=5},
  ipe mark normal, 
  ipe circle/.pic={
    \draw[line width=0.2*\ipemarkscale]
      (0,0) circle[radius=0.5*\ipemarkscale];
    \coordinate () at (0,0);
  },
  ipe disk/.pic={
    \fill (0,0) circle[radius=0.6*\ipemarkscale];
    \coordinate () at (0,0);
  },
  ipe fdisk/.pic={
    \filldraw[line width=0.2*\ipemarkscale]
      (0,0) circle[radius=0.5*\ipemarkscale];
    \coordinate () at (0,0);
  },
  ipe box/.pic={
    \draw[line width=0.2*\ipemarkscale, line join=miter]
      (-.5*\ipemarkscale,-.5*\ipemarkscale) rectangle
      ( .5*\ipemarkscale, .5*\ipemarkscale);
    \coordinate () at (0,0);
  },
  ipe square/.pic={
    \fill
      (-.6*\ipemarkscale,-.6*\ipemarkscale) rectangle
      ( .6*\ipemarkscale, .6*\ipemarkscale);
    \coordinate () at (0,0);
  },
  ipe fsquare/.pic={
    \filldraw[line width=0.2*\ipemarkscale, line join=miter]
      (-.5*\ipemarkscale,-.5*\ipemarkscale) rectangle
      ( .5*\ipemarkscale, .5*\ipemarkscale);
    \coordinate () at (0,0);
  },
  ipe cross/.pic={
    \draw[line width=0.2*\ipemarkscale, line cap=butt]
      (-.5*\ipemarkscale,-.5*\ipemarkscale) --
      ( .5*\ipemarkscale, .5*\ipemarkscale)
      (-.5*\ipemarkscale, .5*\ipemarkscale) --
      ( .5*\ipemarkscale,-.5*\ipemarkscale);
    \coordinate () at (0,0);
  },
  /pgf/arrow keys/.cd,
  ipe arrow normal/.style={scale=1},
  ipe arrow tiny/.style={scale=.4},
  ipe arrow small/.style={scale=.7},
  ipe arrow large/.style={scale=1.4},
  ipe arrow normal,
  /tikz/.cd,
  ipe arrows/.style={
    ipe normal/.tip={
      ipe _pointed[length=1bp, width=.666bp, inset=0bp,
                   quick, ipe arrow normal]},
    ipe pointed/.tip={
      ipe _pointed[length=1bp, width=.666bp, inset=0.2bp,
                   quick, ipe arrow normal]},
    ipe linear/.tip={
      ipe _linear[length = 1bp, width=.666bp,
                  ipe arrow normal, quick]},
    ipe fnormal/.tip={ipe normal[fill=white]},
    ipe fpointed/.tip={ipe pointed[fill=white]},
    ipe double/.tip={ipe normal[] ipe normal},
    ipe fdouble/.tip={ipe fnormal[] ipe fnormal},
    ipe arc/.tip={ipe normal},
    ipe farc/.tip={ipe fnormal},
    ipe ptarc/.tip={ipe pointed},
    ipe fptarc/.tip={ipe fpointed},
  },
  ipe arrows, 
]

\tikzset{
  rgb color/.code args={#1=#2}{%
    \definecolor{tempcolor-#1}{rgb}{#2}%
    \tikzset{#1=tempcolor-#1}%
  },
}

\makeatother

\usepackage{color}

\tikzset{red node/.style={draw=red, circle, fill = red, minimum size = 4pt, inner sep = 0pt}}
\tikzset{yellow node/.style={draw=yellow, circle, fill = yellow, minimum size = 4pt, inner sep = 0pt}}
\tikzset{blue node/.style={draw=celestialblue, circle, fill =celestialblue, minimum size = 4pt, inner sep = 0pt}}
\tikzset{triangle/.style = { regular polygon, regular polygon sides=3, rotate=180}}
\tikzset{small red/.style={draw=red, triangle, fill = red, minimum size = 2pt, inner sep = 0pt}}

\tikzset{black node/.style={draw, circle, fill = black, minimum size = 4pt, inner sep = 0pt}}
\tikzset{small black node/.style={draw, circle, fill = black, minimum size = 3pt, inner sep = 0pt}}
\tikzset{model node/.style={draw=celestialblue, circle, fill = celestialblue, minimum size = 5pt, inner sep = 0pt}}
\tikzset{model node small/.style={draw=celestialblue, circle, fill = celestialblue, minimum size = 3pt, inner sep = 0pt}}
\tikzset{rep node/.style={draw=red, circle, fill = red, minimum size = 3pt, inner sep = 0pt}}
\tikzset{track node 1/.style={draw, circle, fill = black, minimum size = 2pt, inner sep = 0pt}}
\tikzset{track node 2/.style={draw=black!30!white, circle, fill = black!30!white, minimum size = 2pt, inner sep = 0pt}}
\tikzset{track node 3/.style={draw=black!10!white, circle, fill = black!10!white, minimum size = 2pt, inner sep = 0pt}}
\tikzfading[name=fade out, inner color=transparent!0, outer
color=transparent!100]
\tikzfading[name=fade in, inner color=transparent!100, outer
color=transparent!0]
\tikzset{terminal/.style={circle, draw=black, fill=black!50,
                        inner sep=0pt, minimum width=4pt}}
\tikzset{root terminal/.style={circle,thick, draw=black, fill=orange!90!yellow,
                        inner sep=1pt, minimum width=6pt}}
\tikzset{simple/.style={circle, draw=black, fill=black,
                        inner sep=0pt, minimum width=4pt}}
\tikzset{treenode/.style={circle, draw=black, fill=black,
                        inner sep=0pt, minimum width=4pt}}
\tikzset{
diagonal fill/.style 2 args={fill=#2, path picture={
\fill[#1, sharp corners] (path picture bounding box.south west) -|
                         (path picture bounding box.north east) -- cycle;}},
reversed diagonal fill/.style 2 args={fill=#2, path picture={
\fill[#1, sharp corners] (path picture bounding box.north west) |- 
                         (path picture bounding box.south east) -- cycle;}}
}

\newcommand{\yes}{{\sf yes}}

\newcommand{\remove}[1]{}

\newcommand{\fpt}{{\sf fpt}\xspace}

\renewcommand{\O}{\mathcal{O}}

\newcommand{\J}{\mathcal{J}}
\newcommand{\x}{\chi}

\RequirePackage{stmaryrd}
\usepackage{textcomp}
\DeclareUnicodeCharacter{2286}{\subseteq}
\DeclareUnicodeCharacter{2192}{\ifmmode\to\else\textrightarrow\fi}
\DeclareUnicodeCharacter{2203}{\ensuremath\exists}
\DeclareUnicodeCharacter{183}{\cdot}
\DeclareUnicodeCharacter{2200}{\forall}
\DeclareUnicodeCharacter{2264}{\leq}
\DeclareUnicodeCharacter{2265}{\geq}
\DeclareUnicodeCharacter{8614}{\mathbin{\mapsto}}
\DeclareUnicodeCharacter{8656}{\Leftarrow}
\DeclareUnicodeCharacter{8657}{\Uparrow}
\DeclareUnicodeCharacter{8658}{\Rightarrow}
\DeclareUnicodeCharacter{8659}{\Downarrow}
\DeclareUnicodeCharacter{8669}{\rightsquigarrow}
\newcommand{\eqdef}{\stackrel{{\scriptsize\rm def}}{=}}
\DeclareUnicodeCharacter{8797}{\eqdef}
\DeclareUnicodeCharacter{8870}{\vdash}
\DeclareUnicodeCharacter{8873}{\Vdash}
\DeclareUnicodeCharacter{22A7}{\models}
\DeclareUnicodeCharacter{9121}{\lceil}
\DeclareUnicodeCharacter{9123}{\lfloor}
\DeclareUnicodeCharacter{9124}{\rceil}
\DeclareUnicodeCharacter{2208}{\in}
\DeclareUnicodeCharacter{9126}{\rfloor}
\DeclareUnicodeCharacter{9655}{\triangleright}
\DeclareUnicodeCharacter{9665}{\triangleleft}
\DeclareUnicodeCharacter{9671}{\diamond}
\DeclareUnicodeCharacter{9675}{\circ}
\DeclareUnicodeCharacter{10178}{\bot}
\DeclareUnicodeCharacter{10214}{} 
\DeclareUnicodeCharacter{10215}{} 
\DeclareUnicodeCharacter{10229}{\longleftarrow}
\DeclareUnicodeCharacter{10230}{\longrightarrow}
\DeclareUnicodeCharacter{10231}{\longleftrightarrow}
\DeclareUnicodeCharacter{10232}{\Longleftarrow}
\DeclareUnicodeCharacter{10233}{\Longrightarrow}
\DeclareUnicodeCharacter{10234}{\Longleftrightarrow}
\DeclareUnicodeCharacter{10236}{\longmapsto}
\DeclareUnicodeCharacter{10238}{\Longmapsto} 
\DeclareUnicodeCharacter{10503}{\Mapsto}    
\DeclareUnicodeCharacter{10971}{\mathrel{\not\hspace{-0.2em}\cap}}
\DeclareUnicodeCharacter{65294}{\ldotp}
\DeclareUnicodeCharacter{65372}{\mid}

\usepackage{titling}
\thanksmarkseries{arabic}

\hypersetup{
    colorlinks, linkcolor={dark-red},
    citecolor={dark-blue}, urlcolor={medium-blue}
}

\newcommand{\N}{\ensuremath{\mathbb{N}\xspace}}
\renewcommand{\O}{\ensuremath{\mathcal{O}}\xspace}

\newcommand{\ssigma}{\text{\tt {\Sigma}}}
\newcommand{\msynt}{\text{\tt {m}}}
\newcommand{\eqsynt}{\text{\tt {=}}}

\newcommand{\NP}{\ensuremath{\mathsf{NP}}\xspace}

\newcommand{\mbprivate}[1]{}

\newcommand{\IPP}{{\sc I2PP}\xspace}
\newcommand{\IPHS}{{\sc I2PHS}\xspace}
\newcommand{\TPT}{{\sc TPT}\xspace}
\newcommand{\FVST}{{\sc FVST}\xspace}

\newcommand{\I}{\mathcal{I}} 
\renewcommand{\P}{\mathcal{P}} 

\newcommand{\T}{{T}} 
\newcommand{\npair}{\psi}
\newcommand{\Q}{\mathcal{Q}} 
\newcommand{\BI}{\mathcal{BI}} 

\newcommand{\subsetint}{\sqsubseteq}
\newcommand{\subsetintstrict}{\sqsubset}
\newcommand{\nsubsetint}{\nsubseteq}
\newcommand{\unionint}{\sqcup}
\newcommand{\interint}{\sqcap}

\newcommand{\defparproblem}[4]{\par
 \vspace{3mm}
\noindent\fbox{
 \begin{minipage}{0.96\textwidth}
 \begin{tabular*}{\textwidth}{@{\extracolsep{\fill}}lr} #1 & {\bf{Parameter:}} #3 \vspace{1mm} \\ \end{tabular*}
 {\textbf{Input:}} #2
	\vspace{1mm}\\%
 {\textbf{Question:}} #4
 \end{minipage}
 }
 \vspace{3mm}
\par
}

\theoremstyle{plain}

\newcommand{\claimqed}{$\square$}

\hypersetup{
colorlinks = true,
linkcolor = black!30!blue,
citecolor = black!30!green
}

\usepackage{tikz}

\usetikzlibrary{patterns}
\usetikzlibrary{calc,3d}
\usetikzlibrary{intersections,decorations.pathmorphing,shapes,decorations.pathreplacing,fit}
\usetikzlibrary{shapes.geometric,hobby,calc}
\usetikzlibrary{decorations}
\usetikzlibrary{decorations.pathmorphing}
\usetikzlibrary{decorations.text}
\usetikzlibrary{shapes.misc}
\usetikzlibrary{decorations,shapes,snakes}

\begin{document}

\title{\Large Kernelization for Graph Packing and Hitting Problems\\ via Rainbow Matching
}

\author{\bigskip\large 
Stéphane Bessy\thanks{LIRMM, Univ de Montpellier, CNRS, Montpellier, France.}~$^{,}$\thanks{Supported by the ANR project DIGRAPHS (ANR-19-CE48-0013).} \and
Marin Bougeret\footnotemark[1]\and 
Dimitrios M. Thilikos\footnotemark[1]~$^{,}$\thanks{Supported by the ANR projects ESIGMA (ANR-17-CE23-0010) and the French-German Collaboration ANR/DFG Project UTMA (ANR-20-CE92-0027).}\and\large 
\and\large 
Sebastian Wiederrecht\thanks{Discrete Mathematics Group, Institute for Basic Science, Daejeon, South Korea.}~$^{,}$\thanks{The research of Sebastian Wiederrecht was supported by the ANR project ESIGMA (ANR-17-CE23-0010) and the Institute for Basic Science (IBS-R029-C1).}}

\date{}

\pagenumbering{Alph}
\maketitle

\begin{abstract}
\noindent
\stf{We introduce a new kernelization tool, called {\sl rainbow matching
technique}, that is appropriate for the design of polynomial kernels
for packing problems {and their hitting counterparts}.  Our technique capitalizes on the powerful
combinatorial results of [{\sl Graf, Harris, Haxell, SODA 2021}].  We
apply the rainbow matching technique on four (di)graph packing or
hitting problems, namely the \textsc{Triangle-Packing in Tournament}
problem (\TPT), where we ask for a packing of $k$ directed triangles
in a tournament, \textsc{Directed Feedback Vertex Set in Tournament}
problem (\FVST), where we ask for a (hitting) set of at most $k$
vertices which intersects all triangles of a tournament,
the \textsc{Induced 2-Path-Packing} (\IPP) where we ask for a packing
of $k$ induced paths of length two in a graph and \textsc{Induced
2-Path Hitting Set} problem (\IPHS), where we ask for a (hitting) set
of at most $k$ vertices which intersects all induced paths of length
two in a graph.  The existence of a sub-quadratic kernels for these
problems was proven for the first time in [{\sl Fomin, Le, Lokshtanov,
Saurabh, Thomassé, Zehavi.  ACM Trans. Algorithms, 2019}], where they
gave a kernel of $\O(k^{3/2})$ vertices for the two first problems and
$\O(k^{5/3})$ vertices for the two last. In the same paper it was
questioned whether these bounds can be (optimally) improved to linear
ones.  Motivated by this question, we apply the rainbow matching
technique and prove that \TPT and \FVST admit (almost linear) kernels
of $k^{1+\frac{\O(1)}{\sqrt{\log{k}}}}$ vertices and that \IPP
and \IPHS admit kernels of $\O(k)$ vertices.}
\end{abstract}


{\sf 
%
%
%
%
}

\thispagestyle{empty}

\newpage\thispagestyle{empty}

\tableofcontents\thispagestyle{empty}

\newpage
\pagenumbering{arabic}
\newpage

\setcounter{page}{1}

\section{Introduction}

A \emph{parameterized problem} $\Pi$ is a subset of
$\Sigma^{*}\times \mathbb{N}$ for some finite alphabet $\Sigma.$ Given
an instance $(x,k)$ of parameterized problem $\Pi$ we typically refer
to $|x|$ as the {\em size} of the problem and to $k$ as the {\em
parameter} of the problem and the general question of parameterized
computation is whether $\Pi$ admits an algorithm able to decide, given
an instance $(x,k)\in \Sigma^{*}\times \mathbb{N},$ whether $(x,k)\in
Π$ in $f(k)\cdot |x|^{\O(1)}$ time. When this is indeed possible, then
we say that $Π$ is {\sl Fixed Parameter Tractable} or, in short {\sf
FPT}.  Parameterized computation was introduced by Downey and Fellows
in their pioneering work
in~\cite{DowneyF95-I,DowneyF95-II,DowneyF93,AbrahamsonDF95fixed,DowneyF93fixe}
and currently constitutes a fully developed discipline of Theoretical
Computer Science
(see \cite{CyganFKLMPPS15para,FlumGrohebook,Niedermeier-book,DowneyF13fund}
for textbooks).

\medskip\noindent{\bf Kernelization algorithms.} A particularly vibrant field of parameterized computation is {\sl kernelization}. A {\sl kernelization algorithm} for a parameterized problem   $\Pi\subseteq \Sigma^{*}\times \mathbb{N}$ is a polynomial algorithm  able to reduce every instance $(x,k)\in \Sigma^{*}\times \mathbb{N}$ 
to an {\sl equivalent one} whose size depends {\sl exclusively} on the
parameter $k.$ If this size is bounded by a (polynomial) function
$f(k),$ then we say that $Π$ admits a (polynomial) kernel of size
$f(k).$ The design of kernelization algorithms for parameterized
problems has been a prominent topic of parameterized computation,
mainly because it can be seen as a way to formalize preprocessing:
when some parameterization of an {\sf NP}-hard problem admits a kernel
of size $f(k),$ then we may solve it by first applying, as a
preprocessing step, the corresponding kernelization algorithm and then
apply brute force techniques on a problem instance where the size
$|x|$ of the problem has been radically reduced. Clearly, for small
values of $k,$ this approach becomes particularly promising,
especially when the problem in question admits a polynomial
kernel. Unfortunately, not all parameterized problems are amenable to
polynomial kernels and an extensive theory of kernelization has been
developed so to either provide algorithmic techniques for the
derivation of polynomial kernels (see
e.g., \cite{FominS14kerne,GuoN07invit,Kratsch14recen,LokshtanovMS12kerne})
or to develop complexity-theoretic lower bound tools for
kernelization \cite{BodlaenderDFH09onpr,Drucker15newli,FortnowS11infea,HermelinKSWW15acomp,DellM12kerne,DellM14satis,HermelinW12weakc}
(see \cite{FominLSZ19kern} for a dedicated textbook).

\medskip\noindent{\bf Greedy localization.}  A wide family of
parameterized problems that have extensively studied from the
kernelization point of view are {\sl (di)graph packing problems}. Such
problems are defined on graphs or digraphs and, typically, the
question is whether an input (di)graph $G$ contains some collection of
$k$ pairwise disjoint copies of some fixed (induced) sub(di)graph $H.$
A general approach for such problem is the {\sl greedy localization
  technique} (see \cite{DehneFRS04greed,SloperT08anove}). This
technique consists in finding, using some greedy approach, a maximal
collection $\mathcal{C}$ of pairwise disjoint copies of $H$ in $G.$ If
$\mathcal{C}$ has already at least $k$ elements, we may safely output
a positive answer to the problem.  If not, then we know that the $|H|
(k-1)=\O(k)$ vertices of the (di)graphs in $\mathcal{C}$ should {\sl
  cover} every possible solution of the problem. Based on this last
covering property, the challenge is to design a polynomial time
procedure that may discard all but a polynomial number of vertices
from $G$ so that the remaining (di)graph is an equivalent
instance. This procedure varies depending on the definition of the
problem in question.  Moreover, in case a polynomial kernel exists, a
particular challenge towards deriving kernels of low polynomial size
is to
\begin{quote}
{\sl maximize the set of discarded vertices so that the size of the
resulting (di) graph, that is the output of the kernelization  algorithm, is bounded by a low-polynomial function.}
\end{quote}

\medskip\noindent{\bf Rainbow matching technique.}  In this paper, we
propose a framework for tackling the above challenge that we call {\sl
  the rainbow matching technique}.  Our technique capitalizes on the
deep combinatorial results of Graf and Haxell in \cite{Graf2020} (see
also \cite{GrafHH21algor}). In fact, in \autoref{@alternatives}, we
derive the following ``rainbow-matching'' outcome of the main result
of \cite{Graf2020} asserting that {\sl there exists a polynomial time
  algorithm that, given an edge-multicolored graph $G$ (by $p$
  colors), either outputs a matching of $G$ carrying {\sl all} $p$
  colors or outputs a non-empty set of colors
  $C\subseteq \{1,\ldots,p\}$ and a vertex set $X$ of size $\O(|C|)$
  that intersects all edges colored by the colors of $C$}
(\autoref{cor_rainbowvc}).

For the purposes of our technique, we build an auxiliary graph based
on the maximal solution $\mathcal{C}$ yielded by the greedy
localization routine. We then consider a multi-coloring of the
auxiliary graph based on the maximal solution $\mathcal{C}$. Depending
of the outcome of the above algorithm, we either obtain an equivalent
instance or we give a way to reorganize the parts of the (di)graph
that are not in $\mathcal{C}$ in {\em buckets} in a way that will
permit a recursive application of the above procedure until an
equivalent instance is created. We provide a generic description of
our technique in \autoref{@absentminded}.

\stf{The rainbow matching technique seems to naturally apply to
  packing problems, where a maximum disjoint collection of a specific
  (di)graph $H$ is seeking in the input graph. We illustrate the
  technique on two packing problems. As a side effect, we notice that
  the technique also provides equivalent kernel for the corresponding
  dual problems of the considered ones. These hitting problems consist
  in finding a minimum set which intersects all the copies of $H$ in
  the input graph. In all, we apply the rainbow matching technique to
  four problems that we describe in the following.}

\medskip\noindent{\bf Packing induced paths of length two.}  The first
problem where our technique is applied is the following.

\defparproblem{{\sc Induced 2-path-Packing} (\IPP)}{$(G,k)$ where $G$
  is a graph and $k \in \N$}{$k.$}{Does $G$ contain $k$ pairwise
  disjoint induced paths of length two?}

As in the case of \TPT, the above problem admits a kernel on
$\O(k^{2})$ vertices because of the results of Abu Khzam in
\cite{Abu_Khzam10animp,Abu_Khzam09aquad}.  The first time a
sub-quadratic kernel was given for \IPP\ was the one on $\O(k^{5/3})$
vertices by Fomin, Le, Lokshtanov, Saurabh, Thomassé, and Zehavi in
\cite{FominLLSTZ19subqu}.  Again, an open question that appeared in
\cite{FominLLSTZ19subqu} is whether this bound can be improved to a
linear one. As a second application of our technique, we prove that
this is indeed the case.

\medskip\noindent
\stf{{\bf Induced 2-paths Hitting Set.}  Given a grapĥ
$G$, an {\it induced 2-path hitting set} of $G$ is a set $X$ which
intersect all the induced 2-paths of $G$. In other words,
$G\setminus X$ does not containany induced 2-path. Then, the dual
version of the previous problem, on which we apply the rainbow
matching technique, is the following.}

\defparproblem{{\sc Induced 2-paths Hitting Set} (\IPHS)}{\stf{$(G,k)$ where
  $G$ is a graph and $k \in \N$}}{$k.$}{\stf{Is there an induced
  2-paths hitting set of $G$ of size at least $k$?}}

\stf{Here again, this  problem admits a kernel on
$\O(k^{2})$ vertices because of the results of Abu Khzam in
\cite{Abu_Khzam10akadhs}.  And the first
sub-quadratic kernel for \IPHS\ was on $\O(k^{5/3})$ vertices given by
Fomin, Le, Lokshtanov, Saurabh, Thomassé, and Zehavi in
\cite{FominLLSTZ19subqu}. Here also, using our technique, we prove that \IPHS admits a linear kernel.
}

\medskip\noindent{\bf Packing directed triangles in tournament.}  The
third problem we consider is the following .
\defparproblem{\textsc{Triangle-Packing in Tournament}\
  (\TPT)}{$(T,k)$ where $T$ is a tournament and $k \in \N$}{$k$}{Does
  $T$ contain $k$ pairwise disjoint directed triangles?}
\noindent {Recall that a directed graph $T=(V,E)$ is a  {\em tournament}, if for every distinct $x,y\in V,$ either $(x,y)\in E$ or  $(y,x)\in E$ holds (but not both).}

The \NP-completeness of \TPT follows from the results of
\cite{CharbitTY07themin}. Moreover, given that certain patterns are
excluded from its inputs, it can be solved in polynomial time
\cite{CaiDZ02aminm}. For (non) approximability results on the \TPT
problem
see~\cite{GuruswamiRCCW98theve,Cygan13impr,MnichWV16aappr,BessyBT17trian}.
Notice that \TPT can directly be reduced to the {\sc 3-Hitting Set}
problem.  For this problem, Abu Khzam gave, in
\cite{Abu_Khzam10animp,Abu_Khzam09aquad}, a kernel on $\O(k^2)$
vertices obtained by only removing vertices. This directly implies
that \TPT admits a kernel on $\O(k^2)$ vertices.  On the negative
side, Bessy, Bougeret, and Thiebaut proved in \cite{BessyBT17trian}
that \TPT\ does not admit a kernel of (total bit) size $\O(k^{2-ε}),$
unless \NP$\subseteq${\sf co-NP}/{\sf Poly}.  They also proved that
\TPT\ admits a kernel of $O(z)$ vertices, when its input instances are
accompanied with a {\sl feedback arc set}\footnote{Given a tournament
  $T=(V,E)$ we say that en edge set $F\subseteq E$ is a {\em feedback
    arc set} of $T$ if the removal if $F$ from $T$ results to an
  acyclic digraph.}  of $T$ of size $z$ and that \TPT\ restricted to
sparse tournaments\footnote{A tournament $T=(V,E)$ is {\em sparse} if
  it contains a feedback arc set that is a matching.} admits a kernel
of $O(k)$ vertices (i.e., of total bit-size $O(k\log{k})$). Towards
breaking the quadratic bound in the general case, Fomin, Le,
Lokshtanov, Saurabh, Thomassé, and Zehavi gave in
\cite{FominLLSTZ19subqu} a kernel for \TPT\ on $\O(k^{3/2})$ vertices.
The open question of \cite{FominLLSTZ19subqu} is whether a kernel on
$O(k)$ vertices exists for the general \TPT.
In this paper, we use the rainbow matching   technique in order to give a kernel for  \TPT\
on $k^{1+\frac{\O(1)}{\sqrt{\log{k}}}}$ vertices.

\medskip\noindent \stf{{\bf Feedback vertex set in tournament.}}
  \stf{The last problem on which we apply the rainbow matching
    technique is the dual version of the previous one and is described
    below. A triangle hitting set in a tournament $T$ is a set $X$
    intersecting all the triangle of $T$. Equivalently, $T\setminus X$
    does not contain any triangle and so, by a classical argument, is
    acyclic. Consequently, such a set $X$ is called a \emph{feedback
      vertex set} of $T$.}

  \defparproblem{\textsc{Feedback Vertex Set in Tournament}\
    (\FVST)}{\stf{$(T,k)$ where $T$ is a tournament and
    $k \in \N$}}{$k$}{\stf{Does $T$ contain a feedback vertex set of
    size at most $k$?}}
\noindent {}

\stf{Here again, this  problem admits a kernel on
$\O(k^{2})$ vertices because of the results of Abu Khzam in
\cite{Abu_Khzam10akadhs}.  And the first
sub-quadratic kernel for \FVST\ was on $\O(k^{3/2})$ vertices given by
Fomin, Le, Lokshtanov, Saurabh, Thomassé, and Zehavi in
\cite{FominLLSTZ19subqu}.
Like for \TPT, our technique leads us to obtain a kernel on
$k^{1+\frac{\O(1)}{\sqrt{\log{k}}}}$ vertices for \FVST.}

\medskip\noindent{\bf Organization of the paper.}

\stf{The definitions of
the basic concepts that we use are given
in \autoref{@specialisation}. In the same section we present the
combinatorial base of the rainbow coloring lemma
(\autoref{cor_rainbowvc}) as well as the proof of how this is derived
by the results of \cite{Graf2020}. In
\autoref{@absentminded}, we proceed with a
generic description of our technique. An overview is presented in
\autoref{se_ssssssssdfgdfgdfghfdhfdhoverview}, while the main
combinatorial assumptions and invariants are presented in
\autoref{se_ssssssssdfgdfgsssdfghfdhfdhoverview} and
\autoref{@maravillados}. The description of
the main algorithmic routine of the technique is given in
\autoref{@gillyflowers}. We then present specialization of the
technique for \IPP in \autoref{@construction}, for \IPHS in
\autoref{@linearKernelIPHS},  for \TPT in
\autoref{@industriousness} and for \FVST in \autoref{@linearKernelFVST}.
We chose to present the almost linear kernels for
\TPT and \FVST after the linear ones for \IPP and \IPHS, as the application of the
technique for the tournament problems is more technical. We conclude
the paper in
\autoref{@wissenschaftlichen} with some remarks and open problems.
}

\section{Definitions}
\label{@specialisation}
We denote by $\Bbb{N}$ the set of non-negative integers and by
$\mathbb{R}^+$ the set of all non-negative reals.  Given two integers
$p$ and $q,$ the set $[p,q]$ refers to the set of every integer $r$
such that $p ≤ r ≤ q.$ For an integer $p≥ 1,$ we set $[p]=[1,p],$
$[c]_0 = [c] \cup \{0\},$ and $\Bbb{N}_{≥ p}=\Bbb{N}\setminus
[0,p-1].$ 

For a set $S,$ we denote by $2^{S}$ the set of all subsets of $S$ and,
given an integer $r∈[|S|],$ we denote by $\binom{S}{r}$ the set of all
subsets of $S$ of size $r.$ {Given two sets $A,B$ and a function $f:
  A\to B,$ for a subset $X\subseteq A$ we use $f(X)$ to denote the set
  $\{f(x)\mid x∈ X\}.$}


Finally, for every set of subsets $X_i,$ $i \in [c],$ and any subset
of indexes $I \subseteq [c],$ we denote $X_I = \bigcup_{i \in
  I}X_i$. And whenever we refer to \emph{a partition of a set $X$ into
  $c$ sets} we refer to an {(ordered)} set
$\mathcal{X}=\{X_{\ell}\mid \ell\in[c]\}$ where $V(\mathcal{X})=X$ and
where we allow that some of the $X_{\ell}$'s is an empty set.

 \subsection{Basic concepts}
 
\paragraph{Parameterized algorithms and kernels.}
\label{subsec_paraalgoandkern}
A \emph{parameterized problem} $\Pi$ is a subset of $\Sigma^{*}\times
\mathbb{N}$ for some finite alphabet $\Sigma.$  A parameterized
problem $\Pi\subseteq \Sigma^{*}\times \mathbb{N}$ is \emph{fixed
parameter tractable} (in short \fpt) if there is an algorithm that, given an instance
$(x,k)\in\Sigma^{*}\times \mathbb{N},$ decides whether $(x,k)\in\Pi$
or not in time $f(k)\cdot p(|x|),$ where $f:\N\to\N$ is some function
and $p\colon\N\to\N$ is a polynomial in the input size.  The notion of
\emph{kernelization} is formally defined as follows.

\begin{definition}[\textrm{Kernelization}]
Let $\Pi\subseteq \Sigma^{*}\times \mathbb{N}$ be a parameterized
problem and $g$ be a computable function.  We say that $\Pi$
\emph{admits a kernel of size $g$} if there exists an algorithm
$\mathcal{K},$ called \emph{kernelization algorithm}, or, in short, a
     {\emph{kernel}} that given $(x,k)\in \Sigma^{*}\times
     \mathbb{Z}^+,$ outputs, in time polynomial in $|x|+k,$ a pair
     $(x',k')\in \Sigma^{*}\times \mathbb{Z}^+$ such that
\begin{itemize}
\item $(x,k)\in \Pi$ 
if and only if $(x',k')\in \Pi,$ and 
\item $\max\{|x'|, k' \}\leq g(k).$
 \end{itemize}
 When $g(k)=k^{\O(1)}$ or $g(k)=\O(k)$ then we say that $\Pi$
 \emph{admits a polynomial} or {\emph{linear kernel}} respectively.
 We refer to $g(k)$ as \emph{the size of the kernel produced by} the
 kernelization algorithm.
\end{definition}

When dealing with graphs algorithmic problems, the produced kernels
are graphs, and we ofen refer to their size by specifying their number
of vertices. For instance, we will say that a graph problem will admit
a kernel with $\O(k)$ vertices (implicitely implying that the total
size of the kernel is $\O(k^2)$).

 \paragraph{Graphs, digraphs and tournaments.}

All graphs in this paper are finite.  We use the term {\em graph}
when we refer to an undirected graph without
loops (ie. an edge with a unique endpoint) neither parallel edges
(ie. distinct edges with the same endpoints).
Also we use the term {\em multigraph} when  we allow for loops and parallel
edges. Directed graphs are called
\emph{digraphs}.

Given a (multi) (di)graph we denote by $V(G)$ and $E(G)$ the set of
its vertices and edges respectively.  Given a (multi) (di)graph $G$
and a set $X\subseteq V(G),$ we denote by $G[X]$ the sub(di)graph of
$G$ induced by $X.$ Similarly, if $X\subseteq E(G),$ we use $G[X]$ in
order to denote the graph $(V_{X},X),$ where $V_{X}$ is the set
containing all the endpoints of the edges in $X.$ {Given a
  graph (resp. digraph) $G$ and two disjoint subsets of vertices $A$
  and $B,$ we say that $e \in E(G)$ is an edge (resp. arc) between $A$
  and $B$ iff $|e \cap A|=|e \cap B|=1.$}   In a digraph, whenever $(x,y)\in
E(G),$ we say that $x$ \emph{dominates} $y.$ We also say that a
digraph $D$ is a \emph{tournament} if, for every
$\{x,y\}\in{V(D)\choose 2},$ either $(x,y)\in E(D)$ or $(y,x)\in
E(D),$ but not both.  If $\mathcal{S}=\{ S_1,\dots, S_{\ell} \}$ is a
collection of vertex sets or subgraphs of some (multi) (di)graph, we
denote by $V(\mathcal{S})$ the set $\bigcup_{i\in[\ell]}S_i$ or
$\bigcup_{i\in[\ell]}V(S_i)$ respectively.  Given a (multi) (di)graph
$G$ and $A \subseteq V(G),$ we denote by $G \setminus A$ the (multi)
(di)graph $G[V(G) \setminus A]$ obtained after the deletion of all
vertices of $A$.

\subsection{Tools about  rainbow matchings}
\label{@alternatives}

\paragraph{Multigraph colourings.}
Let ${G}$ be a multigraph. Recall that for $e\in E(G)$ we say that it
is a \emph{loop} of $G$ if $|e|=1,$ otherwise we say that it is an
\emph{ordinary edge}.  A \emph{$p$-multiedge coloring} of $G$ is a
surjective function $χ\colon E(G)\rightarrow [p]$ that associates to each edge of 
$G$ a color in $[p]$ such that no two parallel edges of $G$ receive
the same color.  We call the pair $(G,χ)$ a \emph{ $p$-edge colored
multigraph} and, given a set of colors $C\subseteq [p]$, 
we define $(G,χ)[C]=G[χ^{-1}(C)]$.

{Given a $p$-edge colored multigraph $(G,χ)$, a \emph{rainbow matching} of $(G,χ)$ is a set}
$X\subseteq E(G)$ such that
\begin{itemize}
\item $X$ is a matching of $G,$ i.e., every two distinct edges of $X$
  are vertex-disjoint and
\item 
{$|X|=p$, and $χ(X)=[p]$.}
\end{itemize}
 For a subset $C\subseteq [p]$ of colors, recall that a set $X$ of vertices of $G$
is a vertex cover of $(G,χ)[C]$ iff  for every $e\in E(G)$ where
 $χ(e)\in C,$ it holds that $e\cap X\neq \emptyset.$  
 We denote by
 $\textsf{vc}_{C}(G,χ)$ the minimum size of a vertex cover of 
 $(G,χ)[C].$

The purpose of this section is to prove the following lemma.

\begin{lemma}
\label{lemma_rainbowvc}
Let $(G,χ)$ be a $p$-edge colored multigraph and $ε \in (0,1).$
If for every subset of colors $C\subseteq [p]$ we have
$\textsf{vc}_{C}(G,χ)> (4+ε)(|C|-1),$ then $G$ contains a
rainbow matching.  Moreover there is an algorithm that, with input
$(G,χ),$ either outputs a rainbow matching of $(G,χ)$ or a
non-empty subset $C$ of $[p]$ and a vertex cover $X$ of $(G,χ)[C]$
such that $|X|\le(4+ε)(|C|-1).$  This algorithm runs in time
$\O(|V(G)|^8\cdot p^{f(ε)}),$ for some function $f\colon
\mathbb{N}\to\mathbb{N}.$
\end{lemma}

To prove \autoref{lemma_rainbowvc}, we only need a light version of
Graf and Haxell's Theorem, appeared as Theorem 4 in~\cite{Graf2020}, which we
present here (whose seminal version appeared in \cite{Haxell95a}).

Let $G$ be a graph enhanced with a partition $\mathcal{V}=(V_1,\dots
,V_p)$ of its vertex set.  An \emph{independent transversal} of $G$
and $\mathcal{V}$ is an independent set $S$ of $G$ satisfying $|S\cap
V_i|=1$ for every $i\in[p].$  For an integer $r,$ the graph $G$ is
\emph{$r$-claw-free for $\mathcal{V}$} if no vertex of $G$ has $r$
independent neighbors in distinct sets $V_i.$  Finally, a set $X$ of
$G$ is a \emph{dominating set} of $G$ if for every vertex $v$ of $G$
there exists a vertex $x$ of $X$ such that $vx$ is an edge of $G.$

\begin{theorem}[Graf and Haxell~\cite{Graf2020}]
  \label{@hervorbringt}
For any fixed $r\in \N$ and $ε \in (0,1)$ there exists an
algorithm that takes as input a graph $G$ and a partition
$\mathcal{V}=(V_1,\dots ,V_p)$ of its vertex set such that $G$ is
$r$-claw-free for $\mathcal{V}$ and produces
	\begin{itemize}
		\item either an independent transversal of $G$ and
                  $\mathcal{V},$ or
		\item a subset $I$ of $\{1,\dots ,p\}$ and a set $X$
                  of vertices of $G$ such that $X$ is a dominating set
                  of $G[\bigcup_{i\in I} V_i]$ and $|X|\le
                  (2+ε)(|I|-1).$
	\end{itemize}
Moreover, the algorithm runs in $\O(|V(G)|^4\cdot p^{f(r,ε)})$
time, for some function $f\colon \mathbb{N}\times (0,1)
\to\mathbb{N}.$
\end{theorem}

\noindent
{\emph Proof of \autoref{lemma_rainbowvc}}. Consider a
$p$-edge colored multigraph $G$  and some $ε \in (0,1).$ 
We build an ``extended line-graph'' $G'$ of $G$ in order to apply
\autoref{@hervorbringt} on it.  See \autoref{@consternation} for an
example of construction. The graph $G'$ is defined as follows:
\begin{eqnarray*}
V(G')=\{(e,χ(e))\mid e\in E(G)\} &  \mbox{and} &  E(G')=\{\{(e,χ(e)),(e',χ(e'))\}\mid e\cap e'\neq\emptyset\}.
\end{eqnarray*}

\begin{figure}[!ht]
\centering
\scalebox{0.6}{\includegraphics{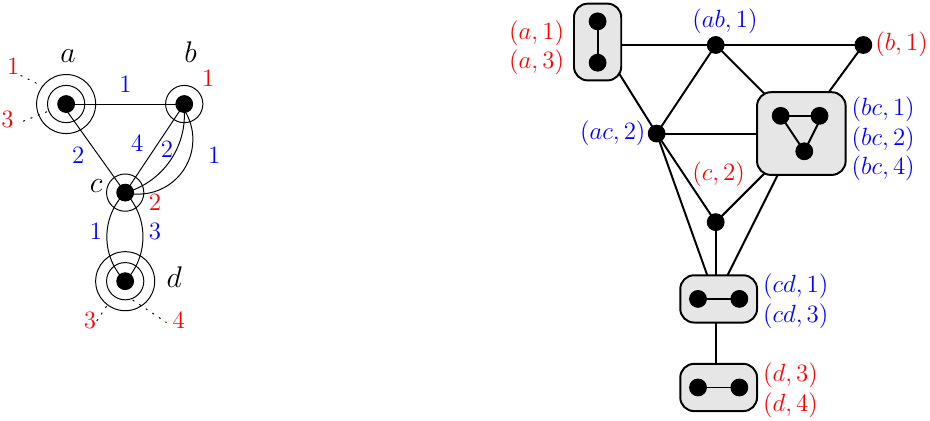}}
\caption{A $p$-edge-colored multigraph $G$ on the left where the sets of
  colors are shown in red for the loops and in blue for the
  ordinary edges of $G.$ On the right, the corresponding extended line-graph of
  $G,$ where the elements in red and in blue respectively correspond
  to loops and ordinary edges of $G.$}
\label{@consternation}
\end{figure}

\noindent Notice that the size of $G'$ is polynomial in the size of
$G$ and $p,$ in particular $|V(G')|=|E(G)| \le \O(|V(G)|^2\cdot p).$  Notice also
that $G'$ is $3$-claw free, independently from the chosen partition of
$\mathcal{V}.$  Indeed the neighborhood of every vertex $(e,χ(e))$ is
either one clique (if $e$ is an loop) or the union of two cliques (if
$e$ is an ordinary edge).

 Now, consider the partition of the vertex set of $G'$ given by the
 colors on the vertices of $G',$ that is $V_i=\{(e,c_0) \in V(G')\mid
 c_0=i \},$ for $i\in[p].$  We then apply on $G'$ the algorithm from
 \autoref{@hervorbringt} with $r=3$ and $ε/2$ as parameters.
 Clearly, this algorithm runs in time $\O(|V(G)|^8\cdot
 p^{f(ε)}),$ for some function
 $f\colon\mathbb{N}\to\mathbb{N}.$

If we obtain an independent transversal of $G',$ then it corresponds
to a rainbow matching of $G.$

Suppose now that the algorithm outputs a subset $I$ of $[p]$ and a set
$X$ of vertices of $G'$ such that $X$ is a dominating set of
$G'[\bigcup_{i\in I} V_i]$ and $|X|\le (2+ε /2)(|I|-1).$  Let
$Y$ be the vertices of $G$ that are involved in an element of $X,$
that is $x\in Y$ if there exists $c_0$ such that $(\{x\},c_0)\in X$ or
there exist $c_0$ and $y\in V(G)$ such that $(\{x,y\},c_0)\in X.$  We
have $|Y|\le 2|X|$ and we can claim that no edge $e$ of $G\setminus Y$
has a color from $I.$  Indeed, towards a contradiction, assume first
that $e=\{x\}$ and $χ(e)=c_0\in I.$  As $X$ is a dominating set of
$G'[\bigcup_{i\in I} V_i],$ there exists an element $w$ of $X$
adjacent to $(\{x\},c_0)$ in $G'.$  If $w$ corresponds to a vertex of $G$ then,
by construction of $G',$ we have $w=(\{x\},c')$ for a color $c'\neq
c_0.$  So, the vertex $x$ belongs to $Y,$ a contradiction.  Now, if
$w$ corresponds to an edge $e=\{u,v\}$ of $G,$ by construction of
$G',$ we have $u=x$ or $v=x$ and again the vertex $x$ belongs to $Y,$
a contradiction.  Finally, assume that $e=\{x,y\}$ and $χ(e)\in I.$
Similarly, there exists an element $w$ of $X$ adjacent to $xy$ in
$G'.$  If $w$ corresponds to a vertex of $G,$ then it must be $x$ or
$y,$ while if $w$ corresponds to an edge of $G,$ then this edge must
have a common endpoint with $xy.$  In both cases, we have $x\in Y$ or
$y\in Y,$ again a contradiction.  So, $Y':=Y\cap V((G,χ)[I])$ is a   
vertex cover of $(G,χ)[I]$ and  $|Y'|\le(4+ε)(|I|-1).$\hfill $\square$\\

\noindent The {\sl rainbow coloring technique} will be based on the  following restatement of \autoref{lemma_rainbowvc}.
\begin{corollary}
\label{cor_rainbowvc}
There exists some function $f\colon \mathbb{N}\to\mathbb{N}$ such that
for every $ε>0$ there is an algorithm that, with input a
$p$-edge-coloring graph $(G,χ),$ outputs either a rainbow matching
of $(G,χ)$ or finds a non-empty subset $C$ of $[p]$ and a vertex cover $X$
of   $(G,χ)[C]$ such that  $|X|<
(4+ε)|C|.$  Moreover, this algorithm runs in time $\O(|V(G)|^8\cdot
p^{f(ε)}).$
\end{corollary}

\section{The rainbow matching technique}
\label{@absentminded}

Both our kernelization  algorithms for \IPP and \TPT will be based on 
the the rainbow matching technique that we next describe in a generic form.
\subsection{Overview of the kernelization algorithm}

\label{se_ssssssssdfgdfgdfghfdhfdhoverview}

To that end, let us consider a generic $H$-packing problem (where $|H|=3$) where given an input $(G,k),$ the objective is to decide if there exist 
$k$ vertex disjoints sets $P_i \subseteq V(G)$ such that for any $i,$ $G[P_i]$ is isomorphic to $H.$
The rainbow matching technique that we introduce in this article can be summarized as follow. Given an input $(G,k)$:
\begin{enumerate}
\item Maintain a partition $(W,B,C)$ of $V(G)$ with size and structure invariants, called a \emph{partial decomposition}, where
  $B \cup C$ is small (typically $\mathcal{O}(k)$), and $W$ is the large part where we want to select an approriate subset. 
\item At each round, apply the following rule:
  \begin{itemize}
  \item define an auxiliary edge-colored multigraph graph $(G,χ)\langle W,B,C\rangle$ with vertex set $W$.
  \item applies Corollary~\ref{cor_rainbowvc} on $(G,χ)\langle W,B,C\rangle$:
    \begin{itemize}
  \item if there exists a rainbow matching $M$ in $(G,χ)\langle W,B,C\rangle,$ stop the kernel and output $G[V(M) \cup B \cup C]$ (we selected $V(M)$ from $W$)
  \item otherwise, use the subset of colors $X$ and its small vertex cover  $T \subseteq W$  to compute
    a new partial decomposition $(W',B',C')$
    \end{itemize}
    \end{itemize}
\end{enumerate}
We point out that the partial decomposition may contain other information than $(W,B,C)$, as it is the case of example for \TPT here, but for the sake of simplicity 
we stick to the triplet $(W,B,C)$ in this generic presentation.
As invariants of a partial decomposition we used for \IPP and \TPT have a lot of similarities, we now explain the common ideas behind these invariants.

\subsection{Origin of the invariants in a partial decomposition.}\label{se_ssssssssdfgdfgsssdfghfdhfdhoverview}
We start with a greedy localization phase where we compute a maximal $H$-packing $\mathcal{P}=\{H_1,\dots,$ $H_{i_0}\},$ and we assume that $i_0 < k$ as otherwise we
get a \yes-instance. We define $C^0=V(\mathcal{P})$ and $W^0 = V(G) \setminus C^0.$ Observe that there is no copy of $H$ inside $W^0,$ and that $|C^0| = \mathcal{O}(k).$
The goal is to select a subset $S \subseteq W^0$ and to ouput $G[S \cup C^0].$ We want $S$ to be small, typically $|S| = \mathcal{O}(k),$ and safe, in the sense that any packing $\mathcal{P}$ in $G$ can be restructured into a packing $\mathcal{P'}$ such that $|\mathcal{P'}| \ge |\mathcal{P}|$ and $V(\mathcal{P'}) \subseteq S \cup C^0.$ This will imply that if $(G,k)$ is a \yes-instance,
then $(G[S \cup C^0],k)$ is also a \yes-instance, and thus that these instances are equivalent as the other implication is straightforward.

Then, we define an auxiliary graph $(G,χ)\langle W^0,\emptyset,C^0\rangle$ as follows (we used the $\langle W^0,\emptyset,C^0\rangle$ notation to
match notation of Section~\ref{se_ssssssssdfgdfgdfghfdhfdhoverview}, as $B=\emptyset$ at the begining). Let $V((G,χ)\langle W^0,\emptyset,C^0\rangle) = W^0,$ and for any $c \in C^0$ and any $\{u,v\} \subseteq W^0$ such that $G[\{c,u,v\}]$ is isomorphic to $H,$
add to $(G,χ)\langle W^0,\emptyset,C^0\rangle$ edge $e=\{u,v\}$ and set color $\x(e)=c.$
{Notice that $C^0$ both denotes a subset of vertices in $G$, and the set of colors of $(G,χ)\langle W^0,\emptyset,C^0\rangle$ and this is a convention that we follow all over the paper.}
Now, apply Corollary~\ref{cor_rainbowvc} on $(G,χ)\langle W^0,\emptyset,C^0\rangle,$ and let us discuss the two possible outcomes of this corollary.

\paragraph*{Case where a rainbow matching exists.} We now consider te case where we have a rainbow matching $M$ in $(G,χ)\langle W^0,\emptyset,C^0\rangle.$ In this case, we stop and return $G[V(M) \cup C^0].$
Observe first that $V(M)$ is small, as $|V(M)| = 2|C^0|=\mathcal{O}(k).$ Let us now examine why it is safe.
Consider an $H$-packing $\mathcal{P}$ in $G.$ We restructure $\mathcal{P}$ as follows (see Figure~\ref{@knuckleheads}). For any $H_i \in \mathcal{P}$ we define
a corresponding $H'_i$ as follows (and we define $\mathcal{P}'=\{H'_i \mid H_i \in \mathcal{P}\}$):
\begin{itemize}
\item If $H_i \subseteq C^0,$ define $H'_i=H_i.$
\item Otherwise, we know that $H_i \cap C^0 \neq \emptyset$ by the maximality of the greedy localization. Let $c \in H_i \cap C^0,$ and
  $e_c=\{u_c,v_c\}$ the edge of $M$ such that $\x(e_c)=c.$ Restructure $H_i$ into $H'_i = \{c,u_c,v_c\}.$
\end{itemize}
We can observe that the  $\{H'_i\}$ are vertex disjoint as $M$ is a matching, and thus that $V(M)$ is safe.

\begin{figure}[!ht]
\centering
\scalebox{0.6}{\includegraphics{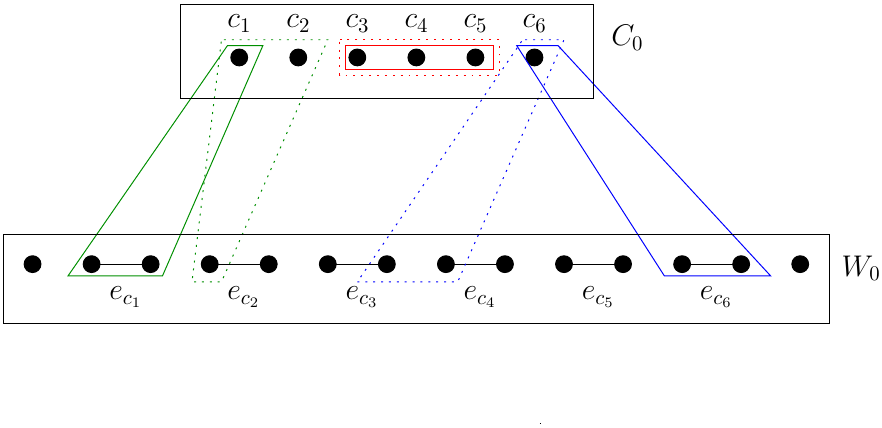}}
\caption{An example showing how to restructure $H'_1,H'_2,H'_3$ (in dashed lines) to $H_1,H_2,H_3$ (in plain lines) using the rainbow matching $\{e_{c_i},1 \le i \le 6\}.$}
\label{@knuckleheads}
\end{figure}

\paragraph*{Case where a small vertex cover exists.} Suppose there exists a subset of colors $X \subseteq C^0$ and a vertex cover $T$ of $(G,χ)\langle W^0,\emptyset,C^0\rangle[X]$ such that $|T| \le (4+ε)|X|.$
Let $B = T \cup X,$ $W = W^0 \setminus T,$ and $C = C^0 \setminus X.$ 
Observe that as $T$ is a vertex cover  of $(G,χ)\langle W^0,\emptyset,C^0\rangle[X],$  for any set $H' \subseteq B \cup W$ such that $G[H']$ is isomorphic to $H,$ we have $|H' \cap W| \le 1.$
This crucial property will help us to control how copies of $H$ can be packed in $G[W \cup B],$ and thus motivates the invariants required in the following (informal) definition of partial decomposition. 

\subsection{Invariants in a partial decomposition}
\label{@maravillados}
We say that $(W,B,C)$ is a \emph{partial decomposition} (with respect to $(C^0,W^0)$) of $G$ if
\begin{enumerate}
\item\label{@ausdrucksform} $W,B,C$ is a partition of $V(G)$
\item\label{@necesariamente} $(W,B)$ is a \emph{nice pair}: $W \subseteq W^0$ and for any  $H' \subseteq B \cup W$ such that $G[H']$ is isomorphic to $H,$ $|H' \cap W| \le 1$
\item\label{@thoroughness} (size invariant) $B$ is small (typically $|B| = \mathcal{O}(|C^<|),$ where $C^< = C^0 \setminus C$
\end{enumerate}
{As $C^0$ and $W^0$ will be fixed (there are only computed once at the begining), we voluntarily do not mention ``with respect to $(C^0,W^0)$'' in the reminder of the article.} 
Observe that the tuple $(W,B,C)$ that we obtained at this end of Section~\ref{se_ssssssssdfgdfgsssdfghfdhfdhoverview} is a partial decomposition.
Properties \ref{@ausdrucksform} and \ref{@necesariamente} listed above are common to both problems, whereas the notion of size to ensure that $B$ is small is ad-hoc (but with
the same objective to guarantee that $|B|+|C| = \mathcal{O}(|C^0|),$ as $B$ and $C$ will be part of the output of the kernel).
The property of being a nice pair will allow us to provide a structural description of $B$: in both problems we will partition $B$ into \emph{buckets} $B_i,$ and obtain
a description on how copies of $H$ in $W \cup B$ intersect the buckets.

\subsection{Description of one round of the kernel}
\label{@gillyflowers}
We now consider an arbitrary round of the kernel, where we have our current partial decomposition $(W,B,C).$
The objective is now to apply again Corollary~\ref{lemma_rainbowvc}, but in a more general setting than in Section~\ref{se_ssssssssdfgdfgsssdfghfdhfdhoverview} where we had $B=\emptyset.$

Let us now discuss how to define the auxiliary graph $(G,χ)\langle W,B,C\rangle.$ We start as before by defining $V((G,χ)\langle W,B,C\rangle) = W,$ and for any $c \in C,$ and any $\{u,v\} \subseteq W$ such that $G[\{c,u,v\}]$ is isomorphic to $H,$ adding to $(G,χ)\langle W,B,C\rangle$ edge $e=\{u,v\}$ and set $\x(e)=c.$
{Again, notice that according to the context, $C$ may denote a subset of $G$, of a subset of our colors in $(G,χ)\langle W,B,C\rangle$.}
The crux of this approach is {to add to $(G,χ)\langle W,B,C\rangle$ some loops in $W$ and a coloring of these loops (using a fresh set of colors $D,$ which is disjoint from $C$), such that:}
\begin{enumerate}[label=\roman*.]
\item\label{@professional} $|D| = \mathcal{O}(|B|)$
\item\label{@argamasillesco} if there exists a rainbow matching $M^D$ only for loops of colors in $D$ (meaning that $M^D = \{v_d, d \in D\}$ and $|M^D|=|D|$), then for any packing $\mathcal{P}$
  of $G[W \cup B],$ there exists a packing $\mathcal{P'}$ such that $|\mathcal{P'}|=|\mathcal{P}|$ and $V(\mathcal{P'}) \subseteq B \cup V(M^D)$
\item\label{@reinterpreten} if there exists a subset of colors $X^D \subseteq D$ and a vertex cover $T$ of $(G,χ)\langle W,B,C\rangle[X^D]$ such that $|T| \le 2(4+ε)|X^D|,$ then we can find in polynomial time a non-empty subset $T' \subseteq T$ such that we can add $T'$ to $B$ without violating
  the size invariant of a partial decomposition. In other words, in this case we define $W'=W \setminus T',$ $B'=B \cup T',$ $C'=C,$ and we want that
  $(W',B',C')$ remains a partial decomposition.
\end{enumerate}
{The way we can define such colored loops to obtain the three previous properties depends on the problem that we consider and thus will be specified in each application of the technique.}
As we guess that, at a first sight, Property~\ref{@reinterpreten} may look like it comes out of nowhere, let us now explain why properties \ref{@professional} to \ref{@reinterpreten} are sufficient to obtain the kernel.
Suppose that we defined an auxiliary graph $(G,χ)\langle W,B,C\rangle$ satisfying the previous properties, and that we apply Corollary~\ref{cor_rainbowvc} on $(G,χ)\langle W,B,C\rangle.$
Let us discuss again the two possible outcomes of this Corollary.

\paragraph*{Case where a rainbow matching exists.} Consider now the case  where there is a rainbow matching $M$ in $(G,χ)\langle W,B,C\rangle.$ In this case, we stop and return $G[V(M) \cup B \cup C].$
Let us partition $M=M^C \cup M^D,$ where $M^C$ (resp. $M^D$) are edges whose color is in $C$ (resp. $D$).
Observe first that the kernel output is small, as $|V(M)| = |D|+2|C|=\mathcal{O}(|B|+|C|)$ (by Property~\ref{@professional}), and thus the output has vertex size
$\mathcal{O}(|B|+|C|)=\mathcal{O}(|C^<|+|C|)=\mathcal{O}(|C^0|)=\mathcal{O}(k)$ (by Property~\ref{@thoroughness}).
Let us now examine why it is safe. 
Consider an $H$-packing $\mathcal{P}$ in $G.$ We restructure $\mathcal{P}$ as follows (see Figure~\ref{asdfsdfdsfdsfdsfdsf_wbc}). For any $H_i \in \mathcal{P}$ we define
a corresponding $H'_i$ as follows (and we define $\mathcal{P}'=\{H'_i \mid H_i \in \mathcal{P}\}$):
\begin{itemize}
\item If $H_i \subseteq C,$ define $H'_i=H_i.$
\item If $H_i \nsubseteq C,$ and $H_i \cap C \neq \emptyset,$ choose $c \in H_i \cap C,$ and define $e_c=\{u_c,v_c\}$ the edge of $M$ such that $\x(e_c)=c.$ Restructure $H_i$ into $H'_i = \{c,u_c,v_c\}.$
\item Now, it only remains to restructure $\mathcal{P}_{(W\cup B)}=\{H_i \in \mathcal{P} \mid H_i \subseteq W \cup B\}.$
  By Property~\ref{@argamasillesco}, there exists a packing $\mathcal{P'}_{(W\cup B)}$ such that $|\mathcal{P'}_{(W \cup B)}|=|\mathcal{P}_{(W \cup B)}|$ and $V(\mathcal{P'}_{(W \cup B)}) \subseteq B \cup V(M^D).$
\end{itemize}

\begin{figure}[!ht]
\centering
\scalebox{0.6}{\includegraphics{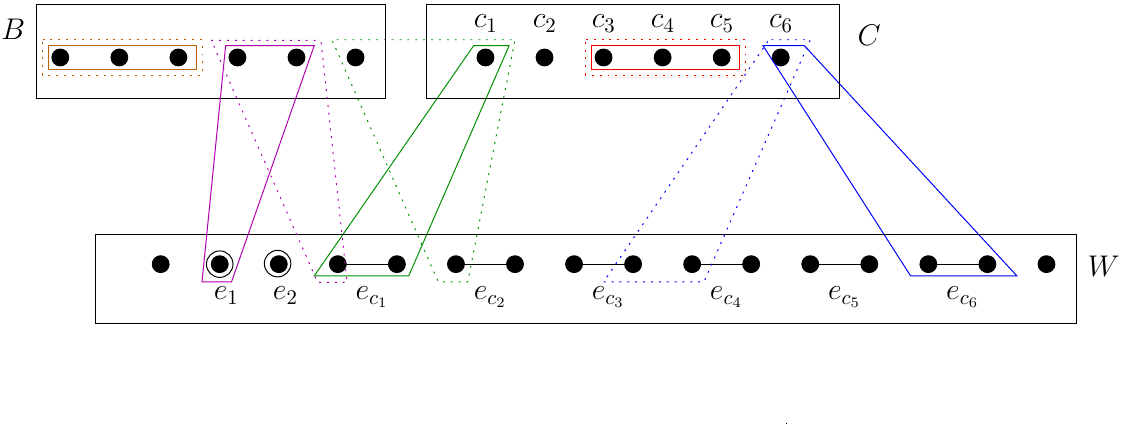}}
\caption{
  {The color set of $(G,χ)\langle W,B,C\rangle$ is $C \cup D$, where here $C=\{c_1,\dots,c_6\}$, and $D=\{d_1,d_2\}$ (where $D\cap C = \emptyset$).
   In this example we assume that there exists a rainbow matching $M$ in $(G,χ)\langle W,B,C\rangle$. We have $M=M^C \cup M^D,$
  where $M^C=\{e_{c_i},1 \le i \le 6\}$ and $M^D=\{e_1,e_2\}.$ 
  Here, we restructure $H'_1,H'_2,H'_3,H'_4,H'_5$ (in dashed lines) to $H_1,H_2,H_3,H_4,H_5$ (in plain lines) according to $M$.}}
\label{asdfsdfdsfdsfdsfdsf_wbc}
\end{figure}

We can observe that the  $\{H'_i\}$ are vertex disjoint as $M$ is a matching, and thus that $V(M)$ is safe.
Observe that the matching was computed without considering the complex (and unknown) structure of ``bad'' copies of $H$ that use vertices of both $B$ and $C.$
In other words, the way we organize the previous restructuration allows us to forget these bad copies in $G.$

\paragraph*{Case where a small vertex cover exists.} Suppose there
exists a non-empty subset of colors $X \subseteq C \cup D$ and a
vertex cover of $T$ of $(G,χ)\langle W,B,C\rangle[X]$ such that
$|T| \le (4+ε)|X|.$ Let $X^C = X \cap C$ and $X^D = X \cap D.$. Notice
first that we cannot simply proceed as in
Section~\ref{se_ssssssssdfgdfgsssdfghfdhfdhoverview} (where $X^D$ was
empty) and define $W'= W \setminus T,$ $C'=C \setminus X^C$ and
$B'=B \cup (T \cup X^C).$ Indeed, if we want the size
property~\ref{@thoroughness} to be respected, we need that what we add
to $B$ is linear in what we remove from $C,$ or more formally that
$|T \cup X^C| = \mathcal{O}(|X^C|).$ However, we only know that
$|T| \le (4+ε)|X|=(4+ε)(|X^D|+|X^C|),$ and thus if $|X^C|$ is small
compared to $|X^D|$ (typically $|X^C|=0$), we don't have the property
we need.  This is where we use (for both problems) the following
win/win trick.

Case 1: if $|X^D| \le |X^C|.$ In this case, the previous inequality
gives us $|T| \le (4+ε)|X|=2(4+ε)|X^C|,$ and we can define
$W'= W \setminus T,$ $C'=C \setminus X^C$ and $B'=B \cup (T \cup X^C)$
while preserving size Property~\ref{@thoroughness}. {Notice that the
  crucial property ensuring that $(W',B')$ is still a nice pair (which
  is required to obtain that $(W',B',C')$ remains a partial
  decomposition) is that there is no $H'$ such that $G[H']$ is
  isomorphic to $H$ where $|H' \cap W'|=2$ and $|H' \cap X^c|=1,$
  which holds because $T$ is a vertex cover of
  $(G,χ)\langle W,B,C\rangle[X]$.}

Case 2: if $|X^D| > |X^C|.$ In this case, the previous inequality
gives us $|T| \le 2(4+ε)|X^D|,$ and we use
Property~\ref{@reinterpreten} to find in polynomial time a non-empty
subset $T' \subseteq T$ such that $(W',B',C')$ remains a partial
decomposition, where $W'=W \setminus T',$ $B'=B \cup T',$ $C'=C.$
Thus, Case 2 corresponds to a case where we discover that
  a certain part $T' \subseteq W$ is small, and can be added to the
  buckets.

Notice that in both Cases 1 and 2, we obtain a new ``smaller'' partial
decomposition where either $|C'|<|C|$ (because $X^{C}$ is non-empty,
in Case 1) or $|W'|<|W|$ (in Case 2). Thus the algorithm terminates.

\subsection{Applicability of the rainbow coloring technique}
The rainbow matching technique consists in applying the algorithm of
Section~\ref{se_ssssssssdfgdfgdfghfdhfdhoverview}.  Adapting this
technique to a particular problem requires to find a way to define
colors for buckets that respects
Properties~\ref{@professional},~\ref{@argamasillesco},
and~\ref{@reinterpreten}, together with a notion of size used in
Property~\ref{@thoroughness}.  In \autoref{@construction} and
\autoref{@industriousness}, we present how this technique can be
applied for \IPP and \TPT respectively.

\section{Linear kernel for  \IPP}
\label{@construction}

\subsection{Notation}

Given a graph $G$ we refer to a path in $G$ of length 2 as a
\emph{2-path} of $G.$ We say that a 2-path of $G$ is \emph{induced} if
there is no edge in the graph between its endpoints. We call a $P_3$
an induced 2-path. When we refer to a $P_3$ in a graph, we will see it
as a subset of vertices rather than an induced subgraph.  An
\emph{induced $P_{3}$-packing} of $G$ is a set
$\P = \{V_i, i \in [x]\}$ of vertex-disjoint induced $P_{3}$'s.

\defparproblem{{\sc Induced 2-path-Packing} (\IPP)}{$(G,k)$ where
  $G$ is a graph and $k \in \N$}{$k.$}{Is there an induced
  $P_3$-packing of size at least $k$?}

In this section we prove the following theorem.

\begin{theorem}
\label{secondi_ess}
There exists some function $f:\mathbb{N}\to\mathbb{N}$ such that for
every $ε>0,$ there exists an algorithm that, given an instance $(G,k)$
of \IPP\ outputs a set $A\subseteq V(G)$ such that $(G[A],k)$ is an
equivalent instance where $|A|≤(243+ε)k.$  Moreover this algorithm
runs in time $\O(|V(G)|^{f(ε)}).$
 In other words,  \IPP admits a kernel of a linear number of vertices.
\end{theorem}

\subsection{Preliminary phase: greedy localization}
\label{@challengingly}
Given an input $(G,k)$ of \IPP, we start by a greedy localization
phase, that is, by finding, in polynomial time, an inclusion-wise
maximal induced 2-path-packing of $G.$ Let $C^0$ be the set of the
vertices in the paths of such a packing.  If $|\P^0| \ge k$ we can
directly answer that $(G,k)$ is a \yes-instance, therefore we may
suppose that $|P^0| < k,$ implying that $|C^0| \leq 3(k-1).$ The
vertices of $C^0$ will be referred as {(distinct)} \emph{colors}.
Observe that as the considered induced 2-path-packing is
inclusion-wise maximal, the graph $W^0:=G \setminus C^0$ does not
contain $P_3$ as an induced subgraph, consequently it is the disjoint
union of a set, say
$\mathcal{W}^{0}=\{W^{0}_1,\ldots,W^{0}_{i_{0}}\},$ such that, for
$i\in[i_0],$ the graph $G[W_{i}^{0}]$ is a clique. Clearly,
$V(G) = C^0 \cup V(\mathcal{W}^0)$ and, for any
$\{i,i'\}\in{[i_{0}]\choose 2},$ there is no edge between $W^0_i$ and
$W^0_{i'}.$ Such a couple $(C^0,\mathcal{W}^0)$ will be referred as a
\emph{greedy localized pair for the input $(G,k)$}.

This completes the initialization phase, and the tuple
$(C^0,\mathcal{W}^0)$ will be given as initial input of our
kernelization algorithm.

\subsection{Nice pair, buckets, partial decomposition, and auxiliary graph}

In all this section we consider that we are given an input $(G,k)$ of
\IPP, and a greedy localized pair $(C^0,\mathcal{W}^0)$ for input
$(G,k).$ Recall that
$\mathcal{W}^{0}=\{W^{0}_1,\ldots,W^{0}_{i_{0}}\}$ and
$W^0=V(\mathcal{W}^0).$

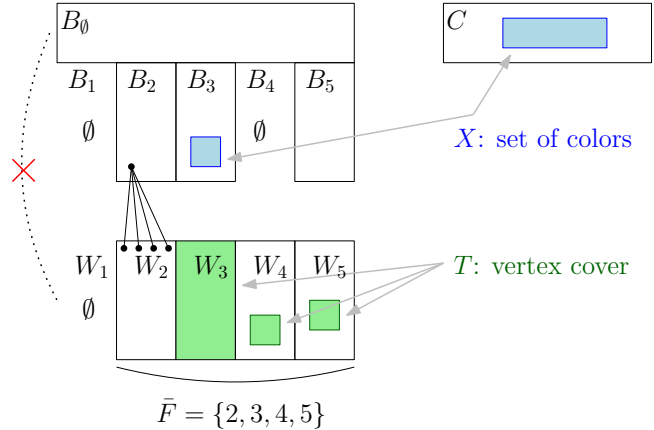
\begin{figure}[htb]
\begin{center} 
\scalebox{.7}{\tikzstyle{ipe stylesheet} = [
  ipe import,
  even odd rule,
  line join=round,
  line cap=butt,
  ipe pen normal/.style={line width=0.4},
  ipe pen heavier/.style={line width=0.8},
  ipe pen fat/.style={line width=1.2},
  ipe pen ultrafat/.style={line width=2},
  ipe pen normal,
  ipe mark normal/.style={ipe mark scale=3},
  ipe mark large/.style={ipe mark scale=5},
  ipe mark small/.style={ipe mark scale=2},
  ipe mark tiny/.style={ipe mark scale=1.1},
  ipe mark normal,
  /pgf/arrow keys/.cd,
  ipe arrow normal/.style={scale=7},
  ipe arrow large/.style={scale=10},
  ipe arrow small/.style={scale=5},
  ipe arrow tiny/.style={scale=3},
  ipe arrow normal,
  /tikz/.cd,
  ipe arrows, 
  <->/.tip = ipe normal,
  ipe dash normal/.style={dash pattern=},
  ipe dash dotted/.style={dash pattern=on 1bp off 3bp},
  ipe dash dashed/.style={dash pattern=on 4bp off 4bp},
  ipe dash dash dotted/.style={dash pattern=on 4bp off 2bp on 1bp off 2bp},
  ipe dash dash dot dotted/.style={dash pattern=on 4bp off 2bp on 1bp off 2bp on 1bp off 2bp},
  ipe dash normal,
  ipe node/.append style={font=\normalsize},
  ipe stretch normal/.style={ipe node stretch=1},
  ipe stretch normal,
  ipe opacity 10/.style={opacity=0.1},
  ipe opacity 30/.style={opacity=0.3},
  ipe opacity 50/.style={opacity=0.5},
  ipe opacity 75/.style={opacity=0.75},
  ipe opacity opaque/.style={opacity=1},
  ipe opacity opaque,
]
\definecolor{red}{rgb}{1,0,0}
\definecolor{blue}{rgb}{0,0,1}
\definecolor{green}{rgb}{0,1,0}
\definecolor{yellow}{rgb}{1,1,0}
\definecolor{orange}{rgb}{1,0.647,0}
\definecolor{gold}{rgb}{1,0.843,0}
\definecolor{purple}{rgb}{0.627,0.125,0.941}
\definecolor{gray}{rgb}{0.745,0.745,0.745}
\definecolor{brown}{rgb}{0.647,0.165,0.165}
\definecolor{navy}{rgb}{0,0,0.502}
\definecolor{pink}{rgb}{1,0.753,0.796}
\definecolor{seagreen}{rgb}{0.18,0.545,0.341}
\definecolor{turquoise}{rgb}{0.251,0.878,0.816}
\definecolor{violet}{rgb}{0.933,0.51,0.933}
\definecolor{darkblue}{rgb}{0,0,0.545}
\definecolor{darkcyan}{rgb}{0,0.545,0.545}
\definecolor{darkgray}{rgb}{0.663,0.663,0.663}
\definecolor{darkgreen}{rgb}{0,0.392,0}
\definecolor{darkmagenta}{rgb}{0.545,0,0.545}
\definecolor{darkorange}{rgb}{1,0.549,0}
\definecolor{darkred}{rgb}{0.545,0,0}
\definecolor{lightblue}{rgb}{0.678,0.847,0.902}
\definecolor{lightcyan}{rgb}{0.878,1,1}
\definecolor{lightgray}{rgb}{0.827,0.827,0.827}
\definecolor{lightgreen}{rgb}{0.565,0.933,0.565}
\definecolor{lightyellow}{rgb}{1,1,0.878}
\definecolor{black}{rgb}{0,0,0}
\definecolor{white}{rgb}{1,1,1}
\begin{tikzpicture}[ipe stylesheet]
  \draw
    (64, 480) rectangle (224, 448);
  \draw
    (96, 448) rectangle (128, 384);
  \draw
    (128, 448) rectangle (160, 384);
  \draw
    (192, 448) rectangle (224, 384);
  \draw
    (272, 480) rectangle (384, 448);
  \node[ipe node, font=\Large]
     at (69.586, 434.414) {$B_{1}$};
  \node[ipe node, font=\Large]
     at (101.586, 434.414) {$B_2$};
  \node[ipe node, font=\Large]
     at (133.586, 434.414) {$B_{3}$};
  \node[ipe node, font=\Large]
     at (165.586, 434.414) {$B_{4}$};
  \node[ipe node, font=\Large]
     at (197.586, 434.414) {$B_{5}$};
  \draw
    (96, 352) rectangle (128, 288);
  \filldraw[fill=lightgreen]
    (128, 352) rectangle (160, 288);
  \draw
    (160, 352) rectangle (192, 288);
  \draw
    (192, 352) rectangle (224, 288);
  \node[ipe node, font=\Large]
     at (77.586, 310.414) {$\emptyset$};
  \node[ipe node, font=\Large]
     at (77.586, 406.414) {$\emptyset$};
  \node[ipe node, font=\Large]
     at (169.586, 406.414) {$\emptyset$};
  \filldraw[draw=blue, fill=lightblue]
    (136, 408) rectangle (152, 392);
  \filldraw[draw=darkgreen, fill=lightgreen]
    (168, 312) rectangle (184, 296);
  \filldraw[draw=darkgreen, fill=lightgreen]
    (200, 320) rectangle (216, 304);
  \filldraw[draw=blue, fill=lightblue]
    (304, 472) rectangle (360, 456);
  \node[ipe node, font=\Large]
     at (273.586, 466.414) {$C$};
  \node[ipe node, font=\Large]
     at (65.586, 466.414) {$B_{\emptyset}$};
  \draw
    (96, 284)
     .. controls (138.6667, 273.3333) and (181.3333, 273.3333) .. (224, 284);
  \node[ipe node, font=\Large]
     at (117.586, 254.414) {$\bar{F}=\{2,3,4,5\}$};
  \draw[ipe pen heavier, ipe dash dotted]
    (60, 464)
     .. controls (38.6667, 413.3333) and (40, 365.3333) .. (64, 320);
  \draw[red, ipe pen heavier]
    (40, 396)
     -- (52, 384);
  \draw[red, ipe pen heavier]
    (40, 384)
     -- (52, 396);
  \node[ipe node, font=\Large, text=blue]
     at (277.586, 402.414) {$X$: set of colors};
  \node[ipe node, font=\Large, text=darkgreen]
     at (277.586, 334.414) {$T$: vertex cover};
  \draw[gray, ipe pen heavier, ->]
    (288, 420)
     -- (308, 452);
  \draw[gray, ipe pen heavier, ->]
    (288, 420)
     -- (156, 396);
  \draw[gray, ipe pen heavier, ->]
    (272, 340)
     -- (164, 328);
  \draw[gray, ipe pen heavier, ->]
    (272, 340)
     -- (188, 308);
  \draw[gray, ipe pen heavier, ->]
    (272, 340)
     -- (220, 312);
  \pic
     at (104, 392) {ipe disk};
  \draw
    (104, 392)
     -- (100, 348);
  \draw
    (104, 392)
     -- (108, 348);
  \draw
    (104, 392)
     -- (116, 348);
  \draw
    (104, 392)
     -- (124, 348);
  \node[ipe node, font=\Large]
     at (73.586, 334.414) {$W_{1}$};
  \node[ipe node, font=\Large]
     at (105.586, 334.414) {$W_{2}$};
  \node[ipe node, font=\Large]
     at (137.586, 334.414) {$W_{3}$};
  \node[ipe node, font=\Large]
     at (169.586, 334.414) {$W_{4}$};
  \node[ipe node, font=\Large]
     at (201.586, 334.414) {$W_{5}$};
  \pic
     at (100, 348) {ipe disk};
  \pic
     at (108, 348) {ipe disk};
  \pic
     at (116, 348) {ipe disk};
  \pic
     at (124, 348) {ipe disk};
\end{tikzpicture}}
\end{center} 
\caption{Example of a bucket decomposition and a partial
  decomposition. The vertex depicted in $B_2$ is adjacent to all vertices of $W_2$, and not adjacent to any other vertex
  in $W$. We depict the case where the kernelization algorithm
  does  not find a rainbow matching in $(G,χ)\langle W,B,C\rangle$
  and therefore it finds a subset of colors $X$ (depicted in blue)
  along with a small set of vertices $T$ (depicted in green) such
  that $T$ is a vertex cover of $(G,χ)\langle W,B,C\rangle[X].$}
\label{fig_partial_decompo}
\end{figure} 

Given two disjoint subsets $W$ and $B$ of $V(G),$ we say that the pair
$(W,B)$ is a \emph{nice pair} of $G$ if $W \subseteq
W^0$ and if every induced 2-path of $B\cup W$ contains at most one
vertex in $W.$\\

\noindent{\bf Bucket decompositions.}  Given
two disjoints subsets $W$ and $B$ of $V(G),$ a \emph{bucket
decomposition} (see Figure~\ref{fig_partial_decompo}) of the pair
$(W,B)$ is defined by the following three partitions:
\begin{itemize}
\item a partition $\{B_{\neq \emptyset},B_\emptyset\}$ of $B$
\item a partition $\{B_{1},\ldots,B_{i_0}\}$ of $B_{\neq\emptyset},$
  (we call the sets of this partition \emph{buckets}) and
\item a partition $\{W_{1},\ldots,W_{i_0}\}$ of $W.$
\end{itemize}
 (recall that we allow empty sets in partitions) 

such that
    \begin{enumerate}
    \item\label{@consequential} for $i\in[i_0]$\ $W_i = W \cap W^0_i$
  \item\label{@dispenserait} $B_\emptyset = \{v \in B \mid N_G(v) \cap W = \emptyset\}$
    and $B_{\neq \emptyset} = \{v \in B \mid N_G(v) \cap W \neq \emptyset\}$
  \item\label{@standardisation} for any $i \in [i_0]$ and $v \in B_i,$ $N_G(v) \cap
    W = W_i$ (the neighborhood in $W$ of any vertex of a bucket
    $B_{i}$ is exactly the vertex set of its ``corresponding'' clique
    $G[W_{i}]$).
  \end{enumerate}

We denote $F = \{i \in [i_0] \mid W_i = \emptyset\},$ $\bar{F} = [i_0]
\setminus F.$  Observe that if $W_i = \emptyset,$ then $B_i =
\emptyset,$ as vertices in $B_i$ belong to $B_{\neq \emptyset}$ and
should have $N(v) \cap W \neq \emptyset.$ The contrapositive is not
true as we may have $W_i \neq \emptyset$ and $B_i = \emptyset.$

\begin{lemma}\label{lemma_nicepair}
  Given $W$ and $B$ disjoint subsets of $G,$
    \begin{enumerate}
    \item\label{@licenciadillo} $(W,B)$ is a nice pair iff it admits a bucket
      decomposition.
    \item\label{@respondiendo} if $(W,B)$ is nice pair, then for any induced
      2-path $P$ in $G$ such that $P \subseteq W \cup B,$ $P \cap W
      \neq \emptyset,$ there exists a unique $i \in [i_0]$ and vertices
      $u,v,w$ such that $P = \{u,v,w\},$ $u \in B,$ $v \in B_i$ and $w
      \in W_i.$  Informally, $P$ must have exactly one vertex in {one
        of the} remaining cliques, one in its corresponding bucket,
      and the last one anywhere in $B.$
  \end{enumerate}
 \end{lemma}

\begin{proof} We start with the proof of the first item.\\
($\Rightarrow$) Suppose $(W,B)$ is a nice pair.  Consider a vertex $v
  \in B_{\neq \emptyset}.$ There exists $i \in [i_0]$ such that $N(v)
  \cap W_i \neq \emptyset.$  Let $u \in N(v) \cap W_i.$  If there
  existed $w \in W_i$ such that $w \notin N(v),$ then, as $W_i$ is a
  clique, we would have that $P=\{v,u,w\}$ is an induced 2-path such
  that $|P \cap W| \ge 2,$ a contradiction. This implies that $N(v)
  \cap W \supseteq W_i.$ Let us now prove that {$N(v) \cap W =
    W_i$}. Suppose, towards a contradiction, that there exists $j \neq
  i$ and $u' \in N(v) \cap W_j.$ Then, as there is no edges between
  $W_j$ and $W_i,$ we obtain that $P=\{v,u,u'\}$ is an induced 2-path
  such that $|P \cap W| \ge 2,$ a contradiction. Thus, we obtain the
  partition $\{B_{1},\ldots,B_{i_0}\}$ of $B_\emptyset,$ where $B_i =
  \{i \in B_{\neq \emptyset} \mid N(v) = W_i\}$ (some of the $B_i$'s
  may be empty).

\noindent ($\Leftarrow$) Let $(W,B)$ be a pair admiting a bucket
decomposition.  Let us prove the following structural property: for
any $P$ induced 2-path of $G[W \cup B]$ such that $P \cap W \neq
\emptyset,$ there exists a unique $i \in i_0$ and vertices $u,v,w$
such that $P = \{u,v,w\},$ $u \in B,$ $v \in B_i,$ and $w \in W_i.$
We cannot have $P \subseteq W$ as $W$ is a union of cliques (as $W_i
\subseteq W_i^0$), implying that $P \cap B \neq \emptyset.$  This
implies that there exists an edge $\{v,w\} \subseteq P$ such that $v
\in B$ and $w \in W.$  By Property \ref{@dispenserait}, we know that $v \notin
B_\emptyset,$ implying that there exists $i \in [i_0]$ such that $v
\in B_i.$  Property \ref{@standardisation} implies that $w \in W_i.$  Let us now
consider the third vertex $u$ of $P.$ Firstly, $u$ cannot be in $W_i$
as, by Property \ref{@standardisation}, this would imply that $\{v,u\}$ is an edge,
and thus $P$ would be a triangle.  Secondly, $u$ cannot be in a $W_j,$
where $j \neq i,$ as by Property \ref{@standardisation}, $N(v) \cap W = W_i.$  This
implies that $u \in B$ and concludes the proof of the structural
property. The fact that $(W,B)$ is a nice pair is now immediate.

The second item immediately follows. Indeed, by the first part of the
result, any nice pair $(W,B)$ admits a bucket decomposition, which
implies the structural property as seen previously.
\end{proof}

Given a nice pair $(W,B),$ we will refer to $B_\emptyset,$ $B_{\neq
  \emptyset},$ $F$ and $\bar{F}$ as defined in the bucket
decomposition.  Informally, $F$ will denote the set of indexes of
cliques that will survive during the course of the kernelization
algorithm.  Recall that, using our notations, $B_{\bar{F}} =
\bigcup_{i \in \bar{F}}B_i$ and observe that $B_{\bar{F}}=B_{\neq
  \emptyset}.$

We fix some $ε>0$ and set $c_{1}=4+ε$ ($c_{1}$ is suited so to permit
the application of Corollary~\ref{cor_rainbowvc}).  Given a nice pair
$(W,B),$ we define the \emph{size} of $(W,B)$
as $$s(W,B)=\frac{|B_\emptyset|}{1+2c_1}+|B_{\neq \emptyset}|.$$

\begin{definition}\label{def:partial_dec}
  We say that a tuple $(W,B,C)$ is a \emph{partial decomposition} iff:\medskip

\noindent \textbf\emph{Partition requirements:} \\Let $C^< = C^0
\setminus C$ (which we will see as the colors already treated by
previous applications of the rule of the kernelization algorithm)
  \begin{enumerate}
  \item\label{@interwovenness} there is a partition $V(G)=W \cup B \cup C,$
  \item\label{@universalisation} $C \subseteq C^0,$ and 
  \item\label{@hypochondriacal} $(W,B)$ is a nice pair. 
  \end{enumerate}
  
 \noindent \textbf\emph{Size requirement:}
  \begin{enumerate}
  \item\label{@conflictivitat} $s(W,B) \le (1+2c_1)|C^<|.$
  \end{enumerate}
 \stf{ Moreover, we will say that the partial decomposition is \emph{clean} if it satisfies the following extra condition:
  \begin{enumerate}
  \item\label{@cleaness} for every vertex $c\in C$ the graph $G[W\cup \{c\}]$ contains an induced 2-path.
  \end{enumerate} }
\end{definition}

\stf{The initial partial decomposition we will consider (which will be
$(W_0,\emptyset,C_0)$) as well as partial decompositions which will be
produced by our kernelization procedure will not necessarily be clean
partial decompositions. However the simple following lemma allows to
clean a partial decomposition.
}

\stf{
\begin{lemma}[Cleaning Lemma]
  \label{lem:clean-I2PP}
  Let $(W,B,C)$ be a partial decomposition and $X$ be set the of
  vertices $x$ of $C$ such that $G[W\cup \{x\}]$ does not contain any
  induced 2-path. Then $(W,B\cup X,C\setminus X)$ is a clean partial
  decomposition.
\end{lemma}
}

\stf{
\begin{proof}
  First let us check that $(W,B\cup X,C\setminus X)$ satisfies the
  requirements of a partial
  decomposition. Properties~\ref{@interwovenness}
  and~\ref{@universalisation} are clearly satisfy, and by choice of
  $X$, no vertex of $X$ is contained in an induced 2-path with two
  vertices of $W$, so $(W,B\cup X)$ is also a nice pair and
  property~\ref{@hypochondriacal} is satisfies. To conlude that
  $(W,B\cup X,C\setminus X)$ is a partial decomposition, let us check
  it fulfills the size requirement. Indeed each time one vertex of $X$
  is added to $B$, $s(W,B)$ increases by at most 1, whereas
  $(1+2c_1)|C^<|$ increases by $1+2c_1$. So, the size
  requirement~\ref{@conflictivitat} is still valid. By repeating this
  counting argument for every vertex of $X$, we can conclude that
  $(W,B\cup X,C\setminus X)$ is a partial decomposition. Moreover, it
  is clear that every vertex of $C\setminus X$ forms an induced 2-path
  with two vertices of $W$ and that $(W,B\cup X,C\setminus X)$ is clean.
\end{proof}
}

\begin{definition}[Auxiliary graph]
Let $(W,B,C)$ be a partial decomposition.
  Let $p=|C|+\sum_{i\in [i_0]}|B_{i}|.$  We define the $p$-edge-colored multigraph
$(\tilde{G},χ)\langle W,B,C\rangle$ where the vertex set of
$\tilde{G}$ is $W$ and the edges of $\tilde{G},$ as well as their
colors, are defined as follows.  For any $i \in \bar{F}$ and any $u
\in B_i$ and for any $v \in W_i,$ we add an edge $e=\{v\}$ and we set
$χ(e)=u.$  {Moreover, for any $u \in C$ and for any $v$ and $w$ in $W$
  such that $\{u,v,w\}$ is an induced 2-path, we add an edge
  $e=\{v,w\}$ and we set $χ(e)=u.$}
\end{definition}

Observe that an edge $\{v,w\}$ 
of $(\tilde{G},χ)\langle W,B,C\rangle$ can be either inside a $W_i$
(in this case $\{v,w\}$ is also an edge in $G$ as $G[W_i]$ is a
clique) or between $W_i$ and $W_{i'}$ for $i \neq i'$ (in this case
$\{v,w\}$ is a non-edge in $G$).  Notice that if $W = \emptyset,$ then
$(\tilde{G},χ)\langle W,B,C\rangle$ is the empty graph and we consider
that it admits a rainbow matching $M = \emptyset.$  Notice also that
$C$ or $B_{\neq \emptyset}$ may be empty. If $C=B_{\neq
  \emptyset}=\emptyset$ then $(\tilde{G},χ)\langle W,B,C\rangle$ has
no edges and we consider that it admits a rainbow matching $M =
\emptyset$ (but this case will not occur), however cases where one of
the two sets $C,$ $B_{\neq \emptyset}$ is empty will occur in the
kernel.



\subsection{Analysis of the two cases: rainbow matching or small vertex cover}


First, let us look at the case where the reduction rule produces a
rainbow matching.

\begin{lemma}[Case of rainbow matching]\label{lemma_safe_rainbow}
Let $(W,B,C)$ be a \stf{clean} partial decomposition. Suppose that the colored multigraph
$(\tilde{G},χ)\langle W,B,C\rangle$ admits a rainbow matching $M$.
 Let $A=V(M) \cup B \cup C,$ and $G'=G[A].$ Then,
  \begin{enumerate}
  \item $(G,k)$ and $(G',k)$ are equivalent instances of \IPP,
  \item $|A| \le 3(1+2c_1)^2k.$
  \end{enumerate}
\end{lemma}
  
\begin{proof} {\bf Equivalence property.} First, if $W = \emptyset,$
  then $(\tilde{G},χ)\langle W,B,C\rangle$ is the empty graph then
  $M=\emptyset$ and $G'=G,$ implying the equivalence.  Let us now
  assume that $W \neq \emptyset.$ Let us assume that $(G,k)$ is
  \yes-instance and prove that $(G',k)$ is as well (the other
  direction is straightforward as $G'$ is an induced subgraph of $G$).
  Let $\P^* = \{P^*_i, i \in [k]\}$ be an induced $P_3$-packing in $G$
  of size $k.$
  
Notice that the only vertices of $G$ not belonging to $G'$ are in $W.$  Let us
partition $\P^* = \P^*_1 \cup \P^*_W,$ where $\P^*_W = \{P \in \P^*
\mid P \cap W \neq \emptyset\}.$  Our objective is to restructure
$\P^*$ into another induced 2-path-packing $\P = \P^*_1 \cup \P_W$
such that $|\P_W| = |\P^*_W|.$  To that end, we will associate to each
$P^*_i \in \P^*_W$ a set $P_i \in \P_W$ such that \smallskip
  \begin{enumerate}
  \item\label{@dishonourable} for any $P_i \in \P_W,$ $P_i$ is an induced
    2-path in $G'$
  \item\label{@prevaleciera} for any $P_i, P_j$ in $\P_W,$ $P_i \cap P_j =
    \emptyset$ (paths in ${\cal P}_W$ are vertex-disjoint)
  \item\label{@gentillesses} for any $i,$ $P_i \cap (V(G)\setminus W)
    \subseteq P^*_i \cap (V(G)\setminus W)$ (roughly speaking outside
    $W,$ path $P^*_i$ uses more vertices than
    $P_i$)
  \end{enumerate}
  
  Observe that Properties~\ref{@prevaleciera} and~\ref{@gentillesses} imply
  that paths in $\P$ are vertex-disjoint. In particular, if, towards a
  contradiction, a path $P \in \P^*_1$ intersected a path $P_i \in
  \P_W,$ then, by Property~\ref{@gentillesses}, we would also have $P \cap
  P^*_i \neq \emptyset,$ which is a contradiction.  Let us now
  define the $P_i$'s.

  Let us partition $\P^*_W = \P^*_{CW} \cup \P^*_{\bar{C}W},$ where
  $\P^*_{CW} = \{P \in \P^*_W \mid P \cap C \neq \emptyset\}$ and
  $\P^*_{\bar{C}W} = \{P \in \P^*_W \mid P \cap C = \emptyset\}$.  As
  paths in $\P^*_{\bar{C}W}$ use a vertex in $W,$ and no vertex in
  $C$, it intersects $B$ and by Property~\ref{@respondiendo} of
  Lemma~\ref{lemma_nicepair}, for any $P^*_i \in \P^*_{\bar{C}W}$
  there exists a unique $l_i$ and vertices $u_i, v_i, w_i$ such that
  $u_i \in B,$ $v_i \in B_{l_i},$ and $w_i \in W_{l_i}.$ Thus, we can
  partition
  $\P^*_{\bar{C}W} = \bigcup_{l \in \bar{F}} \P^{(*,l)}_{\bar{C}W},$
  where
  $\P^{(*,l)}_{\bar{C}W} = \{P^*_i \in \P^*_{\bar{C}W} \mid l_i=l\}.$
  Observe that
  \begin{itemize}
  \item any $P^*_i \in \P^{(*,l)}_{\bar{C}W}$ uses exactly one vertex
    ($w_i$) in $W_l$
  \item any $P^*_i \in \P^{(*,l)}_{\bar{C}W}$ uses at least one vertex
    ($v_i$) in $B_l,$ implying $|\P^{(*,l)}_{\bar{C}W}| \le |B_l|$
  \end{itemize}
  Let us now prove the following property used to restructure paths in
  $\P^{(*,l)}_{\bar{C}W}.$
  \begin{quote}
  \textsf{Property ($\textsf{P}_1$)}: For any $P^*_i \in
  \P^{(*,l)}_{\bar{C}W},$ if we replace $w_i$ by any $w'_i \in W_{l},$
  then $\{u_i,v_i,w'_i\}$ is still an induced 2-path.
  \end{quote}
  Indeed, by Property~\ref{@standardisation} of a bucket decomposition, $N(v_i) \cap
  V(W)=W_l,$ implying that $\{v_i,w'_i\} \in E(G).$

 Let us now prove that $\{u_i,w_i\} \in E(G)$ iff $\{u_i,w'_i\} \in
 E(G),$ which will imply that $\{u_i,v_i,w'_i\}$ is an induced 2-path.
 If $u_i \in B_\emptyset,$ then, by Property~\ref{@dispenserait} of a bucket
 decomposition, we have that $\{u_i,w_i\} \notin E(G) $ and
 $\{u_i,w'_i\} \notin E(G).$  If $u_i \in B_{l'},$ for $l' \neq l,$
 then, by Property~\ref{@standardisation} of a bucket decomposition, we have
 $\{u_i,w_i\} \notin E(G)$ and $\{u_i,w'_i\} \notin E(G).$ Finally, if
 $u_i \in B_{l},$ then, by Property~\ref{@standardisation} of a bucket
 decomposition, we have $\{u_i,w_i\} \in E(G)$ and $\{u_i,w'_i\} \in
 E(G).$
  
 Let us now define $\P_W.$ For any color $c$ of
 $(\tilde{G},χ)\langle W,B,C\rangle,$ let $M(c)$ be the edge of color
 $c$ in $M.$ Recall that a color $c$ can either belong to $B_l$ for
 some $l \in \bar{F}$ (in which case $|M(c)|=1$ and
 $V(M(c)) \subseteq W_l$) or belong to $C$ (in which case $|M(c)|=2$
 and $V(M(c)) \subseteq W_{\bar{F}}$). \stf{Moreover, as the partial
   decomposition $(W,B,C)$ is clean, for every vertex $c\in C$ there
   exists an edge of $(\tilde{G},χ)\langle W,B,C\rangle$ with color
   $c$ and then $M$ also contains an edge with color $c$ and $M(c)$ is
   well defined.}  Thus, for any $l \in \bar{F}$ and any
 $P^*_i \in \P^{(*,l)}_{\bar{C}W}$ (such that $P^*_i=\{u_i,v_i,w_i\}$
 with $v_i \in B_l$ and $w_i \in W_l$), we define
 $P_i = \{u_i,v_i,M(v_i)\}.$ For any $P^*_i \in \P^*_{CW},$ let
 $u_i \in P^*_i \cap C,$ and let $P_i = \{u_i\} \cup M(u_i).$
 
 Let us now prove Properties~\ref{@dishonourable}, \ref{@prevaleciera}, and
 \ref{@gentillesses}.  We start with Property~\ref{@dishonourable}.  For any $l
 \in \bar{F}$ and any $P^*_i \in \P^{(*,l)}_{\bar{C}W},$ $P_i$ is an
 induced 2-path as $M(v_i) \in W_l$ and according to Property
 ($\textsf{P}_1$), and $P_i \subseteq V(G')$ as $M(v_i) \in V(G')$ and
 $\{u_i,v_i\} \subseteq B.$ For any $P^*_i \in \P^*_{CW},$ $P_i$ is an
 induced 2-path as $M(u_i)$ is an edge of $G(S)$ of color
 $u_i.$ Moreover, $P_i \subseteq V(G')$ as $M(u_i) \subseteq V(G')$
 and $u_i \in C.$  Property~\ref{@gentillesses} is direct from the
 definition, and Property~\ref{@prevaleciera} is verified as, inside $W,$
 all used vertices are defined using the matching and, outside $W,$ we
 have Property~\ref{@gentillesses}.\medskip

 {\bf Size requirement.} Now our objective is to upper bound
 $|V(G')|.$ Let us assume first that $W \neq \emptyset,$ and thus that
 $(\tilde{G},χ)\langle W,B,C\rangle$ is not the empty graph.  Let us
 start with $|V(M)|.$ Recall that the set of colors in
 $(\tilde{G},χ)\langle W,B,C\rangle$ is
 $C \cup \bigcup_{i \in \bar{F}}B_i=C \cup B_{\neq \emptyset} .$ As
 edges of colors in $C$ has size two and edges of colors in
 $B_{\neq \emptyset}$ have size one, and $M$ contains one edge of each
 color, we get $|V(M)| = 2|C|+|B_{\neq \emptyset}|.$

We are ready to bound the size of $G'$:
  \begin{eqnarray*}
    |V(G')| &=& |B|+|C|+|V(M)| \\
    &=& |B_{\emptyset}|+|B_{\neq \emptyset}| + |C|+ 2|C|+|B_{\neq \emptyset}|\\
    & \le& |B_\emptyset|+|B_{\neq \emptyset}|(1+2c_1) + 3|C|\\
    & \le &(1+2c_1)s(W,B) + 3|C| \\
    & \le & (1+2c_1)^2|C^<| + 3|C|, \mbox{ according to the {size requirement}~{\ref{@conflictivitat}}} \\
    & \le & (1+2c_1)^2(|C^<| + |C|)\\
    & = & (1+2c_1)^2|C^0|\\
   & \le & 3(1+2c_1)^2k\\
  \end{eqnarray*}

If $W = \emptyset,$ then $(\tilde{G},χ)\langle W,B,C\rangle$ is the
empty graph and $M = \emptyset$ and the above equations still hold.
\end{proof}

\stf{Now, let us pay attention to the cases where the reduction rule
provides a small vertex cover. Notice that for these cases, we do not
necessarily need the considered partial decompositions to be clean.}

\begin{lemma}[Case 1 of small vertex cover]\label{lemma_safe_rule2P3}
  Let $(W,B,C)$ be a partial decomposition.
 Suppose that there is a non-empty set of colors  $X$ such that
    {$(\tilde{G},χ)\langle W,B,C\rangle[X]$} admits a vertex cover $T$ such
    that $|T| \le {c_1}|X|$ (see Figure~\ref{fig_partial_decompo}).
    Let $X^B = X \cap B$ and $X^C = X \cap C.$  Let $S(X^B) = \{i \in
    \bar{F} \mid X^B \cap B_i \neq \emptyset\}$ be the set of buckets
    containing a color in $X^B$ and $W(X^B) = \bigcup_{i \in  S(X^B)}W_i.$
    \begin{itemize}
\item (Case 1:) If $|X^B| \le |X^C|$, let $W' = W \setminus T,$ $B' = B \cup T \cup X^C,$ and $C'=C \setminus X^C.$ Then $(W',B',C')$ is a partial decomposition. 
    \end{itemize}
\end{lemma}
\begin{proof}

  \textbf{Partition requirements.}  The only non-trivial
 property is that $(W',B')$ is a nice pair. Now we will prove it using
 the definition of nice pair (rather than providing a bucket
 decomposition of $(W',B')$).  Let $Z = T \cup X^C.$  Recall that $W'
 = W \setminus T$ and $B'=B \cup Z.$  We clearly have that $W' \cap B'
 = \emptyset,$ $W' \subseteq W^0$ and thus we only have to prove that
 for any induced 2-path $P$ of $G[W' \cup B'],$ $|P \cap W'| \le 1.$
 Notice first that, as $T$ is a vertex cover of {$(\tilde{G},χ)\langle
   W,B,C\rangle[X]$}, it is in particular a vertex cover for all edges
 of colors $X^C,$ implying that for any edge $e$ with color $c \in
 X^C,$ $T \cap V(e) \neq \emptyset.$  This implies that there is no
 induced 2-path $P$ where $|P \cap X^C|=1$ and $|P \cap W'|=2,$ as $P
 \cap W'$ would be an edge (of a color in $X^C$) not covered by $T.$
 Moreover, there is no induced 2-path $P$ where $|P \cap T|=1$ and $|P
 \cap W'|=2$ as this would imply $P \subseteq W,$ and $G[W]$ is a
 disjoint union of cliques.  These last two observations imply that
 there is no induced 2-path $P$ where $|P \cap Z|=1$ and $|P \cap
 W'|=2.$  Moreover, as $(W,B)$ was a nice pair, there is also no
 induced 2-path $P$ where $|P \cap B|=1$ and $|P \cap
 W|=2.$ Therefore, there is no induced 2-path $P$ where $|P \cap B|=1$
 and $|P \cap W'|=2,$ as $W' \subseteq W.$ This implies that $(W',B')$
 is a nice pair.\smallskip
 
 \textbf{Size requirements.}  Let us prove that $s(W',B') \le
 (1+2c_1)|C^{'<}|.$ Recall first that $C^{'<}=C^0 \setminus C'$ and
 $C'=C \setminus X^C,$ implying that $C^{'<}=C^< \cup X^C.$  Notice
 also that $|T| \le c_1|X| = c_1(|X^B|+|X^C|) \le 2c_1|X^C|,$ as we
 are in Case 1.  Observe first that $B_\emptyset \subseteq
 B'_\emptyset,$ as $W' \subseteq W.$ This implies that there exists
 $Z_1 \subseteq Z$ such that $B'_\emptyset = B_\emptyset \cup Z_1.$
 
 We remark that $B_{\neq \emptyset} \subseteq B'_{\neq \emptyset}$ is not true as a vertex $v$ in a $B_i$ may move to $B'_\emptyset$ if  $W_i  \subseteq T$ (for example the case for vertices in $B_3$ in Figure~\ref{fig_partial_decompo}).
 We have
  \begin{eqnarray*}
    s(W',B') &=& \frac{|B'_\emptyset|}{1+2c_1}+|B'_{\neq \emptyset}| \\
    &=& \frac{|B_\emptyset|+|Z_1|}{1+2c_1}+|B'_{\neq \emptyset}|\\
    &\le& \frac{|B_\emptyset|}{1+2c_1}+|Z_1|+|B'_{\neq \emptyset}|
  \end{eqnarray*}
Now, observe that
  \begin{eqnarray*}
    B'_\emptyset \cup B'_{\neq \emptyset}  &=& B \cup Z \\
    \Rightarrow B_\emptyset \cup Z_1 \cup B'_{\neq \emptyset}  &=& B_\emptyset \cup B_{\neq \emptyset} \cup Z \\
    \Rightarrow  Z_1 \cup B'_{\neq \emptyset}  &=& B_{\neq \emptyset} \cup Z 
  \end{eqnarray*}
  Thus, we get
  \begin{eqnarray*}
    s(W',B') &\le& \frac{|B_\emptyset|}{1+2c_1}+ |B_{\neq \emptyset}|+ |Z| \\
    &\le& \frac{|B_\emptyset|}{1+2c_1}+ |B_{\neq \emptyset}|+ |T|+|X^C| \\
    &\le& \frac{|B_\emptyset|}{1+2c_1}+ |B_{\neq \emptyset}|+ (1+2c_1)|X^C| \\
    &=& s(W,B)+ (1+2c_1)|X^C| \\
    &\le&(1+2c_1)|C^{<}|+ (1+2c_1)|X^C|, \mbox{ as $(W,B,C)$ is a partial decomposition}\\
    &=&(1+2c_1)|C^{'<}|
  \end{eqnarray*}

\end{proof}
  
Let us now prove that the vertex cover and colors we add in the kernel
output in Case 2 are small.

 \begin{lemma}[Case 2 of small vertex cover]\label{lemma_safe_rule1P3}
   Let $(W,B,C)$ be a partial decomposition.
 Suppose that there is a non-empty set of colors  $X$ such that
    {$(\tilde{G},χ)\langle W,B,C\rangle[X]$} admits a vertex cover $T$ such
    that $|T| \le {c_1}|X|$ (see Figure~\ref{fig_partial_decompo}).
    Let $X^B = X \cap B$ and $X^C = X \cap C.$  Let $S(X^B) = \{i \in
    \bar{F} \mid X^B \cap B_i \neq \emptyset\}$ be the set of buckets
    containing a color in $X^B$ and $W(X^B) = \bigcup_{i \in  S(X^B)}W_i.$
    \begin{itemize}
\item (Case 2:) If $|X^B| > |X^C|$, let $W' = W \setminus W(X^B),$ $B'= B \cup W(X^B),$ and $C'=C.$ Then $(W',B',C')$ is a partial decomposition. 
    \end{itemize}
 \end{lemma}

    \begin{proof}
      Let us first prove the following claim.
\begin{claim}
\label{claim:claim1}
 We have $|W(X^B)| \le 2c_1 |B_{S(X^B)}|.$
\end{claim}

{\it Proof of Claim~1.}  Notice first that, as $T$ is a vertex cover
of $(\tilde{G},χ)\langle W,B,C\rangle[X],$ then for any $i \in S(X^B)$
and any color $v \in X^B \cap B_i,$ as
$(\tilde{G},χ)\langle W,B,C\rangle[X]$ contains all edges
$\{u\} \in W_i$ of color $v,$ it implies that $T \supseteq W_i.$ This
implies that $T \supseteq W(X^B).$ Moreover, by
Corollary~\ref{cor_rainbowvc},
$|T| \le c_1|X| = c_1(|X^B|+|X^C|) \le 2c_1|X^B|,$ as we are in Case
2.  As $|W(X^B)| \le |T|$ and $|X^B| \le|B_{S(X^B)}|,$ we obtain the
result of the claim.\hfill \claimqed
\medskip

  Let us now prove that $(W',B',C')$ is a partial decomposition.\medskip

  \textbf{Partition requirements.}  The only non-trivial
 property is that $(W',B')$ is a nice pair, which we will prove by
 providing a bucket decomposition of $(W',B').$  Let us define the
 following partitions which, informally, correspond to the result of
 moving (for any $i \in S(X^B)$) each $W_i$ to $B_i$:
 \begin{itemize}
 \item the partition $\{W'_i\mid i \in [i_0]\}$ of $W',$ where $W'_i =
   W_i$ if $i \notin S(X^B)$ and $W'_i = \emptyset$ otherwise,
 \item the partition $\{B'_i\mid i \in [i_0]\}$ of $B',$ where $B'_i =
   B_i$ if $i \notin S(X^B)$ and $B'_i = \emptyset$ otherwise,
 \item $B'_\emptyset = B_\emptyset \cup W(X^B) \cup B_{S(X^B)}$
 \end{itemize}
 Let us check that these partitions verify conditions of a bucket
 decomposition. The only non-trivial part is to verify that all
 vertices $v$ added to $B_\emptyset$ have $N(v) \cap W' =
 \emptyset.$ If $v \in W(X^B),$ then $N(v) \cap W' = \emptyset,$ as
 there we no edges between the $W_i$'s.  If $v \in B_{S(X^B)},$ where
 $v \in B_i$ for some $i \in S(X^B),$ then, as $N(v) \cap W = W_i$ and
 $W_i \subseteq W(X^B),$ we also get $N(v) \cap W' = \emptyset.$
 Thus, Property~\ref{@licenciadillo} of Lemma~\ref{lemma_nicepair} implies that
 $(W',B')$ is a nice pair.\medskip
 
  \textbf{Size requirements.}
 Let us prove that $s(W',B') \le (1+2c_1)|C^{'<}|.$ Recall first that $C^{'<}=C^<,$ as $C'=C.$
 As $(W,B,C)$ is a partial decomposition, we have $s(W,B) \le (1+2c_1)|C^{<}|,$ and thus it only remains to prove that $s(W',B') \le s(W,B).$
 By definition we have $s(W,B)=\frac{|B_\emptyset|}{1+2c_1}+|B_{\neq \emptyset}|.$
 Recall that $B_{\neq \emptyset}=B_{\bar{F}},$ and observe that there is a partition
 $\{B_{S(X^B)}, B_{\bar{F} \setminus S(X^B)}\}$ of   $B_{\bar{F}}.$
 This implies that $s(W,B) =\frac{|B_\emptyset|}{1+2c_1}+|B_{S(X^B)}| + |B_{\bar{F} \setminus S(X^B)}|.$
 Now, observe that the bucket decomposition of $(W',B')$ is exactly the partition of $W'$ and $B'$ defined above and that
 $B'_{\neq \emptyset} = B_{\bar{F} \setminus S(X^B)}.$
 This implies that
  \begin{eqnarray*}
    s(W',B') &=& \frac{|B'_\emptyset|}{1+2c_1}+|B'_{\neq \emptyset}| \\
    &=& \frac{|B_\emptyset|+|W(X^B)|+|B_{S(X^B)}|}{1+2c_1}+|B_{\bar{F} \setminus S(X^B)}|\\
    &\le& \frac{|B_\emptyset|+(1+2c_1)|B_{S(X^B)}|}{1+2c_1}+|B_{\bar{F} \setminus S(X^B)}|, \mbox{ by Claim 1}\\
    &=& s(W,B)
  \end{eqnarray*}
 
\end{proof}

\subsection{Analysis of the overall kernel}
\stf{We are now ready to define the unique, other than the cleaning
  phases, rule for our kernelization. This rule summarizes the
  different cases arising in the previous lemmas.}

\begin{definition}[Reduction Rule for \IPP]~ \stf{Given a clean partial
  decomposition $(W,B,C)$, with the associated partition for $B$ and
  $W$ denoted by
  $(B_\emptyset, B_1,\dots, B_{i_0}, W_1,\dots ,W_{i_0})$ with
  $F = \{i \in [i_0] \mid W_i = \emptyset\}$ and
  $\bar{F} = [i_0] \setminus F$, let us define the output $R(W,B,C)$
  of the rule $R$ as follows:}
  \begin{itemize}
  \item Compute using Corollary~\ref{cor_rainbowvc} (with a suitable
    ε precised later) if there exists a rainbow matching
    $M$ in the $p$-edge-colored mutigraph $(\tilde{G},χ)\langle W,B,C\rangle$ (where $p=|C|+\sum_{i\in
      [i_0]}|B_{i}|$ -- notice that $p=\O(|V(G)|)$).
  \item If $M$ exists, then return $V(M) \cup B \cup C.$
  \item Otherwise, let $X$ be the non-empty set of colors such that
    {$(\tilde{G},χ)\langle W,B,C\rangle[X]$} admits a vertex cover $T$ such
    that $|T| \le {c_1}|X|$ (see Figure~\ref{fig_partial_decompo}).
    Let $X^B = X \cap B$ and $X^C = X \cap C.$  Let $S(X^B) = \{i \in
    \bar{F} \mid X^B \cap B_i \neq \emptyset\}$ be the set of buckets
    containing a color in $X^B$ and $W(X^B) = \bigcup_{i \in
      S(X^B)}W_i.$
\item If $|X^B| \le |X^C|$ (case 1),  let $W' = W \setminus T,$ $B' = B \cup T
  \cup X^C,$ and $C'=C \setminus X^C.$ Return $(W',B',C').$
\item Otherwise (case 2), let $W' = W \setminus W(X^B),$ $B'
  = B \cup W(X^B),$ and $C'=C.$ Return  $(W',B',C').$
 \end{itemize}
\end{definition}

\noindent
We then obtain the following.

\begin{lemma}\label{lemma:ruleIPP}
  \stf{Given a clean partial decomposition $(W,B,C)$, Rule $R(W,B,C)$ either returns}:
  \begin{itemize}
    \item a set $A\subseteq V(G)$ such that, if $G'=G[A],$ then $(G',k)$
      and $(G,k)$ are equivalent instances of \IPP\ and $|A|\leq 3(1+2c_1)^2k.$
    \item or a partial decomposition $(W',B',C')$ such that $|\bar{F'}|+|C'| < |\bar{F}|+|C|$
  \end{itemize}
\end{lemma}

\begin{proof}
  If $R$ finds a rainbow matching $M$ in
  $(\tilde{G},χ)\langle W,B,C\rangle$, then
  \autoref{lemma_safe_rainbow} immediately implies that the set
  $A=V(M) \cap B \cup C$ verifies the claimed properties.  Let us now
  consider that $R$ does not find a rainbow matching.

  If $R$ falls into Case 1, then according to
  Lemma~\ref{lemma_safe_rule2P3}, $(W',B',C')$ is a partial
  decomposition.  Moreover, in Case 1, as $X \neq \emptyset,$ either
  $X^B \neq \emptyset,$ implying that $S(X^B) \neq \emptyset$ and thus
  that all cliques $W_i$ for $i \in S(X^B)$ become empty and therefore
  $|\bar{F}|$ strictly decreases (and $|C|$ does not increase).
  Otherwise, $X^C \neq \emptyset,$ implying that $|C|$ strictly
  decreases (and $|\bar{F}|$ does not increase).

  If $R$ falls into Case 2, then according to
  Lemma~\ref{lemma_safe_rule1P3}, $(W',B',C')$ is a partial
  decomposition.  Moreover, in Case 2, $X^B \neq \emptyset,$ implying
  as previously that $|\bar{F}|$ strictly decreases (and $|C|$ does
  not change).
\end{proof}

Finally, we can prove the kernelization algorithm for \IPP stated in
\autoref{secondi_ess}.

\begin{proof}[Proof of \autoref{secondi_ess}]
  \stf{Given an input $(G,k),$ we define the kernelization Algorithm
    \texttt{B} which starts with a greedy localization phase, as
    explained in \autoref{@challengingly}. Assume that it {does {not}
      find a packing of size $k$} and therefore computes a greedy
    localized pair $(C^0,W^0).$ We then consider the partition
    $(W_0,\emptyset ,C_0)$ of $V(G)$. It is straightforward to check
    that this partition is a partial decomposition of $G$. Using
    Lemma~\ref{lem:clean-I2PP}, we then obtain the first clean partial
    decomposition $(W,B,C)$ of the process.  Now, a step of
    Algorithm~\texttt{B} will be made of Reduction Rule $R$ for \IPP
    and a cleaning phase. Algorithm~\texttt{B} exhaustively performs
    steps, obtaining a clean partial decomposition at the end of each
    step and stopping only when it
    falls into the matching case.\\
    Let us prove, by induction on $|\bar{F}|+|C|,$ that applying
    exhaustively steps of Algorithm~\texttt{B} terminates in
    polynomial time and outputs an equivalent instance $(G',k)$ where
    $|V(G')| \le 3(1+2c_1)^2k.$ In order to this, notice first that in
    the case where Rule $R$ applied on a clean partial decomposition
    $(W,C,B)$ returns a partial decomposition $(W',B',C')$ with
    $|\bar{F'}|+|C'| < |\bar{F}|+|C|$, we apply a cleaning phase on
    $(W',B',C')$ to obtain a clean partial decomposition
    $(W'',B'',C'')$. Then it is easy to show that we also have
    $|\bar{F''}|+|C''| < |\bar{F}|+|C|$. Indeed, in the cleaning
    phase, the set $W'$ is unchanged, so it is for $F'$, and the size
    of $C'$ can only decrease. We obtain
    $|\bar{F''}|+|C''|\le |\bar{F'}|+|C'| < |\bar{F}|+|C|$ as desired.\\
    Now, we can finish the analysis of the process.  If
    $\bar{F} = C = \emptyset,$ then $W = \emptyset,$
    $(\tilde{G},χ)\langle W,B,C\rangle$ is the empty graph, and we
    consider that $R$ returns the rainbow matching
    $M=\emptyset$. Thus, according to \autoref{lemma:ruleIPP}, the
    rule $R$ outputs a set $A\subseteq V(G)$ such that, if $G'=G[A],$
    then $(G',k)$ and $(G,k)$ are equivalent instances of \IPP\ and
    $|A|\leq 3(1+2c_1)^2k.$ Now, if $|\bar{F}|+|C| >0$, it is
    immediate, using induction, by \autoref{lemma:ruleIPP} and the
    previous remarks about cleaning phases, that \texttt{B} terminates
    in polynomial time and outputs an equivalent
    instance $(G',k)$ where $|V(G')| \le 3(1+2c_1)^2k.$\\
    As $c_{1}=4+ε,$ we conclude that $|A|≤(243+12ε^2+108ε)k=(243+ε')k$
    for a suitable $ε'$ as required.}
\end{proof}


\section{Linear kernel for \IPHS}
\label{@linearKernelIPHS}

\stf{In this section we focus on the induced 2-paths hitting set probem,
restated below.}

\defparproblem{{\sc Induced 2-paths Hitting Set} (\IPHS)}{$(G,k)$ where
  $G$ is a graph and $k \in \N$}{$k.$}{Is there an induced
  2-paths Hitting Set of $G$ of size at least $k$?}

\stf{We obtain a linear kernel for this problem, which is obtained by the
same algorithm that the one designed in the previous section. The only
thing that we will have to check is that, when the algorithm stops, we
obtain an equivalence instance than the input instance for \IPHS.}

\begin{theorem}
\label{theo:kernelIPHS}
\stf{There exists some function $f:\mathbb{N}\to\mathbb{N}$ such that for
every $ε>0,$ there exists an algorithm that, given an instance $(G,k)$
of \IPHS\ outputs a set $A\subseteq V(G)$ such that $(G[A],k)$ is an
equivalent instance where $|A|≤(243+ε)k.$ Moreover this algorithm runs
in time $\O(|V(G)|^{f(ε)}).$ In other words, \IPHS admits a kernel of a
linear number of vertices.}
\end{theorem}

\stf{The key tool to obtain the linear kernel for \IPHS is the following
lemma which is the analog of Lemma~\ref{lemma_safe_rainbow} for \IPHS.
All the notations and definitions follow previous section.}

\begin{lemma}[Case of rainbow matching for \IPHS]\label{lem:rainbow-hitting-P2}
\stf{Let $(W,B,C)$ be a clean partial decomposition. Suppose that the colored multigraph
$(\tilde{G},χ)\langle W,B,C\rangle$ admits a rainbow matching $M$.
 Let $A=V(M) \cup B \cup C,$ and $G'=G[A].$ Then,
  \begin{enumerate}
  \item $(G,k)$ and $(G',k)$ are equivalent instances of \IPHS,
  \item $|A| \le 3(1+2c_1)^2k.$
  \end{enumerate}}
\end{lemma}

\begin{proof}
\stf{The size requirement concerning $A$ follows from
  Lemma~\ref{lemma_safe_rainbow}. Let us prove that $(G,k)$ and
  $(G',k)$ are equivalent instances of \IPHS. As $G'$ is an induced
  subgraph of $G$, it is clear that if $G$ admits a 2-induced paths
  hitting set of size at most $k$, then it is also the case for $G'$.\\
  For the converse direction, assume that $G'$ admits a 2-induced
  paths hitting set $X$ of size at most $k$, and let us see how to
  build one for $G$. By removing vertices from $X$ if necessary, we
  can assume that $X$ is a induced 2-paths hitting of $G'$ minimal by
  inclusion. Let us first precise the structure of $X$. As $(W,B,C)$
  is clean, for every vertex $c$ of $C$ there exists at least one edge
  of $(\tilde{G},χ)\langle W,B,C\rangle$ with color $c$. And as $M$ is
  a rainbow matching of this graph, there exist $v_c$ and $w_c$ in $W$
  such that $v_cw_c$ is an edge of $M$ and so, such that
  $\{v_c,w_c,c\}$ induces a 2-path $P_c$ of $G$. Notice that $X$ has
  to intersect $P_c$ for every vertex $c$ of $C$. Denote by $M_C$ the
  set $\bigcup_{c\in C}\{v_c,w_c\}$ and by $M_B$ the set
  $M\setminus M_C$. By the previous remark, we obtain that
  $|X\cap (M_C\cup C)|\ge |C|$. Let us focus now on $M_B$.\\
  For any $i\in [i_0]$, for every vertex $b\in B_i$, every vertex of
  $W_i$ receives color $b$. So, as $M$ is a rainbow matching of
  $(\tilde{G},χ)\langle W,B,C\rangle$, it has to contain a vertex of
  $W_i$ with color $b$. That is $M_B$ contains exactly $|B_i|$
  vertices of $W_i$ for every $i\in [i_0]$. Moreover, assume that for
  $i\in [i_0]$ the set $X$ contains a vertex $x$ of $M_B\cap W_i$. By
  minimality of $X$ there exists an induced 2-path $P$ of $G'$ with
  $X\cap V(P)=\{x\}$. As $(W,B)$ is a nice pair, we know that $P=xby$
  with $b\in B_i$ and $y\in B\cup C$. In particular, by
  Lemma~\ref{lemma_nicepair}, for any $x'\in M_B\cap W_i$ the path
  $x'by$ is also an induced 2-path. Thus, as $b\notin X$ and
  $y\notin X$ we must have $x'\in X$, and in all, we obtain
  $M_B\cap W_i\subseteq X$.}

\stf{Now, we can modify $X$ in order to obtain an induced 2-paths hitting
  set of $G$. Denote by $J$ the subset of $[i_0]$ corresponding to the
  indices of $W_i$'s intersected by $X$, that is
  $J=\{i\in [i_0]\ :\ X\cap W_i\neq \emptyset \}$. Let us define
  $X'=C\cup (X\cap B) \cup \bigcup_{i\in J} B_i$ and show that $X'$ is
  induced 2-paths hitting set of $G$ and has size no more than $|X|$.
  For the latter property, using the previous remarks, we have:
  \begin{eqnarray*}
    |X'| &\le& |C|+ |(X\cap B)|+ \sum_{i\in J} |B_i|\\
         &\le& |X\cap (C\cup M_C)|+|(X\cap B)|+\sum_{i\in J}|M_B\cap W_i|\\
         &\le& |X\cap (C\cup M_C)|+|(X\cap B)|+|X\cap M_B|\\
         &\le& |X|
  \end{eqnarray*}
  To prove that $X'$ is an induced 2-paths hitting set of $G$, assume
  by contradiction that $P=xyz$ is an induced 2-path of
  $G\setminus X'$. As $C\subseteq X'$ we have $V(P)\subseteq W\cup B$,
  and as $(W,B)$ is a nice pair, we have for instance $x\in W$,
  $y\in B$ and $z\in B$. So, by Lemma~\ref{lemma_nicepair}, there
  exists more precisely $i\in [i_0]$ such that $x\in W_i$ and
  $y\in B_i$. The path $P$ is not intesected by $X'$, meaning in
  particular that $i\notin J$ (as otherwise we would have
  $B_i\subseteq X'$). So we have $y\notin X$ and $z\notin X$ (as
  $X'\cap B=X\cap B$) and $W_i\cap X=\emptyset$. However, $M$ contains
  at least one vertex $x'$ of $W_i$, with color $y$ for instance. But
  then $x'yz$ is an induced 2-path of $G'$ not intersected by $X$, a
  contradiction.}
\end{proof}

\stf{Now, using Lemma~\ref{lem:rainbow-hitting-P2} instead of
  Lemma~\ref{lemma_safe_rainbow} in the proof of
  Lemma~\ref{lemma:ruleIPP}, we directly obtain the analog of this
  latter one for \IPHS.}

\begin{lemma}\label{lemma:ruleIPHS}
  \stf{Given a clean partial decomposition $(W,B,C)$, Rule $R(W,B,C)$ either returns:
  \begin{itemize}
    \item a set $A\subseteq V(G)$ such that, if $G'=G[A],$ then $(G',k)$
      and $(G,k)$ are equivalent instances of \IPHS\ and $|A|\leq 3(1+2c_1)^2k.$
    \item or a partial decomposition $(W',B',C')$ such that $|\bar{F'}|+|C'| < |\bar{F}|+|C|$
  \end{itemize}}
\end{lemma}

\stf{Now, proof of Theorem~\ref{theo:kernelIPHS} works the same than
  the kernelization process for \IPP, that
  is~\autoref{secondi_ess}. We start by computing a localized pair
  $(C^0,W^0)$. The set $C^0$ induces a packing of induced 2-paths. If
  there is more than $k$ induced 2-paths in the packing, then $G$ has
  no induced 2-paths hitting set of size at most $k$. Otherwise, we
  consider the initial partial decomposition $(W_0,\emptyset ,C_0)$ of
  $V(G)$. Then, we exhaustively alternate a cleaning phase with an
  application of the rule $R$, until this last one falls into the
  matching case. As in the proof of~\autoref{secondi_ess}, using
  Lemma~\ref{lemma:ruleIPHS} an induction on $|\bar{F}|+|C|$ shows
  that this later case appears after a polynomial number of steps.
  Then we conclude with Lemma~\ref{lem:rainbow-hitting-P2}.
 } 


\section{An (almost) linear kernel for \TPT}
\label{@industriousness}

\subsection{Notations}
Given a tournament $T,$ a \emph{triangle} $\Delta$ in $T$ is a
subgraph on three vertices where each vertex has in-degree and
out-degree exactly one, i.\@e.\@ a directed cycle of length three.  A
\emph{triangle-packing} $\P = \{\Delta_i, i \in [x]\}$ is a set of
vertex-disjoint triangles of $T.$  The \emph{size} of $\P$ is
$|\P|=x.$  
\defparproblem{\textsc{triangle-packing in
    Tournament}\ (\TPT)}{$(T,k)$ where $T$ is a tournament and $k
  \in \N$}{$k$}{Is there a triangle-packing of size at least $k$?}

In this section we prove the following theorem.

%
%

\begin{theorem}
\label{main_treh}
There exists an algorithm that, given an instance $(T,k)$ of \TPT
outputs a set $S\subseteq V(T)$ such that $(T[S],k)$ is an equivalent
instance of $(T,k)$ where, for every $\delta$ with $1<\delta \le 2,$ we
have $|S|\leq 6534\cdot c(δ)\cdot  k^\delta$ (where $c(\delta)=\max(\frac{20}{(2^\delta-2)},(\frac{21}{\delta})^\frac{1}{\delta-1})$).
In other words, for any $\delta$ with $1 < \delta \le 2$ \TPT admits a 
kernel with $6534\cdot c(\delta)k^\delta$ vertices.
\end{theorem}

By fixing the suitable value $\delta_0=1+\frac{\sqrt{\log{21}}}{\sqrt{\log k}}$ (assuming the
non-trivial case where $k\geq 2$) we obtain the following.

\begin{corollary}
  \label{@CorollaryTPT}
 \TPT admits a kernel with $k^{1+\frac{\O(1)}{\sqrt{\log{k}}}}$  vertices.
\end{corollary}
\begin{proof}
  Let us upper bound $c(\delta_0)$. For any $1 < \delta \le 2$, we have $c(\delta) \le c'(\delta)$ where $c'(\delta)=21^{\frac{1}{\delta-1}}$.
 Indeed, $\frac{20}{(2^\delta-2)} \le \frac{21}{\delta-1} \le 21^{\frac{1}{\delta-1}}$, and on the other hand, $(\frac{21}{\delta})^\frac{1}{\delta-1} \le 21^\frac{1}{\delta-1}$
 Thus, the vertex size of the kernel given by Theorem~\ref{main_treh} for $\delta_0$ is at most
 \begin{eqnarray*}
   6534\cdot  c'(\delta_0)k^\delta_0 &= & 6534\cdot 2^{\left( \sqrt{\log{21}}\sqrt{\log{k}}\right)}k^{\delta_0} \\
   &= &6534\cdot 2^{\left(\sqrt{\log{21}}\sqrt{\log{k}}\right)}k\cdot 2^{\left(\sqrt{\log{21}}\sqrt{\log{k}}\right)}\\
   & = & 6534\cdot k\cdot 2^{2\cdot\sqrt{\log{21}}\sqrt{\log{k}}}\\
      & =&  \O(k\cdot 2^{\frac{{9.17\cdot \log{k}}}{\sqrt{\log{k}}}})\\
   & =&  \O(k\cdot k^{\frac{9.17}{\sqrt{\log{k}}}})\\
   & =&  k^{1+\frac{\O(1)}{\sqrt{\log{k}}}}
 \end{eqnarray*}
\end{proof}

\subsection{Preliminary phase: greedy localization}
\label{greedy_loc}

Given an instance $(T,k)$ of \TPT, we first greedily compute a
maximal set of vertex-disjoint triangles.  If we get at least $k$
triangles, then $(T,k)$ is a positive instance of \TPT.  Otherwise
we denote by $C^0$ the set of vertices contained in the triangles of
the greedy packing and by $W^0$ the set $V(T)\setminus C^0.$  We
denote by $t_0$ the size of $W^0.$  The vertices of $W^0$ clearly
induce an acyclic subtournament of $T,$ and we call such a partition
$(C^0,W^0)$ a \emph{greedily localized pair} of $T$ for \TPT.

In the remainder of the section, we consider that $W^0$ is sorted
according to its topological ordering\footnote{Notice that the
topological ordering of an acyclic tournament is unique.} and we
number the elements of $W^0$ following this order, that is
$W^0=\{w_1,\dots ,w_{t_0}\}$ with $(w_i,w_j) \in A(T)$ iff $i<j.$
Finally, for any subset $W$ of $W^0$ and any $i,j \in [t_0]^2$ we write $W_{[i,j]} = \{w_k\in W \mid k\in[i,j] \}$
and $W_{[i,j[}=W_{[i,j]} \setminus \{w_j\}.$

\subsection{Nice pairs, buckets, and partial decomposition}

Recall that a basic principle of our technique is to maintain a tuple $(W,B,C,\dots)$ called a partial decomposition
  where $(W,B)$ is a nice pair. To capitalize on the specific problem considered here, we study in this section which structural properties hold in a nice pair.

In the entirety of this section we consider that we are given an input
$(\T,k)$ of \TPT, and a greedily localized pair $(C^0,W^0)$ for our
input $(\T,k),$ where $W^0=\{w_1,\dots ,w_{t_0}\}.$
\begin{definition}
Given an instance $(\T,k)$ of \TPT, and two subsets $W, B$ of
$V(\T),$ we say that $\npair=(W,B)$ is a \emph{nice pair} if $W
\subseteq W^0$ and $B \cap W = \emptyset,$ and for any triangle
$\Delta$ of $\T[W \cup B]$ we have $|\Delta \cap W| \le 1.$  Notice
that $B$ may contain vertices of $C^0,$ as well as vertices of
$W^0\setminus W.$
\end{definition}

An illustration of next proposition is depicted in \autoref{@experimentation}.
  
\begin{proposition}\label{@accompanying}
Let $\npair=(W,B)$ be a nice pair.  There exists a unique $S^\npair
\subseteq [t_0] \cup \{\infty\},$ where $\infty$ is a token
representing some value greater than $t_0,$ and a unique partition
$\{B_i\mid i\in S^{ \psi}\}$ of $B$ into non-empty sets such that for
any $i \in S^\npair,$ the set $B_i$ contains every vertex $v$ of $B$
where
    \begin{itemize}
    \item all arcs between $W_{[1,i-1]}$ and $v$ are oriented from
      $W_{[1,i-1]}$ to $v$
    \item all arcs between $W_{[i,t_0]}$ and $v$ are oriented from $v$
      to $W_{[i,t_0]}$
    \end{itemize}
Such a partition, along with the choice of $S^{\psi},$ is called a
\emph{bucket decomposition} of $\npair,$ and sets $B_i$ are called
\emph{buckets}.
\end{proposition}

\begin{proof}
For any $i\in[t_0]\cup\{\infty\},$ denote by $B_i$ the set containing
every vertex of $B$ that is dominated by $W_{[1,i-1]}$ and dominates
$W_{[i,t_0]}.$  As there is no cycle of length two in $\T,$ the $B_i$
are pairwise disjoint.  Moreover, let $b$ be a vertex of $B.$  If $W$
dominates $b,$ then we have $b\in B_\infty.$  Otherwise, we denote by
$i_0$ the minimum integer $j$ such that $b$ dominates $w_j$ and
$w_j\in W.$  If there exists $i>i_0$ such that $w_ib\in A(\T)$ and
$w_i\in W,$ then $bw_{i_0}w_i$ would be a triangle containing two
vertices in $W$ and one in $B,$ which is not possible.  Moreover, by
definition of $i_0,$ the set $W_{[1,i_0-1]}$ dominates $b.$

So, we have $b\in B_{i_0},$ and more generally $(B_i)_{i=1,\dots
  ,t_0,\infty}$ forms a partition of $B.$ To conclude, we just denote
by $S^\npair$ the subset of indices $i$ in $[t_0] \cup \{\infty\}$ with $B_i\neq \emptyset.$
\end{proof}

\begin{figure}[!ht]
  \centering
  \includegraphics[scale=.75]{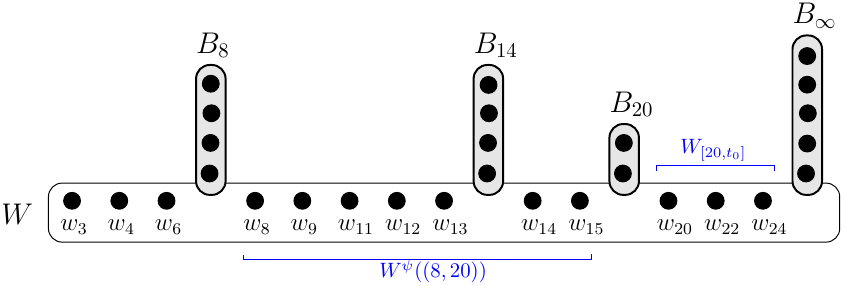}
\caption{A nice pair $\npair=(W,B)$ and its bucket decomposition with
  $S^\npair=\{8,14,20,\infty \}.$ The arcs inside $W$ and from $W$ to
  the $B_i$'s all go from left to right, in increasing order of the
  indices.  A bucket interval $I_0=(8,20)$ and its associated set of
   vertices $W^\npair(I).$  Let $I_1=(14,\infty),$ $I_2 = (20,\infty)$
  and $\I=\{I_0,I_1,I_2\}.$  We have $I_0 \prec I_1$ and $I_2
  \subsetint I_1.$}
\label{@experimentation}
\end{figure}

Informally, in a bucket decomposition, all vertices from the same
bucket have the same neigborhood in $W.$  We will use several times
the next observation following from the definition of a nice pair, and
asserting that any triangle inside $W \cup B$ contains at
least two vertices from the buckets. We refer again the reader to \autoref{@experimentation} for the next definition.

\begin{definition}[bucket interval]
Let us consider given a nice pair $\npair=(W,B).$  We say that
$I=(l_I,r_I)$ is a \emph{bucket interval of $\npair$} if $\{l_I,r_I\}
\subseteq S^\npair$ and $l_I < r_I.$  We denote by $\BI(\npair)$ (or
$\BI(W,B)$) the set of bucket intervals of $\npair.$

  Given two bucket intervals $I_1$ and $I_2$ 
  \begin{itemize}
    \item we say that $I_1 \subsetint I_2$ iff $[l_{I_1},r_{I_1}] \subseteq [l_{I_2},r_{I_2}],$ and
    \item $I_1 \prec I_2$ iff $l_{I_1} < l_{I_2} < r_{I_1} < r_{I_2}.$
    \item Moreover, we define $I_1 \unionint I_2 = \left(\min\left\{l_{I_1},l_{I_2}\right\},\max\left\{ r_{I_1},r_{I_2}\right\}\right),$ and
    \item if $I_1 \prec I_2,$ we set $I_1 \interint I_2 = (l_{I_2},r_{I_1})$
  \end{itemize}

  Given any bucket interval $I \in \BI(\npair),$ we write
  \begin{itemize}
  \item $B^\npair(I) = [l_I,r_I] \cap S^\npair$ the indices of buckets in $I$
  \item $W^\npair(I) = W_{[l_I,r_I[}$
  \item $P^\npair(I) = \{ I' \in \BI(\npair) \mid I' \subsetint I \}$
  \item for any $X \subseteq \BI(\npair),$ we write $P^\npair_X(I) = P^\npair(I) \cap X.$
  \end{itemize}
When the nice pair is clear from context, we will drop the $^\npair$
from the previous notations.
\end{definition}

\begin{observation}\label{@significativo}
Given a nice pair $(W,B),$ for any triangle $\Delta$ in $\T[W \cup
  B],$ either $\Delta \subseteq V(B),$ or there exists a bucket
interval $I \in \BI(\npair)$ such that
  \begin{itemize}
    \item $|\Delta \cap B_{l(I)}| = |\Delta \cap B_{r(I)}|= 1,$
      $\Delta \cap W \subseteq W(I),$ and
    \item for any $v \in W(I),$ $(\Delta \cap B) \cup \{v\}$ is still
      a triangle.
  \end{itemize}
  \end{observation}

Recall that in a given round of the rainbow matching technique, if we do not find a rainbow matching then
  we find a set of vertices $U \subseteq W$ and $X^C \subseteq C^0 \setminus B$ (having some properties corresponding to the hypothesis of \autoref{lemma_addtonicepair}),   that we have to add to $B$, while preserving in particular that what we obtain is still a nice pair. 
  This motivates the next lemma.

\begin{lemma}\label{lemma_addtonicepair}
Let $(W,B)$ be a nice pair.  Let $U \subseteq W$ and $X^C \subseteq
C^0 \setminus B$ such that for any triangle $\Delta$ of $\T[W\cup
  X^C]$ we have $|\Delta \cap (W \setminus U)| \le 1.$  Let $W' = W
\setminus U$ and $B' = B \cup U \cup X^C.$  Then, $(W',B')$ is a nice
pair.
\end{lemma}

\begin{proof}
As $X^C \subseteq C^0$ and $W \cap C^0 = \emptyset,$ we have that $W'
\cap B' = \emptyset.$  Let us now suppose, towards a contradiction,
that there exists a triangle $\Delta$ of $\T[W' \cup B']$ such that
$|\Delta \cap W'| \ge 2.$  This implies that $|\Delta \cap W'| =2
$ and $|\Delta \cap B'| \ge 1,$ as we cannot have $|\Delta \cap W'| =
3,$ because $W' \subset W^0$ and $\T[W^0]$ is acyclic.  Let $\Delta
\cap B'=\{u\}$ and $\Delta \cap W' = \{v,w\}.$ If $u \in X^C,$ then
this contradicts the hypothesis that for any triangle $\Delta$ of
$\T[W \cup X^C]$ we have $|\Delta \cap (W \setminus U)| \le 1.$  If $u
\in U,$ then $\Delta \subseteq W,$ contradicting the fact that
$\T[W^0]$ is acyclic.  Finally, if $u \in B,$ then this contradicts
the fact that $(W,B)$ is a nice pair as, for any triangle $\Delta$ of
$\T[W \cup B],$ we should have $|\Delta \cap W| \le 1.$
\end{proof}

As buckets will correspond to vertices that we want to keep in our
kernel, we need to control their size, motivating the following
definition, which is illustrated in Figure~\ref{@secularizations}.

\begin{figure}[!ht]
  \centering
  \includegraphics[scale=.75]{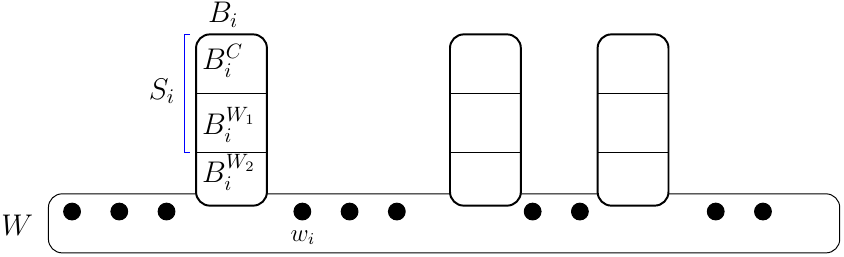}
\caption{A bucket partition}
\label{@secularizations}
\end{figure}

\begin{definition}[Bucket partition of a nice pair]
\label{buk_par_pair}
Let $f$ be a function from $\mathbb{N}$ to $\mathbb{R}^+$ and
$\npair=(W,B)$ be a nice pair.  Let $B^C = B \cap C^0,$ and $B^W = B
\cap W^0.$ For any $i \in S^\npair,$ let $B_i^C = B_i \cap B^C.$  For
any partition $(B^{W_1},B^{W_2})$ of $B^W$ and any $i \in S^\npair,$
let $B_i^{W_1} = B_i \cap B^{W_1},$ $B_i^{W_2} = B_i \cap B^{W_2}$ and
let $S_i =B_i^C \cup B_i^{W_1}.$

A \emph{bucket partition of a nice pair $\npair=(W,B)$} is a partition
$(B^{W_1},B^{W_2})$ of $B^W$ such that
  \begin{enumerate}
    \item for any $i \in S^\npair,$ $|S_i| \ge 1$
  \item $|B^{W_1}| \le 10 |B^C|.$
  \end{enumerate}
      We say that a bucket partition has local size $f$ if for any $i \in S^\npair,$ $|B_i^{W_2}| \le f(|S_i|).$
\end{definition}
Notice that the size condition $|B^{W_1}| \le 10 |B^C|$ required in
the bucket partition is a ``global'' constraint on the size of the
$B^{W_1},$ while the condition size $|B_i^{W_2}| \le f(|S_i|)$
required by the local size function is a ``local'' constraint on every
bucket. We introduced this local notion of size as a way to control more precisely the number of vertices added to $B$ (and thus to the kernel output),
  which will be critical when typically adjacent buckets are merged using the \textit{add} operation o  Definition~\autoref{@impoverished}. However,
when the kernel will find a rainbow matching and stop, the only role of these local sizes will be to upper bound the total size of $B$.

The object that our kernel will manipulate is a partial decomposition,
as defined below.

\begin{definition}[Partial decomposition]
We say that a tuple $(W,B,C,B^{W_1},B^{W_2})$ is a \emph{partial
decomposition} if
  \begin{enumerate}
  \item\label{@proclamations} there is a partition $V(\T) = W \cup B \cup C$
  \item\label{@inextinguishably} $C \subseteq C^0$
  \item\label{@inconvinientes} $(W,B)$ is a nice pair
    \item\label{@congregating} $(B^{W_1},B^{W_2})$ is a bucket partition of $(W,B)$
    
  \end{enumerate}
  We say that a partial decomposition has local size $f$ if
  $(B^{W_1},B^{W_2})$ has local size $f.$\\
  \stf{Moreover, we will say that a partial decomposition is
    \emph{clean} if it satisfies the following extra condition:
  \begin{enumerate}
  \item\label{@cleaness-TPT} for every vertex $c\in C$ the tournament
    $T[W\cup \{c\}]$ contains a triangle.
  \end{enumerate} }
\end{definition}

\stf{Here again, we have a simple process, called {\it the cleaning
    phase}, to obtain a clean partial decomposition from any partial
  decomposition.}

\stf{
\begin{lemma}[Cleaning Lemma]
  \label{lem:clean-TPT}
  Let $(W,B,C,B^{W_1},B^{W_2})$ be a partial decomposition and $X$ be
  set the of vertices $x$ of $C$ such that $T[W\cup \{x\}]$ does not
  contain any triangle. Then
  $(W,B\cup X,C\setminus X, B^{W_1},B^{W_2})$ is a clean partial
  decomposition of $T$, with the same local size than
  $(W,B,C,B^{W_1},B^{W_2})$.
\end{lemma}
}

\stf{
\begin{proof}
  Let us first check that $(W,B\cup X,C\setminus X, B^{W_1},B^{W_2})$
  satisfies the requirements of a partial
  decomposition. Properties~\ref{@proclamations}
  and~\ref{@inextinguishably} are clearly satisfy, and by choice of
  $X$, no vertex of $X$ is contained in a triangle with two vertices
  of $W$, so $\npair'=(W,B\cup X)$ is also a nice pair and
  property~\ref{@inconvinientes} is satisfies. Notice that vertices of
  $X$ can create new buckets or be added in existing buckets of
  $\npair$, but in all cases we have
  $S^\npair \subseteq S^{\npair'}$. Finally, in the buckets partition
  of $\npair'$, the vertices of $X$ will be added to
  $B^C$. That is, with the notations of Definition~\ref{buk_par_pair},
  we have $(B\cup X)^W=B^W$, $(B\cup X)_i^{W_j} = B_i^{W_j}$ for $i\in
  S^\npair$ and $j=1,2$ and $(B\cup X)_i^{W_j} = \emptyset$ for $i\in
  S^{\npair'}\setminus S^\npair$ and
  $j=1,2$. It is then straightforward to check that Properties~{\it
    1.} and~{\it 2.}, as well as the local size requirement, from
  Definition~\ref{buk_par_pair} still hold for
  $\npair'$. In all, we can conclude that $(W,B\cup X,C\setminus X,
  B^{W_1},B^{W_2})$ is a clean partial decomposition with the same
  local size than $(W,B,C, B^{W_1},B^{W_2})$.
\end{proof}
}

\subsection{Intervals: demand definition and basic properties} 
In this section we introduce the notion of \emph{demand} for a partial
decomposition.  Informally, a demand is a set of bucket intervals with
a value attached to each interval.  The value $\textsf{val}(I)$
attached to interval $I$ depends on the contents of the buckets $\{B_i \mid i \in B(I)\}$ contained in $I$.

The notion of demand will be used in \autoref{@dialectician} to define the auxiliary edge-colored multigraph necessary in our approach.
  In particular, given a $\subsetint$-minimal interval $I$ (implying that buckets $B_{l_I}$ and $B_{r_I}$ are consecutive),
  it will be important that $\textsf{val}(I)$ upper bounds the size of packing $\P$ where all triangles in $\P$ have one vertex in $B_{l_I}$, one in
  $B_{r_I}$, and one in $W(I)$. Indeed, $\textsf{val}(I)$ will correspond to the number of vertices that we want to find in $W(I)$ in the rainbow matching we look for.
  Then, if we indeed find such a rainbow matching, and thus a set $Q_I \subseteq W(I)$ with $|Q_I|=\textsf{val}(I)$, we must be able to repack such a packing $\P$ (corresponding to a part of an optimal solution) into a packing $\P'$ of same size with $V(\P') \subseteq B_{l_I} \cup B_{r_I} \cup Q_I$ by changing vertices used in $W(I)$ to take instead vertices in $Q_I$.

The following definition is illustrated Figure~\ref{@biancheggiar}.

\begin{definition}[Block partition]
We call a set $\I^\textrm{max}\subseteq \BI(\npair)$ \emph{proper} if
there do not exist $I_1,I_2\in\I^\textrm{max}$ such that $I_1 \neq
I_2$ and $I_1 \subsetint I_2.$
      
Given a proper set $\I^\textrm{max}\subseteq \BI(\npair)$ of bucket
intervals we define the \emph{block partition} of $\I^\textrm{max},$
denoted $\{Z_1,\dots,Z_t\},$ as follows.  Let us order
$\I^\textrm{max}$ according to the left points, meaning that
$\I^\textrm{max}=\{I_1,\dots,I_x\},$ where $l_i < l_{i+1}$ (and $r_i <
r_{i+1}$) for any $i.$
      
We find the largest $x_1$ such that $I_1 \prec I_2 \prec \dots \prec
I_{x_1}$ and define $Z_1 = \{I_i \mid x_0+1 \le i \le x_1\}$ where
$x_0=0.$  Then, we find the largest $x_2$ such that $I_{x_1+1} \prec
I_{x_1+2} \prec \dots \prec I_{x_2},$ and define $Z_2 = \{I_{i} \mid
x_1+1\le i \le x_2\}.$  We continue until we have a partition
$\{Z_\ell, \ell \in [t]\}$ of $\I^\textrm{max}$ (implying
$x_t=x$)\footnote{Please note that our relation $\prec$ is not
transitive and thus $x_i\neq x_{i+1}.$}.
      
Notice that, for any $\ell \in [t-1],$ $I_{x_l} \nprec I_{x_l+1},$
implying that $r_{I_{x_l}} \le l_{I_{x_l+1}}.$  We also define the
\emph{block intervals} as $I^*_\ell = \unionint_{I \in Z_\ell}I$ for
any $\ell \in [t].$
\end{definition}
\begin{figure}[!ht]
  \centering
  \includegraphics[scale=.75]{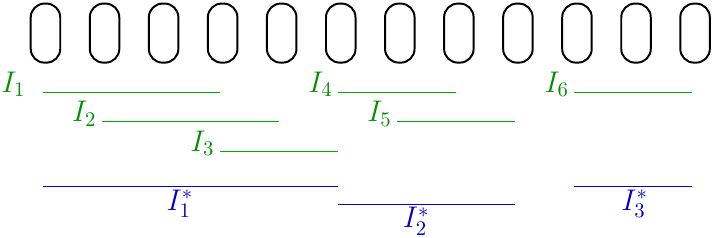}
\caption{A block partition of $\I^\textrm{max}=\{I_1,\dots,I_6\}.$ We have $Z_1 = \{I_1,I_2,I_3\},$ $Z_2 = \{I_4,I_5\}$ and $Z_3 = \{I_6\}.$}
\label{@biancheggiar}
\end{figure}

The next lemma follows from the previous definition.

\begin{lemma}\label{lemma_block}
Let $\npair$ be a nice pair, $\J \subseteq \BI(\npair)$ a set of
bucket intervals, and let $\J^\textrm{max}\subseteq \J$ be the subset
of $\subsetint$-wise maximal bucket intervals of $\J.$  Let $\{Z_\ell,
\ell \in [t]\}$ and $\{I^*_\ell,\ell \in [t]\}$ be the block partition
and block intervals of $\J^\textrm{max}$ (which is well defined as
$\J^\textrm{max}$ is proper).  Then, we have $\J=\bigcup_{\ell \in
  [t]}\P_{\J}(I^*_\ell)$
  \end{lemma}
\begin{proof} Indeed, we have the following.
   \begin{eqnarray*}
   \J = \bigcup_{I \in \J^\textrm{max}}P_{\J}(I) = \bigcup_{\ell \in
     [t]}\bigcup_{I \in Z_\ell}P_{\J}(I) = \bigcup_{\ell \in
     [t]}\P_{\J}(I^*_\ell)
 \end{eqnarray*}
\end{proof}

The two next definitions introduce the notion of demand for a partial
decomposition of a nice pair and are illustrated in
Figure~\ref{@gaspilleront}.

\begin{definition}
  Let $(W,B,C,B^{W_1},B^{W_2})$ be a partial decomposition and $I$ a
  bucket interval.  We define
  \begin{itemize}
  \item $\Sigma(I) = \sum_{i \in B(I)}|S_i|$
    \item $m(I) = \sum_{i \in B(I)}|B_i|-|B_{i_0}|$ where $i_0 \in \textrm{ argmax}_{i \in B(I)}|B(i)|$
    \item $\mu(I) = \min\big\{\Sigma(I),m(I)\big\}$
    \item $t(I) = \begin{cases} 
          \ssigma & \text{if } \Sigma(I) < m(I) \\
          \msynt & \text{if } \Sigma(I) > m(I) \\
          \eqsynt & \text{otherwise.}
        \end{cases}$
  \end{itemize}
\end{definition}

\begin{figure}[!ht]
  \begin{minipage}{6cm}
    \scalebox{0.75}{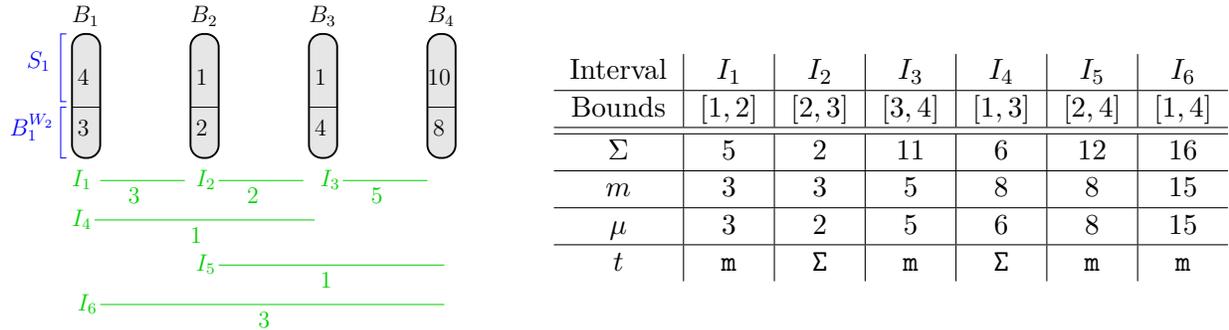}
\end{minipage}
  \hspace*{1cm}
\begin{minipage}{6cm}
\begin{tabular}{c|c|c|c|c|c|c}
    Interval & $I_1$ & $I_2$ & $I_3$ & $I_4$ & $I_5$ & $I_6$\\
    \hline
    Bounds & $[1,2]$ & $[2,3]$ & $[3,4]$ & $[1,3]$ & $[2,4]$ & $[1,4]$\\
    \hline
    \hline
    $\Sigma$ & 5 & 2 & 11 & 6 & 12 & 16\\
    \hline
    $m$ & 3 & 3 & 5 & 8 & 8 & 15\\
    \hline
    $\mu$ & 3 & 2 & 5 & 6 & 8 & 15\\
    \hline
    $t$ & \msynt & \ssigma & \msynt & \ssigma & \msynt  & \msynt\\
\end{tabular}
\end{minipage}
  \caption{Example of a demand $(\I,\textsf{val}).$ On the left, the
    buckets are depicted and the cardinalities of $S_1$ and
    $B_1^{W_2}$ are indicated. Below, stand the intervals with their
    value. On the right, for every interval, the values $\Sigma,$ $m,$
    $\mu$ and $t$ are computed.}
\label{@gaspilleront}
\end{figure} 

To get an insight on why we used the previous definition of $\mu$, one can check that the property announced at the begining of this section holds:
 let us consider a $\subsetint$-minimal interval $I$ (implying that buckets $B_{l_I}$ and $B_{r_I}$ are consecutive),
  and check that $\textsf{val}(I)$ upper bounds the size of packing $\P$ where all triangles in $\P$ have one vertex in $B_{l_I}$, one in
  $B_{r_I}$, and one in $W(I)$.
  On one hand, as any such triangle consumes one vertex in $B_{l_I}$ and one in $B_{r_I}$, we get $|\P| \le \min(|B_{l_I}|,|B_{r_I}|) = m(I)$.
  On the other hand, as any triangle in $\P$ cannot use only vertices in $W^0$, it must uses at least one vertex from $S_{l_I} \cup S_{r_I}$, implying that $|\P| \le \ssigma(I)$.
We discuss at the end of \autoref{@intentionally} why taking simply $\mu(I)=\ssigma(I)$ or $\mu(I)=m(I)$ would not be sufficient to get a kernel in $\O(k^\delta)$ for any $\delta > 1$.

\begin{definition}[Demand]
Let $(W,B,C,B^{W_1},B^{W_2})$ be a partial decomposition of $\npair.$
Let us define the following polynomial algorithm which, given
$(W,B,C,B^{W_1},B^{W_2}),$ computes a set $\I$ of bucket intervals of
$\npair,$ and a value $\textsf{val}(I)$ for any $I \in \I.$
  \begin{itemize}
  \item start with $X=\emptyset$
  \item for any $L \in [2,|S^\npair|]$ do
    \begin{itemize}
    \item for any bucket interval $I$ of $\npair$ where $|B(I)|=L$:
    \item let $\textsf{val}(P_X(I)) = \sum_{I \in
      P_X(I)}\textsf{val}(I).$ Notice that at this stage $I \notin
      P_X(I),$ and in particular we have $\textsf{val}(P_X(I))=0$ if
      $L=2.$
    \item if $\textsf{val}(P_X(I)) \le \mu(I)$ then define $\textsf{val}(I)=\mu(I)-\textsf{val}(P_X(I)),$ and add $I$ to $X$ (and thus, in this case, $I \in P_X(I)$)
    \end{itemize}
  \item let $\I = X$
  \item return $(\I,\textsf{val})$
  \end{itemize}

We call the pair $(\I,\textsf{val})$ the \emph{demand for
$(W,B,C,B^{W_1},B^{W_2})$}.  We denote $$\I_{>0}=\{I \in \I \mid
\textsf{val}(I) > 0\}.$$ For any $X \subseteq \I,$ we also denote
$\textsf{val}(X)=\sum_{I \in X}\textsf{val}(I).$
\end{definition}

Observe that in the example of Figure~\ref{@gaspilleront} we have
$\I_{>0}=\I=\BI,$ but it may be the case that $\I_{>0} \subsetneq \I
\subsetneq \BI.$  Let us now provide properties on demands.

\begin{lemma}\label{@personalmente}
Let $(W,B,C,B^{W_1},B^{W_2})$ be a partial decomposition and
$(\I,\textsf{val})$ be its demand.  Then, for any $I \in \BI(W,B)$:
  \begin{enumerate}
  \item $\textsf{val}(P_{\I}(I)) = \textsf{val}(P_{\I_{>0}}(I))$
  \item $\textsf{val}(P_{\I}(I)) \ge \mu(I)$
  \item $I \in \I \Leftrightarrow \textsf{val}(P_\I(I)) = \mu(I)
    \Leftrightarrow \textsf{val}(P_\I(I) \setminus \{I\}) \le \mu(I) $
  \item $I \in \I_{>0} \Leftrightarrow \textsf{val}(P_\I(I)\setminus
    \{I\}) < \mu(I)$
  \item if $t(I)= \ssigma,$ then for any $i_0 \in \textrm{ argmax}_{i
    \in B(I)}|B(i)|,$ we have $|S_{i_0}| < \sum_{i \in B(I) \setminus
    \{i_0\}} |B_i^{W_2}|$
  \item if there exists $i_0 \in \textrm{ argmax}_{i \in B(I)}|B(i)|$ such that
    $|S_{i_0}| < \sum_{i \in B(I) \setminus \{i_0\}} |B_i^{W_2}|,$ then we have $t(I)= \ssigma$
  \item if $t(I)= \msynt,$ then for any $i_0 \in \textrm{ argmax}_{i
    \in B(I)}|B(i)|,$ we have $|S_{i_0}| > \sum_{i \in B(I) \setminus \{i_0\}}
    |B_i^{W_2}|$
  \end{enumerate}
\end{lemma}

Before proving Lemma~\ref{@personalmente}, notice that for every interval $I$
in Figure~\ref{@gaspilleront}, computing its demand leads to $I\in
\I$ and $\textsf{val}(P_\I(I)) = \mu(I),$ illustrating Property~{\bf
  3} above.


\begin{proof}
Let $I \in \BI(W,B).$ In what follows, we consider the iteration where
the algorithm considers $I,$ and let $X$ denote the variable of the
algorithm at the line where it computes
$\textsf{val}(P_X(I)).$ Observe first that, no matter whether the algorithm
decides to add $I$ in $X$ or not, we have $P_X(I) = P_\I(I) \setminus
\{I\}.$  In particular, we have the following equivalences :
  \begin{itemize}
  \item $I \notin \I$ $\Leftrightarrow P_\I(I)=P_X(I) \Leftrightarrow \textsf{val}(P_X(I)) > \mu(I).$
  \item $I \in \I$ $\Leftrightarrow  P_\I(I)=P_X(I) \cup \{I\} \Leftrightarrow \textsf{val}(P_X(I)) \le \mu(I)$ . 
  \end{itemize}
  Now, we can prove Properties {\bf 1} to {\bf 7}.
  
Property $\mathbf{1}$ is immediate.
  
Property $\mathbf{2}$  holds in the case where $I \notin \I.$  In
case $I \in \I,$ we have
$\textsf{val}(P_{\I}(I))=\textsf{val}(P_{X}(I))+\textsf{val}(I)=\mu(I)$
and the property holds.

The first equivalence of Property {\bf 3} follows from the proof of
Property {\bf 2} above.  As $\textsf{val}(P_X(I)) =
\textsf{val}(P_\I(I) \setminus \{I\}),$ we get the second part of
Property $\mathbf{3}$ from the seminal observation.
  
Property $\mathbf{4}$ can be rewritten as $I \in \I_{>0}
\Leftrightarrow \textsf{val}(P_X(I)) < \mu(I).$  Since $\I \in \I
\Leftrightarrow \textsf{val}(I)=\mu(I)-P_X(\I),$ we get
$\textsf{val}(I)>0 \Leftrightarrow \textsf{val}(P_X(I)) < \mu(I),$
implying Property $\mathbf{4}.$

For Properties $\mathbf{5}$ and $\mathbf{6},$ let $i_0 \in \textrm{
  argmax}_{i \in B(I)}|B(i)|.$
  \begin{eqnarray*}
    t(I)=\ssigma & \Leftrightarrow & \sum_{i \in B(I)}|S_i| < \sum_{i \in B(I) \setminus \{i_0\}}|B_i| \\
    & \Leftrightarrow & \sum_{i \in B(I)}|S_i| < \sum_{i \in B(I) \setminus \{i_0\}}(|S_i|+|B^{W_2}_i|)\\
    & \Leftrightarrow & |S_{i_0}| < \sum_{i \in B(I) \setminus \{i_0\}}|B^{W_2}_i|
  \end{eqnarray*}

Property $\mathbf{7}$ is obtained by reversing the above inequalities.
\end{proof}

\subsection{The small total demand property}\label{@bedingungslos}
The objective of the section is only to prove
\autoref{lemma_valblock}.  To that end, we prove that the
``intersection property'' of \autoref{lemma_intersection} implies the
``union property'' of \autoref{lemma_union}, which finally implies
\autoref{lemma_valblock}.

\begin{lemma}[The intersection property]\label{lemma_intersection}
  Let $(W,B,C,B^{W_1},B^{W_2})$ be a partial decomposition and $(\I,\textsf{val})$ be its demand.
  Let $I_1, I_2$ in $\I_{>0}$ such that $I_{1} \prec I_{2}.$
  Then, $t(I_1 \interint I_2)=\ssigma.$
\end{lemma}

An example of the intersection property can be seen in
Figure~\ref{@gaspilleront} where $t(I_4 \interint
I_5)=t(I_2)=\ssigma.$

\begin{proof}
We refer the reader to \autoref{@glorification} for an illustration of
the notations used in this proof.
  \begin{figure}[!ht]
  \centering
  \includegraphics[scale=.75]{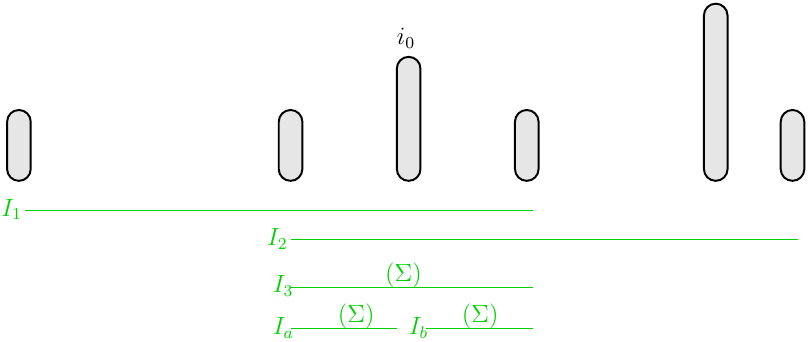}
  \caption{Notations for the intersection property proof.}
\label{@glorification}
  \end{figure}

  Let $I_3 = I_1 \interint I_2,$ and observe that $I_3=[l_{I_2},r_{I_1}].$
  Let $i_0 \in \textrm{ argmax}_{i \in B(I_3)}|B_i|.$
  Let $I_a = (l_{I_2},i_0)$ and $I_b = (i_0,r_{I_1}).$ Observe that $I_a$ or $I_b$ may not be a bucket interval (when $i_0=l_{I_2}$ or $i_0=r_{I_1}$),
  but  at least one of them is a bucket interval.  Let $X = \{I_a,I_b\} \cap \BI.$

  We now show that, for any $I \in X,$ $t(I)=\ssigma.$  Without loss
  of generality, assume $I_a \in X,$ which implies $l_{I_2}<i_0.$
  Suppose, towards a contradiction, that $t(I_a) \in
  \{\eqsynt,\msynt\}.$  Observe that $i_0 \in \textrm{ argmax}_{i \in
    B(I_a)}|B_i|,$ implying $\mu(I_a)=\sum_{i \in
    B(I_a)}|B_i|-|B_{i_0}|.$  Let $I_{a'}=(i_0,r_{I_2}).$  Observe
  that $I_{a'}$ is a bucket interval as $i_0 \le r_{I_1} < r_{I_2}$
  (as we have $I_1 \prec I_2$).  Note that $P_\I(I_a) \cap
  P_\I(I_{a'}) = \emptyset,$ $P_\I(I_a) \subseteq P_\I(I_2) \setminus
  \{I_2\},$ and $P_\I(I_{a'}) \subseteq P_\I(I_2) \setminus \{I_2\}.$
  By combining these observations, we obtain
  \begin{eqnarray*}
    \textsf{val}(P_\I(I_2) \setminus \{I_2\}) &\ge& \textsf{val}(P_\I(I_a))+ \textsf{val}(P_\I(I_{a'})) \\
    & \ge & \mu(I_a)+\mu(I_{a'}) \mbox{ by Lemma~\ref{@personalmente} Property {\bf 2}}
  \end{eqnarray*}
  Let us now prove that $\mu(I_a)+\mu(I_{a'}) \ge \mu(I_2).$
  This will imply $\textsf{val}(P_\I(I_2) \setminus \{I_2\}) \ge \mu(I_2),$ which contradicts \autoref{@personalmente} Property~{\bf 4} as $\textsf{val}(I_2) > 0.$ 

  We now distinguish two cases according to $t(I_{a'}).$\\
  
  \noindent{\bf Case 1}: $t(I_{a'}) \in \{\ssigma,\eqsynt\}.$\\
  In this case, $\mu(I_{a'})=\sum_{i \in B(I_{a'})}|S_i|$ and we have the following.
\begin{eqnarray*}
\mu(I_a)+\mu(I_{a'}) & = & \sum_{i\in B(I_a)}|B_i|-|B_{i_0}|+\sum_{i
  \in B(I_{a'})}|S_i|\\ & = &\sum_{i \in B(I_a) \setminus
  \{i_0\}}|B_i|+\sum_{i \in B(I_{a'})}|S_i| \\ & \ge & \sum_{i \in
  B(I_a) \setminus \{i_0\}}|S_i|+\sum_{i \in B(I_{a'})}|S_i| \\ & \ge
& \Sigma(I_2) \\ & \ge & \mu(I_2)
  \end{eqnarray*}
  
  \noindent{\bf Case 2}: $t(I_{a'}) = \msynt.$\\  In this case,
  $\mu(I_{a'})=\sum_{i \in B(I_{a'})}|B_i|-|B_{i_1}|,$ where $i_1 \in
  \textrm{ argmax}_{i \in B(I_{a'})}|B_i|.$  Notice that $|B_{i_1}|
  \ge |B_{i_0}|.$  Let $i_2 \in \textrm{ argmax}_{i \in
    B(I_2)}|B_i|.$ Observe that $|B_{i_2}|=|B_{i_1}|,$ and that in the
  two cases ($i_1 = i_0$ or $i_1 \neq i_0$) we have
  $\mu(I_a)+\mu(I_{a'}) = \sum_{i \in B(I_2)}|B_i|-|B_{i_1}| = m(I_2)
  \ge \mu(I_2).$  This concludes the proof of our claim, and we now
  assume that for any $I \in X,$ $t(I)=\ssigma.$
  
 Let us now prove that $t(I_3)=\ssigma.$  If $i_0=l_{I_2}$ or
 $i_0=r_{I_1},$ then $X=\{I_3\},$ and we are done. Suppose now
 $l_{I_2} < i_0 < r_{I_1}.$  Observe that $i_0 \in \textrm{ argmax}_{i
   \in B(I_3)}|B_i|$ implies $i_0 \in \textrm{ argmax}_{i \in
   B(I_a)}|B_i|.$  As $t(I_a)=\ssigma$ and as $i_0 \in \textrm{
   argmax}_{i \in B(I_a)}|B_i|,$ Lemma~\ref{@personalmente} Property {\bf 5}
 implies $|S_{i_0}| < \sum_{i \in B(I_a) \setminus \{i_0\}}
 |B_i^{W_2}| \le \sum_{i \in B(I_3) \setminus \{i_0\}} |B_i^{W_2}|.$
 By \autoref{@personalmente} Property {\bf 6}, this implies $t(I_3)=\ssigma.$
\end{proof}

\begin{lemma}[The union property]\label{lemma_union}
  Let $(W,B,C,B^{W_1},B^{W_2})$ be a partial decomposition and
  $(\I,\textsf{val})$ be its demand.  {Let $\mathcal{Z} \subseteq
    \I_{>0},$ where $\mathcal{Z}= \{I_1,\dots,I_x\}$} for some $x \ge
  1,$ and where $I_{i} \prec I_{i+1}$ for any $i=1,\dots ,x-1.$  Then,
  we have $\unionint_{I \in {\mathcal Z}}I \in \I.$
\end{lemma}
\begin{proof}
  We refer the reader to Figure~\ref{@entstandenen} for notations.
  \begin{figure}[!ht]
  \centering
  \includegraphics[scale=.71]{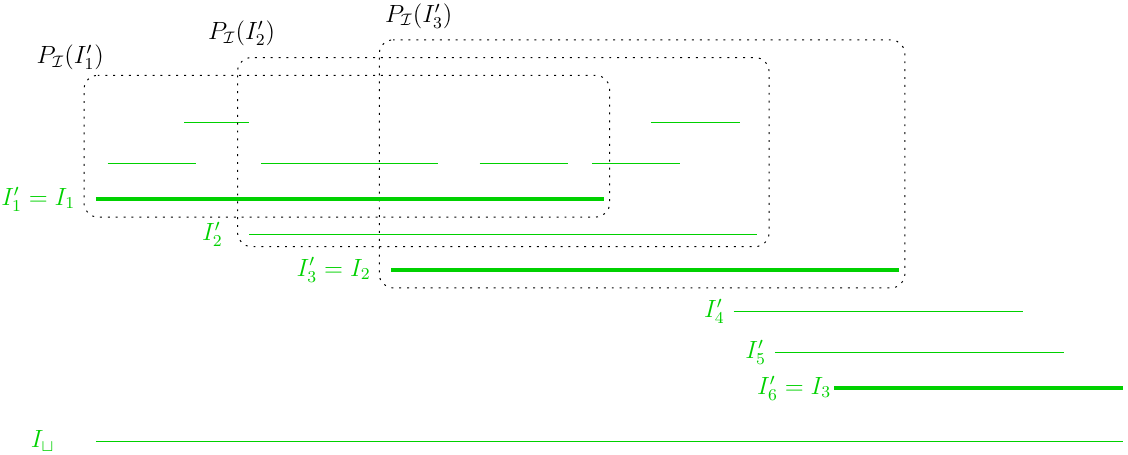}
  \caption{Interval of $\mathcal Z$ are represented in bold.}
\label{@entstandenen}
  \end{figure}

  Let $I_\unionint = \unionint_{I \in Z}I.$  Let
  also $${\mathcal{Z}^\textrm{max}}=\{I\in \I_{>0}\mid I \subseteq
  I_\unionint, I \neq I_\unionint \mbox{ and $I$ is $\subsetint$-wise
    maximal} \}.$$ Let $\mathcal{Z}^\textrm{max}=\{I'_1,\dots,I'_y\},$
  where intervals are ordered according to their left endpoint.  If
  $x=1,$ the lemma follows immediately, and thus we may assume $x \ge
  2$ which implies $y \ge 2$ as well.

Let us prove that we even have $I'_{i} \prec I'_{i+1},$ for any $i.$
Let $i \in [y-1]$ and let $j_0$ be the minimum $j$ such that $r_{I'_i}
< r_{I_j}.$  Such a $j_0$ exists, as $i \le y-1,$ and thus $r_{I'_i} <
r(I_\unionint).$  Moreover, $j_0 > 1$ as otherwise we would have $I'_i
\subsetintstrict I_{j_0},$ a contradiction to the fact that $I'_i$ is
$\subsetint$-wise maximal.  Notice first that $I_{j_0} \nsubsetint
I'_\ell,$ for any $\ell \le i.$  As $I_{j_0-1} \prec I_{j_0},$ we get
$l_{I_{j_0}} < r_{I_{j_0-1}}$ and, moreover, by the definition of
$j_0,$ we get $r_{I_{j_0-1}} \le r_{I'_i},$ implying $l_{I_{j_0}} <
r_{I'_i}.$  Thus, if, towards a contradiction, we do not have
$I_i'\prec I_{i+1}'$ then we had $r_{I'_i} \le l_{I'_{i+1}} $ and we
would have $l_{I_{j_0}} < l_{I'_{i+1}},$ implying that $I_{j_0}
\nsubsetint I'_\ell$ for any $\ell \ge i+1.$  Thus, there would be no
$I \in \mathcal{Z}^\textrm{max}$ such that $I_{j_0} \subsetint I,$
which is a contradiction as $I_{j_0} \neq I_\unionint.$  This
concludes the proof that $I'_{i} \prec I'_{i+1},$ for any $i.$

According to Lemma~\ref{@personalmente}, $I_\unionint \in \I \Leftrightarrow
\textsf{val}(P_\I(I_\unionint) \setminus \{I_\unionint\}) \le
\mu(I_\unionint),$ and thus our objective is to prove that
$\textsf{val}(P_\I(I_\unionint) \setminus \{I_\unionint\}) \le
\mu(I_\unionint).$

By definition of $\mathcal{Z}^\textrm{max},$ we get $P_\I(I_\unionint)
\setminus \{I_\unionint\} = \bigcup_{\ell \in [y]}P_\I(I'_\ell)$ (see
Figure~\ref{@entstandenen}).  Thus, $\textsf{val}(P_\I(I_\unionint)
\setminus \{I_\unionint\}) = \textsf{val}(\bigcup_{\ell \in
  [y]}P_\I(I'_\ell)).$  Let us now prove that
$\textsf{val}(\bigcup_{\ell \in [y]}P_\I(I'_\ell))=s,$
where $$s=\sum_{\ell \in [y]}\textsf{val}(P_\I(I'_\ell))-\sum_{\ell
  \in [y-1]}\textsf{val}(P_\I(I'_\ell) \cap P_\I(I'_{\ell+1})).$$ Let
$I \in \bigcup_{\ell \in [y]}P_\I(I'_\ell)$ and let us prove that
$\textsf{val}(I)$ appears exactly one time in $s.$  Let $\ell_1$
(resp. $\ell_2$) the minimum (resp. maximum) value such that $I \in
\P_I(I'_\ell).$  Observe that $\textsf{val}(I)$ appears
$\ell_2-\ell_1+1$ times in $\sum_{\ell \in
  [y]}\textsf{val}(P_\I(I'_\ell)),$ and $\ell_2-\ell_1$ times in
$\sum_{\ell \in [y-1]}\textsf{val}(P_\I(I'_\ell) \cap
P_\I(I'_{\ell+1})),$ and thus exactly $1$ time in $s.$  Moreover, as
for any $\ell \in [y-1],$ $P_\I(I'_\ell) \cap P_\I(I'_{\ell+1}) =
P_\I(I'_\ell \cap I'_{\ell+1}),$ we can even rewrite
  $$s=\sum_{\ell \in [y]}\textsf{val}(P_\I(I'_\ell))-\sum_{\ell \in [y-1]}\textsf{val}(P_\I(I'_\ell \cap I'_{\ell+1})).$$
  It now remains to prove that $s \le \mu(I_\unionint).$
   \begin{eqnarray*}
    s & = & \sum_{\ell \in [y]}\textsf{val}(P_\I(I'_\ell))-\sum_{\ell
      \in [y-1]}\textsf{val}(P_\I(I'_\ell \cap I'_{\ell+1})) \\ & = &
    \sum_{\ell \in [y]}\mu(I'_\ell)-\sum_{\ell \in
      [y-1]}\textsf{val}(P_\I(I'_\ell \cap I'_{\ell+1})), \mbox{ by
      Lemma~\ref{@personalmente} Property {\bf 3} as $I'_\ell \in \I$} \\ &
    \le & \sum_{\ell \in [y]}\mu(I'_\ell) -\sum_{\ell \in
      [y-1]}\mu(I'_\ell \cap I'_{\ell+1}), \mbox{ by
      Lemma~\ref{@personalmente} Property {\bf 2}} \\ & \le & \sum_{\ell \in
      [y]}\mu(I'_\ell) -\sum_{\ell \in [y-1]}\sum_{i \in B(I'_\ell
      \cap I'_{\ell+1})}|S_i|, \mbox{ by
      Lemma~\ref{lemma_intersection} and as we have } I_\ell'\in
    \I_{>0} \mbox{ for all } \ell \in [y] \\
  \end{eqnarray*}

   Let us now distinguish two cases according to $t(I_\unionint).$\\
   
   {\bf Case 1:} $t(I_\unionint) \in \{\ssigma,\eqsynt\}.$\\ In this
   case, $\mu(I_\unionint)=\sum_{i \in B(I_\unionint)}|S_i|,$ and we
   have the following.
   \begin{eqnarray*}
s & \le & \sum_{\ell \in [y]}\sum_{i \in B(I'_\ell)}|S_i| -\sum_{\ell
  \in [y-1]}\sum_{i \in B(I'_\ell \cap I'_{\ell+1})}|S_i|, \mbox{ by
  definition of $\mu$}\\ & = & \sum_{i \in B(\bigcup_{\ell \in
    [y]}I'_l)}|S_i|, \mbox{ as for any $i,$ the term $|S_i|$ appears
  exactly once} \\ & = & \Sigma(I_\unionint) \\ & \le &
\mu(I_\unionint) \mbox{ as we are in case 1}
  \end{eqnarray*}

{\bf Case 2:} $t(I_\unionint) = \msynt.$\\ Let $i_0 \in \textrm{
  argmax}_{i \in B(I_\unionint)}|B_i|.$  In this case,
$\mu(I_\unionint)=\sum_{i \in B(I_\unionint)}|B_i|-|B_{i_0}|,$ and
according to \autoref{@personalmente} Property {\bf 7}, we have $|S_{i_0}| >
\sum_{i \in B(I_\unionint) \setminus \{i_0\}} |B_i^{W_2}|.$
  
  Let us prove the following property: Let $B^* = \Big( B(I'_1)
  \setminus \bigcup_{2 \le \ell \le y}B(I'_\ell) \Big) \cup \Big(
  B(I'_y) \setminus \bigcup_{1 \le \ell \le y-1}B(I'_\ell)\Big),$ then
  $i_0 \in B^*.$  This property means that $i_0$ is in one of the two
  ``extreme'' sides of the union.  For example, in
  \autoref{@gaspilleront}, where $I'_1=I_4, I'_2=I_5,$ {and}
  $I_\unionint = I_6,$ we have that $t(I_\unionint)=\msynt,$ $i_0=4,$
  and $i_0 \in B(I'_2) \setminus B(I'_1).$
  
  Observe that, as $I'_\ell \prec I'_{\ell+1}$ for any $\ell,$ it
  holds that, for all $i \in B(I_\unionint) \setminus B^*,$ there
  exists $\ell$ such that $i \in B(I'_\ell) \cap B(I'_{\ell+1}).$
  Assume, towards a contradiction, that $i_0 \notin B^*,$ implying
  that there exists $\ell_0$ such that $i_0 \in B(I'_{\ell_0}) \cap
  B(I'_{\ell_0+1}).$  Therefore $i_0 \in B(I'_\interint),$ where
  $I'_\interint = I'_{\ell_0} \interint I'_{\ell_0+1}.$  By
  \autoref{lemma_intersection}, as $I'_{\ell_0}$ and $I'_{\ell_0}$ are
  in $\I_{>0}$ and $I'_{\ell_0} \prec I'_{\ell_0+1},$ we get
  $t(I'_\interint)=\ssigma.$  By \autoref{@personalmente} Property {\bf 5},
  and as $i_0 \in \textrm{ argmax}_{i \in B(I'_\interint)},$ we get
  $|S_{i_0}| < \sum_{i \in B(I'_\interint) \setminus \{i_0\}}
  |B_i^{W_2}| < \sum_{i \in B(I_\unionint) \setminus \{i_0\}}
  |B_i^{W_2}|,$ a contradiction.
  
  Hence the property holds and, without loss of generality, we may
  assume $i_0 \in B(I'_1) \setminus \bigcup_{2 \le \ell \le
    y}B(I'_\ell).$  Note that the case where $i_0 \in B(I'_y)
  \setminus \bigcup_{1 \le \ell \le y-1}B(I'_\ell)$ is symmetric.

  Observe that $i_0 \in \textrm{ argmax}_{i \in B(I'_1)} |B(i)|$ which
  implies $m(I'_1)=\sum_{i \in B(I'_1) \setminus \{i_0\}}|B_i|.$
  
As $i_0 \in B(I'_1) \setminus \bigcup_{2 \le \ell \le y}B(I'_\ell),$
we have $i_0 \in [l(I'_1),l(I'_2)[,$ and thus
      $$m(I'_1)=\sum_{i \in [l(I'_1),l(I'_2)[ \setminus \{i_0\}}|B_i|+\sum_{i \in [l(I'_2),r(I'_1)]}(|S_i|+|B^{W_2}_i|).$$  
{We conclude that}
           \begin{eqnarray*}
  s & \le &  \sum_{\ell \in [y]}\mu(I'_\ell) -\sum_{\ell \in [y-1]}\sum_{i \in B(I'_\ell \cap I'_{\ell+1})}|S_i| \\
  & \le & m(I'_1)+\sum_{2 \le \ell \le y}\Sigma(I'_\ell)  -\sum_{\ell \in [y-1]}\sum_{i \in B(I'_\ell \cap I'_{\ell+1})}|S_i| \\
    & = & \sum_{i \in [l(I'_1),l(I'_2)[ \setminus \{i_0\}}|B_i|+\sum_{i \in [l(I'_2),r(I'_1)]}(|S_i|+|B^{W_2}_i|) +\\
        & & ~~~~~~~~~~~~~~~~~~~~~~~~~~~~~~~~~~~~~~~~~~~~~~~~~~~~~~~~  \sum_{2 \le \ell \le y}\sum_{i \in B(I'_\ell)}|S_i|-\sum_{\ell \in [y-1]}\sum_{i \in B(I'_\ell \cap I'_{\ell+1})}|S_i| \\ 
        & = & \sum_{i \in [l(I'_1),l(I'_2)[ \setminus \{i_0\}}|B_i|+\sum_{i \in [l(I'_2),r(I'_1)]}|B^{W_2}_i|+\sum_{1 \le \ell \le y}\sum_{i \in B(I^*_\ell)}|S_i|-\sum_{\ell \in [y-1]}\sum_{i \in B(I^*_\ell \cap I^*_{\ell+1})}|S_i| \\
            &~ & ~\mbox{ by defining $I^*_1 = [l(I'_2),r(I'_1)]$ and $I^*_\ell = I'_\ell,$ for any $2 \le \ell \le y$}\\
          & = & \sum_{i \in [l(I'_1),l(I'_2)[ \setminus \{i_0\}}|B_i|+\sum_{i \in [l(I'_2),r(I'_1)]}|B^{W_2}_i|+ \sum_{i \in B(\bigcup_{\ell \in [y]}I^*_l)}|S_i| \\
              & ~& ~\mbox{ as, for any $i,$ each term $|S_i|$ appears exactly once} \\
              & = & \sum_{i \in B(I'_1) \setminus \{i_0\}}|B_i|+\sum_{i \in B(I_\unionint) \setminus B(I'_1)}|S_i| \\
              & \le & \sum_{i \in B(I'_1) \setminus \{i_0\}}|B_i|+\sum_{i \in B(I_\unionint) \setminus B(I'_1)}|B_i| \\
              & = & m(I_\unionint) \\
                & \le & \mu(I_\unionint), \mbox{ as we are in Case 2.}
                       \end{eqnarray*}
\end{proof}

The following lemma will be used in both cases of our kernelization algorithm  (rainbow matching or small vertex cover)  in order to prove
  that what is added in the kernel is small.
\begin{lemma}[The small total demand property]\label{lemma_valblock}
  Let $(W,B,C,B^{W_1},B^{W_2})$ be a partial decomposition and  
$(\I,\textsf{val})$ be its demand.  Let $\J \subseteq \I_{>0},$ and
let $\J^\textrm{max}\subseteq \J$ be the subset of $\subsetint$-wise
maximal bucket intervals of $\J.$  Let $\{Z_\ell, \ell \in [t]\}$ and
$\{I^*_\ell,\ell \in [t]\}$ be the block partition and block intervals
of $\J^\textrm{max}.$  Then, we have $\textsf{val}(\J)\le \sum_{\ell
  \in [t]}\mu(I^*_\ell).$
\end{lemma}

\begin{proof}
According to \autoref{lemma_block} we have, $\J = \bigcup_{\ell \in
  [t]}\P_{\J}(I^*_\ell)$ and thus, $\textsf{val}(\J) \le \sum_{\ell
  \in [t]}\textsf{val}(\P_{\J}(I^*_\ell)).$  Let $\ell \in [t].$  As
$Z_\ell = \{I_{x_{l-1}+1},\dots,I_{x_l}\},$ where $I_{x_{l-1}+i} \prec
I_{x_{l-1}+i+1}$ for any $i,$ and all $I \in Z_\ell$ satisfy
$\textsf{val}(I) > 0,$ \autoref{lemma_union} implies that $I^*_\ell
\in \I.$  Thus, \autoref{@personalmente} Property {\bf 3} implies that
$\textsf{val}(P_\J(I^*_\ell)) \le \mu(I^*_\ell),$ implying in turn
that $\textsf{val}(\J) \le \sum_{\ell \in [t]}\mu(I^*_\ell).$
\end{proof}

\subsection{Auxiliary graph and bucket allocation}\label{@dialectician}
In this section we first define an auxiliary graph based
  on a demand. Then, we prove in \autoref{lemma_safebucketT} that if a
  rainbow matching gives us in particular a matching for all the
  colored loops of this auxiliary graph, which correspond to what we
  call here a \emph{bucket all-ocation}, then we can repack any
  triangle-packing contained in $W \cup B$ into this bucket
  allocation.
  
We refer the reader to \autoref{@reciprocation} for an illustration
of the situation.
  \begin{figure}[!ht]
  \centering
  \includegraphics[scale=.90]{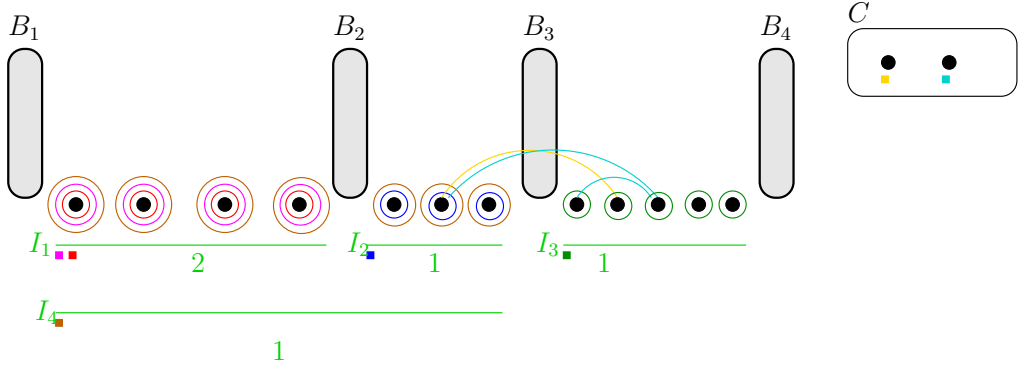}
  \caption{Example of an auxiliary graph. Squares indicate the colors associated to each interval or element of $C.$
    For example, $D_{I_1}$ is composed of colors pink and red. Notice that there is a rainbow matching.}
\label{@reciprocation}
  \end{figure}

\begin{definition}[Auxiliary graph]
Let $(W,B,C,B^{W_1},B^{W_2})$ be a partial decomposition and
$(\I,\textsf{val})$ be its demand.  Let also $p=|C|+\sum_{i\in
  \I_{>0}}\textsf{val}(I).$  We next define the $p$-edge-colored mutigraph
$(G, χ)\langle W,B,C,B^{W_1},B^{W_2}\rangle$ where the vertex set of
$G$ is $W$ and the edges of $G,$ as well as their colors, are defined
as follows.  We start with $E_1=E_2=\emptyset.$ Then, for any $c \in
C,$ for any $\{u,w\} \subseteq W$ such that $\{u,w,c\}$ is a triangle,
we add the edge $e=\{u,w\}$ to $E_1$ and we set $χ(e)=c.$  Moreover,
for any $I \in \I_{>0},$ we define a set of new colors $D_I$ where
$|D_I|=\textsf{val}(I).$  Then, for any $u \in D_I$ and $v \in W(I),$
we add to $E_2$ edge $e=\{v\}$ and we set $χ(e)=u.$  Finally, we
denote $\bigcup_{I \in \I_{>0}}D_I$ by $D.$
\end{definition}

\stf{Notice that if the partial decomposition
  $(W,B,C,B^{W_1},B^{W_2})$ is clean, then for any vertex $c\in C$,
  there exists an edge with color $c$ in $E_1$.} Notice also that the
value of $p$ is a polynomial function of the size $|V(T)|$ of the
tournament.  In fact, $p=\O(|V(T)|^3).$

\begin{definition}[Bucket allocation]
Let $(W,B,C,B^{W_1},B^{W_2})$ be a partial decomposition and
$(\I,\textsf{val})$ be its demand.  A \emph{bucket allocation} for
$(W,B,C,B^{W_1},B^{W_2})$ is a set $\Q=\{Q_I,I\in \I_{>0}\}$ such that
\begin{itemize}
	\item for any $I \in \I_{>0},$ $Q_I \subseteq W(I)$
	\item for any $I \in \I_{>0},$ $|Q_I| = \textsf{val}(I)$
	\item for any $I_1$ and $I_2$ in $\I_{>0},$ $Q_{I_1} \cap Q_{I_2} = \emptyset$
\end{itemize}
We also denote $V(\Q)=\bigcup_{I \in \I_{>0}}Q_I.$
\end{definition}

The next lemma shows that a bucket allocation allows to repack any
packing triangle inside $\T[W \cup B].$

\begin{lemma}[Safeness of a bucket allocation for packing]\label{lemma_safebucketT}
Let $(W,B,C,B^{W_1},B^{W_2})$ be a partial decomposition and let $\Q$
be a bucket allocation for $(W,B,C,B^{W_1},B^{W_2}).$  Let also $\P$
be a triangle-packing with $V(\P) \subseteq W \cup B.$  Then, there
exists a triangle-packing $\P'$ such that
\begin{itemize}
	\item $V(\P') \subseteq V(\Q) \cup B$ and 
	\item $|\P'|=|\P|$
\end{itemize}
\end{lemma}
\begin{proof}
According to \autoref{@significativo}, we can partition $\P=\P_1 \cup
\P_2$ such that, for any $\Delta \in \P_1,$ $\Delta \subseteq B$ and,
moreover, for any $\Delta \in \P_2,$ there exists $I \in \BI(W,B)$
such that $|\Delta \cap B_{l_I}|=|\Delta \cap B_{r_I}|=1$
and $\Delta \cap W \subseteq W_I.$  Let $A(\P_2) = \{\Delta \cap B,
\Delta \in \P_2\}$ be the set of backward arcs used by triangles in
$\P_2.$

\stf{Let us now define an auxiliary graph $H$ that will help us to describe
how to repack $\P.$  The vertex set of $H$ is $B$ and for any bucket
interval $I=[l_I,r_I]$ of $(W,B),$ we define $E_I=\{\{u,v\}\mid u \in
B_{l_I} \setminus B^{W_2}_{l_I} \mbox{ or } u \in B_{r_I} \setminus
B^{W_2}_{r_I}\},$ and we set $E(H) = \bigcup_{\BI(W,B)}E_I.$  Informally,
for any $i \neq i',$ $H[B_i \cup B_{i'}]$ is the complete bipartite
graph $K_{|B_i|,|B'_i|}$ where we remove all edges between $B^{W_2}_i$ and
$B^{W_2}_{i'}.$  For any $e \in E(H),$ we denote $I(e)$ the unique bucket
interval $I$ such that $e \in E_I.$}
     
\stf{Observe that $A(\P_2) \subseteq E(H)$ as for any $\{u,v\} \in
A(\P_2),$ where $\{u,v\}=\Delta \cap B,$ there exists $I \in \BI(W,B)$
such that $u \in B_{l(I)},$ $v \in B_{r(I)}.$ Moreover, we cannot have
$u \in B^{W_2}_{l(I)}$ and $v \in B^{W_2}_{r(I)},$ as this would imply
$\Delta \subseteq W_0.$ Finally, as $\P_2$ is a packing, $A(\P_2)$ is
even a matching in $H.$}

\paragraph*{Property $\Pi_1.$}
Let us prove the following property $\Pi_1$: 
     
\begin{quote}{\sl Property $\Pi_1$}: 
For any $I \in \BI(W,B),$ $\textsf{MaxM}(H(I)) \le \mu(I),$ where
$\textsf{MaxM}$ denotes the size of a maximum matching, and $H(I)=
H[\bigcup_{i \in B(I)}B_i].$
\end{quote}     

\noindent \textsl{Proof of $\Pi_1$}: Firstly, as any edge in $H(I)$
uses at least one vertex from a set $S_i$ for $i \in B(I),$ we deduce
that $\textsf{MaxM}(H(I)) \le \Sigma(I).$  Let $i_0 = \textrm{
  argmax}_{i \in B(I)}|B_i|.$  Moreover, as for any graph $G'$ and any
independent set $X' \subseteq V(G')$ we have $\textsf{MaxM}(G') \le
|V(G')|-|X'|,$ and as $B_{i_0}$ is an independent set in $H,$ we get
$\textsf{MaxM}(H(I)) \le m(I).$ This implies that $\textsf{MaxM}(H(I))
\le \mu(I).$

\paragraph*{Property $\Pi_2.$}
We proceed by proving property $\Pi_2$ which allows us to associate,
in a well defined way, a vertex in $V(\Q)$ to any arc in $A(\P_2).$
An example of this property can be found in
\autoref{@mountaintops}.
     
 \begin{figure}[!ht]
  \centering
  \includegraphics[scale=.90]{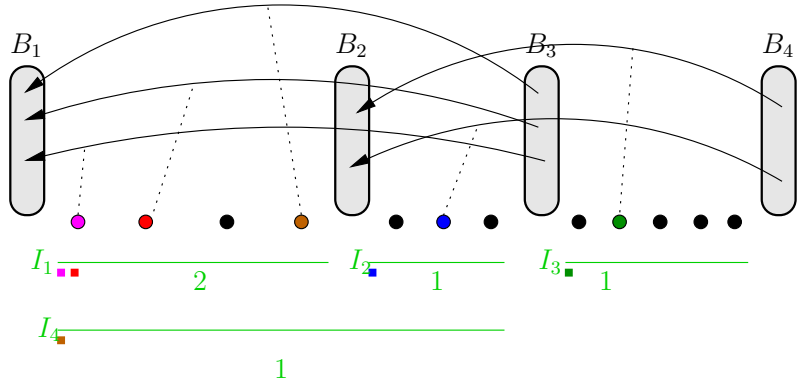}
  \caption{Example of property $\Pi_2.$ The bucket allocation is
    represented by the non-black vertices. Function $f,$ which
    associates injectively to each backward arc of the initial packing
    a vertex of the bucket allocation, is depicted using dashed
    lines.}
\label{@mountaintops}
  \end{figure}

\begin{quote}
\textsl{Property $\Pi_2$}:  For any matching $M$ in $H,$ there is a function $f$ from $M$ to $V(\Q)$ such that $f$ is injective and, for any $e \in M,$ $f(e) \in W(I(e)) \cap V(\Q).$
\end{quote}
     
\noindent \textsl{Proof of $\Pi_2$}: Let $U$ be a bipartite graph,
where $V(U)=(M,Q),$ and for any $e \in M,$ $N(e)=W(I(e)) \cap V(\Q).$
What remains is the proof that there is a perfect matching $f$ (which
associates to each $e \in M$ a vertex $f(e) \in N(e)$) in $U$ which
saturates $M.$  According to Hall's Theorem, it is sufficient to prove
that, for any $M' \subseteq M,$ $|N(M')| \ge |M'|.$

Let $M' \subseteq M.$  Observe that $\bigcup_{e \in
  M'}[l_{I(e)},r_{I(e)}]$ is a union of $t \ge 1$ disjoint intervals
denoted $\{I'_1,\dots,I'_t\},$ and that $M'$ can be partitioned into
$\bigcup_{i \in [t]}M'_t$, such that for any $i \in [t],$ $\bigcup_{e
  \in M'_i}[l_{I(e)},r_{I(e)}]$ $=I'_i.$

As an example, consider a matching $M' = \{e_1,e_2,e_3,e_4,e_5\}$
where
$I(e_1)=[1,10],I(e_2)=[3,4],I(e_3)=[10,11],I(e_4)=[20,22],I(e_5)=[21,23],$
$\bigcup_{e \in M'}[l_{I(e)},r_{I(e)}]=\{[1,11],[20,23]\}$ and $M'_1 =
\{e_1,e_2,e_3\}.$
     
Let $i \in [t].$  Observe that, for any $v \in W(I'_i),$ there exists
$e \in M'_i$ such that $v \in W(I(e)).$  Thus, for any vertex $v \in
\bigcup_{I' \in P_{\I_{>0}}(I'_i)}Q_{I'},$ as $v$ belongs to
$W(I'_i),$ and as there exists $e \in M'_i$ such that $v \in W(I(e)),$
we get $v \in W(I(e)) \cap Q = N(e).$ This implies $\bigcup_{I' \in
  P_{\I_{>0}}(I'_i)}Q_{I'} \subseteq N(M'_i).$  As $Q$ is a bucket
allocation for $(W,B,C,B^{W_1},B^{W_2}),$ we know that the $Q_{I'}$
are disjoint and $|Q_{I'}|=\textsf{val}(I'),$ implying that
$$|\bigcup_{I' \in P_{\I_{>0}}(I'_i)}Q_{I'}|=\sum_{I' \in P_{\I_{>0}}(I'_i)}|Q_{I'}|=\sum_{I' \in P_{\I_{>0}}(I'_i)}\textsf{val}(I')=\textsf{val}(P_{\I_{>0}}(I'_i)).$$
According to \autoref{@personalmente}, $\textsf{val}(P_{\I_{>0}}(I'_i)) \ge \mu(I'_i).$ According to Property $\Pi_1,$ $\mu(I'_i) \ge \textsf{MaxM}(H(I'_i)).$
This implies that $|N(M'_i)| \ge \textsf{MaxM}(H(I'_i)).$
Since $M'$ is a matching in $H,$ $M'_i$ is a matching in $H(I'_i).$
This implies $\textsf{MaxM}(H(I'_i)) \ge |M'_i|$ and therefore $|N(M'_i)| \ge |M'_i|.$
To conclude, let us now consider the partition $\{N(M'_i)\mid i \in [t]\}$ of $N(M').$
As the $N(M'_i)$s are vertex-disjoint, we get
$$|N(M')|=\sum_{i \in [t]}|N(M'_i)| \ge \sum_{i \in [t]}\textsf{MaxM}(H(I'_i)).$$
As $M'$ is a matching in $\bigcup_{i \in [t]}H(I'_i),$ and as $\sum_{i \in [t]}\textsf{MaxM}(H(I'_i)) \ge \textsf{MaxM}(\bigcup_{i \in [t]}H(I'_i)),$ we conclude that $|N(M')| \ge |M'|.$\medskip

We can now conclude the proof of Lemma.  As $A(\P_2)$ is a matching in
$H,$ by Property $\Pi_2,$ there is a function $f$ from $A(\P_2)$ to
$V(\Q)$ such that $f$ is injective and, for any $e \in A(\P_2),$ $f(e)
\in W(I(e)) \cap V(\Q).$  By \autoref{@significativo}, $e \cup f(e)$ is
still a triangle.  Thus, we define $\P'_2 = \{e \cup f(e), e \in
A(\P_2)\}$ and $\P' = \P_1 \cup \P'_2.$ As $V(\P'_2) \cap B = V(\P_2)
\cap B$ and as $f$ is injective, $\P'$ is still a triangle-packing,
and $|\P'|=|\P|.$  Moreover, as $V(\P_1) \subseteq B$ and $V(\P'_2)
\subseteq B \cup V(\Q),$ we get $V(\P') \subseteq B \cup V(\Q).$
\end{proof}


\stf{Finally, the last lemma of the section indicates how to build a
  feedback vertex set of $T[W\cup B]$ with the help of a bucket
  allocation. This will be usefull for the kernel of \FVST designed
  Section~\ref{@linearKernelIPHS}. We enounce this lemma here, as the
  notations and the techniques used in its proof are similar, though
  easier, than in the previous lemma.}

\stf{
\begin{lemma}[Safeness of a bucket allocation for hitting]\label{lemma_safebucketT_hit}
Let $(W,B,C,B^{W_1},B^{W_2})$ be a partial decomposition and let $\Q$
be a bucket allocation for $(W,B,C,B^{W_1},B^{W_2}).$ Let also $X$ be
a feedback vertex set of $T[V(\Q)\cup B]$.  Then, there exists $X'$ a
feedback vertex set of $T[W\cup B]$ with $|X'|\le|X|$.
\end{lemma}
}
\begin{proof}\stf{
We partition $X$ into two sets $X_B=X\cap B$ and $X_{\Q}=X\cap V(\Q)$.
And we define $H$ being the graph on vertex set $B$ and with edge set
$\{uv\ :\ uv\in A(T), u\in B_i\setminus X_B, v\in B_j\setminus X_B
\textrm{ with } j<i\}$. Informally, from $T[B]$ we only keep in $H$
the backward arcs between the $B_i$'s not incident with vertices of
$X_B$. Notice that in particular, there is no edge in $H$ between any
two $B_i^{W_2}$, as there where originally part of $W_0$.}

\paragraph*{Property $\Pi_1'.$}
\stf{Let us prove the following property $\Pi_1'$: 
\begin{quote}{\sl Property $\Pi_1'$}: 
For any $I \in \BI(W,B),$ $\textsf{MinVC}(H(I)) \le \mu(I),$ where
$\textsf{MinVC}$ denotes the size of a minimum vertex cover, and
$H(I)= H[\bigcup_{i \in B(I)}B_i].$
\end{quote}     
}
\noindent \textsl{Proof of $\Pi_1'$}: \stf{As every edge of $H$ is incident
with a vertex of one $S_i$ (which is $B_i\setminus B_i^{W_2}$), the
set $\cup_i S_i$ is a vertex cover of $H$, and then
$\textsf{MinVC}(H(I)) \le \Sigma(I)$. Similarly, let $i_0 = \textrm{
  argmax}_{i \in B(I)}|B_i|$, the set $(\cup_i B_i) \setminus B_{i_0}$
is clearly a vertex cover of $H$. Then we get $\textsf{MinVC}(H(I))
\le m(I)$, and finally conclude that $\textsf{MinVC}(H(I)) \le \mu(I)$.\\}

\stf{Now, let us consider an arc $uv$ of $T$ such that $uv\in E(H)$. We
have $u\in B_i$ and $v\in B_j$ for some $j<i$ and denote by $I$ the
bucket interval $[j,i]$. By Observation~\ref{@significativo}, $uvw$ is
a triangle for every vertex $w$ of $W(I)$. In particular, for every
$w\in W(I)\cap V(\Q)$ we must have $w\in X_{\Q}$.\\ So, denote by
$\cal J$ the set of bucket intervals whose extremities contains the
end of an edge of $H$ (formally, $[j,i] \in {\cal J}$ if there exists
$uv\in E(H)$ with $u\in B_i$ and $v\in B_j$). Let $J^\star$ be a block
of the block partition of $\cal J$, meaning that $J^\star$ is a
non-empty, inclusion-wise minimal bucket interval with the property
that for every $J\in {\cal J}$, either $J\cap J^\star=\emptyset$ or
$J\subseteq J^\star$. By the previous remark remark, we know that
$X_{\Q}$ contains at least $|W(J^\star)\cap V(\Q)|$ vertices in
$W(J^\star)$. As $\Q$ is a bucket allocation, we have in particular
that $|W(J^\star)\cap V(\Q)|=|\cup \{ Q_J\ :\ Q_J\in {\cal Q} \textrm{
  and } J\subseteq J^\star\}|=\textsf{val}(P_{\I}(J^\star))\ge
\mu(J^\star)$ by Lemma~\ref{@personalmente}.}

\stf{Finally, by Property $Pi_1'$, there exists a vertex cover
$X_{J^\star}'$ of $H(J^\star)$ containing at most $\mu(J^\star)\le
|W(J^\star)\cap X_\Q|$ vertices. Considering all the block intervals
$J_1^\star,\dots ,J_p^\star$ of $\cal J$, we obtain a vertex cover
$X'_\Q=\cup_{j} X_{J_j^\star}'$ of $H$ such that $|X'_\Q|\le
|X_\Q|$. In particular, for every arc $uv$ of $T$ with $u\in B_i$ and
$v\in B_j$ with $j<i$ we have $u\in X'_\Q\cup X_B$ or $u\in X'_\Q\cup
X_B$.\\ We now define $X'= X'_\Q\cup X_B$. We have $|X'|\le |X|$ and
$X'$ is a feedback vertex set of $T[W\cup B]$. Indeed, let $\Delta$ by
a triangle of $T[W\cup B]$. If $V(\Delta)\subseteq B$, then
$V(\Delta)\cap X_B\neq \emptyset$. Otherwise, as $(W,B)$ is a nice
pair, there exist buckets $B_i$ and $B_j$ with $j<i$ such that
$\Delta=uvw$ with $u\in B_i$, $v\in B_j$ and $w\in W_{[j,i]}$. But in
this case we must have $u\in X'_\Q$ or $v\in X'_\Q$.}
\end{proof}


\subsection{Operations on buckets}\label{@intentionally}
Our kernelization algorithm has only one rule that is applied
exhaustively.  Each application of the this {rule} (except the last
one that finds a rainbow matching) triggers an ``add'' operation,
defined below, where we add new vertices to the buckets. In this
section, we define the two variants of the add operation and prove
that, given a partial decomposition, these operations output another
partial decomposition.

\begin{definition}\label{@impoverished}
Let $(W,B,C,B^{W_1},B^{W_2})$ be a partial decomposition.  Let $U
\subseteq W$ and $X^C \subseteq C.$  For any $j \in
\{1,2\},$ we define $\textsf{add}_j(W,B,B^{W_1},B^{W_2},U \cup X^C)$
as $(W',B',C',B^{'W_1},B^{'W_2})$ where
\begin{itemize}
	\item $W' = W \setminus U$
	\item $B' = B \cup U \cup X^C$
	\item $C' = C \setminus X^C$
	\item $B^{'W_j} = B^{W_j} \cup U$ and $B^{'W_{3-j}} = B^{W_{3-j}}.$ 
\end{itemize}
\end{definition}

\begin{figure}[!ht]
  \centering
  \includegraphics[scale=.90]{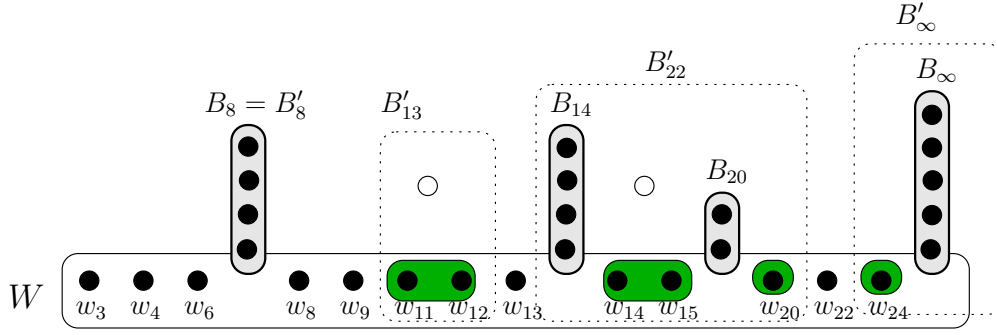}
  \caption{An application of $\textsf{add}_1$ (or
    $\textsf{add}_2$). Vertices of $U$ are {inside} the green
    territories and the two vertices of $X^C$ are the two non-filled
    vertices.  We have $S^\npair = \{8,14,20,\infty\}$ and
    $S^{\npair'}=\{8,13,22,\infty\}.$}
\label{@saltimbanque}
  \end{figure}

\begin{lemma}\label{lemma_add1}
Let $(W,B,C,B^{W_1},B^{W_2})$ be a partial decomposition of local size
$f.$ Suppose $f(x)=c x^\delta$ for some constants $c$ and $\delta,$
where $\delta > 1.$  Let $U \subseteq W$ and $X^C \subseteq C$ such
that $|U| \le 10 |X^C|$ and such that, for any triangle $\Delta$ of
$\T[W \cup X^C],$ we have that $|\Delta \cap (W \setminus U)| \le 1.$
Then, $\textsf{add}_1(W,B,B^{W_1},B^{W_2},U \cup X^C)$ is a partial
decomposition of local size $f.$
\end{lemma}
\begin{proof}
Let $(W',B',C',B^{'W_1},B^{'W_2}) =
\textsf{add}_1(W,B,B^{W_1},B^{W_2},U \cup X^C).$  We begin by  proving
that $(W',B',C',B^{'W_1},B^{'W_2})$ is a partial decomposition.
Properties~\ref{@proclamations} and~\ref{@inextinguishably} are clear.  As, by assumption,
for any triangle $\Delta$ of $\T[W \cup X^C],$ we have that $|\Delta
\cap (W \setminus U)| \le 1,$ \autoref{lemma_addtonicepair} implies
that $\npair'=(W',B')$ is a nice pair, and thus we obtain
Property~\ref{@inconvinientes}.

Let us now prove Property~\ref{@congregating}.  We start by proving the
following structural property of the $B'_i$ (see \autoref{@saltimbanque}).
  \begin{quote}
  For any $i \in S^{\npair'},$ $\exists U' \subseteq U,$ $X' \subseteq
  X^C$ and $S' \subseteq S^\npair$ such that $B'_i = \bigcup_{i \in
    S'}B_i \cup X' \cup U'.$
\end{quote}
  We point out that we could even prove a stronger property (for
  example that buckets in $S'$ are consecutive), but we don't need it
  in the remainder of the proof.
  
Let $i \in \npair'.$  According to \autoref{@accompanying},
$B'_i$ is the set of vertices $v$ of $B'$ such that all arcs between
$W'_{[1,i-1]}$ and $v$ are oriented from $W'_{[1,i-1]}$ to $v$ and all
arcs between $W'_{[i,t_0]}$ and $v$ are oriented from $v$ to
$W'_{[i,t_0]}.$ Assume that there exists $j \in \npair$ and $v_1 \in
B_j$ such that $v_1 \in B'_i.$  Now, observe that for any $v_2 \in
B_j$ and $w \in W',$ $v_1w \in A(\T)$ iff $v_2w \in A(\T),$ and the
same holds for $wv_1$ and $wv_2.$  This implies that $v_2 \in B'_i$
and thus $B_j \subseteq B'_i.$  As all the vertices we add to $B$ are
from $U \cup X^C,$ this proves the above structural property.

Next we prove that $(B^{'W_1},B^{'W_2})$ is a bucket partition of
$(W',B').$  The fact that $(B^{'W_1},B^{'W_2})$ is a partition of
$B^{'W}$ is clear, as $B^{'W}=B^W \cup U,$ $B^{'W_1}=B^{W_1} \cup U,$
and $B^{'W_2}=B^{W_2}.$  Let $i \in S^{\npair'}.$ By the previous
structural property, $B'_i = \bigcup_{i \in S'}B_i \cup X' \cup U'.$
If $S' \neq \emptyset,$ then $|S'_i| \ge \bigcup_{i \in S'}|S_i| \ge
1.$  Otherwise, as $X' \cup U' \neq \emptyset,$ $U' \subseteq
B^{'W_1}$ and $X' \subseteq B^{'C}_i,$ we get $U' \cup X' \subseteq
S^{'}_i,$ implying $|S'_i| \ge 1.$  Let us prove that $|B^{'W_1}| \le
10 |B^{'C}|.$
\begin{eqnarray*}
    |B^{'W_1}|& = &|B^{W_1}|+|U| \\
    & \le & 10|B^C|+|U| \\
    & \le &10|B^C|+10 |X^C|\\
    & = & 10|B^{'C}|.
\end{eqnarray*}
This concludes the proof that $(B^{'W_1},B^{'W_2})$ is a bucket
partition of $(W',B').$

Let us finally prove that $(B^{'W_1},B^{'W_2})$ has local size $f,$
meaning that for any $i \in S^{\npair'},$ we have $|B_i^{'W_2}| \le
f(|S'_i|)=c|S_i'|^\delta.$  Let $i\in S^{\npair'}.$ By the previous
structural property, there exists $S'\subseteq S^{\npair}$ such that
$S'_i \supseteq \bigcup_{j \in S'}S_j.$ In particular, we have $|S'_i|
\ge \sum_{j \in S'}|S_j|$ as the $S_j$'s are disjoint.  Now, noticing
that $B'^{W_2}=B^{W_2},$ we have the following.
\begin{eqnarray*}
    |B_i^{'W_2}| & = & |\bigcup_{i \in S'}B_i \cap B^{'W_2}|+ |(X'
    \cup U') \cap B^{'W_2}| \\ & = & |\bigcup_{i \in S'}B_i \cap
    B^{W_2}|+ 0\\ & \le & \sum_{i \in S'}c|S_i|^\delta \\ & \le &
    c(\sum_{i \in S'}|S_i|)^\delta \\ & \le & c|S'_i|^\delta.
\end{eqnarray*}
\end{proof}

Now, we analyze the $\textsf{add}_2$ operation.

\begin{lemma}\label{lemma_add2}
Let $(W,B,C,B^{W_1},B^{W_2})$ be a partial decomposition of local size
$f.$ Suppose $f(x)=c x^\delta$ for some constants $c$ and $\delta,$
where $1<\delta ≤ 2$ and $c \ge \max(\frac{20}{(2^\delta-2)},(\frac{21}{\delta})^\frac{1}{\delta-1}).$
 Let $I$ be a
bucket interval and $U=W(I)$ such that $|U| \le 10 \mu(I).$  Then,
$\textsf{add}_2(W,B,B^{W_1},B^{W_2},U)$ is a partial decomposition of
local size $f.$
\end{lemma}

\begin{proof}
Let $(W',B',C',B^{'W_1},B^{'W_2}) =
\textsf{add}_2(W,B,B^{W_1},B^{W_2},U).$  As a first step we prove that
$(W',B',C',B^{'W_1},B^{'W_2})$ is a partial decomposition.
Properties~\ref{@proclamations} and~\ref{@inextinguishably} are clear.  By applying
\autoref{lemma_addtonicepair} (for $X^C = \emptyset$ here), we get
that $\npair'=(W',B')$ is a nice pair, and thus we obtain
Property~\ref{@inconvinientes}.

Let us now prove Property~\ref{@congregating}.  Let us first describe the
bucket decomposition in $\npair'.$  Informally, all buckets of $B(I)$
are merged, together with $U,$ into a new one ($B'_{r_I}$), and {all
  the other} buckets remain unchanged.  More formally, the following
three properties hold
  \begin{itemize}
    \item $S^{\npair'}=(S^\npair \setminus B(I)) \cup \{r_I\},$ 
    \item for any $i \in S^{\npair'} \setminus \{r_I\},$ it holds that 
    \begin{itemize}
    \item $B'_i = B_i,$ 
    \item   $B^C_i=B_i,$ and
    \item for any $j \in [2],$ $B^{'W_j}_i=B^{W_j}_i,$  
    \end{itemize}
    \item $B'_{r_I} = U \cup \bigcup_{i \in B(I)}B_i.$
  \end{itemize}

The next step is to prove that $(B^{'W_1},B^{'W_2})$ is a bucket
partition of $(W',B').$  The fact that $(B^{'W_1},B^{'W_2})$ is a
partition of $B^{'W}$ is clear as $B^{'W}=B^W \cup U,$
$B^{'W_1}=B^{W_1}$ and $B^{'W_2}=B^{W_2} \cup U.$  Let $i \in
S^{\npair'}.$  If $i \in S^{\npair'} \setminus \{r_I\},$ then, by the
previous property, we get $S'_i = S_i,$ implying $|S'_i| \ge 1,$ as
$(B^{W_1},B^{W_2})$ is a bucket partition of $(W,B).$  If $i = r_{I},$
then $B'_{r_I} \supseteq \bigcup_{i \in B(I)}B_i$ and, as
$B^{'W_1}=B^{W_1},$ $|S'_{r_I}| \ge \sum_{i \in B(I)}|S'_i| \ge |B(I)|
\ge 2.$  Moreover, we have $|B^{'W_1}| = |B^{W_1}| \le 10 |B^{C}| = 10
|B^{'C}|.$  This concludes the proof that $(B^{'W_1},B^{'W_2})$ is a
bucket partition of $(W',B').$
  
Let us finally prove that $(B^{'W_1},B^{'W_2})$ has local size $f,$
meaning that, for any $i \in S^{\npair'},$ $|B_i^{'W_2}| \le
f(|S'_i|).$  Recall that $(W',B')$ is a nice pair, implying, because
of Proposition~\ref{@accompanying}, that there is a unique bucket
decomposition $B'_i$ of $B'.$  Thus, for any $i \in S^{\npair'}
\setminus \{r_I\},$ $|B^{'W_2}_i| = |B^{W_2}_i| \le
f(|S_i|)=f(|S'_i|).$

Let $B(I)=\{i_1,\dots,i_L\}$, where $|S_{i_x}| \ge |S_{i_{x+1}}|,$ for
any $x$.  Notice that $f(|S'_{r_I}|) = c|S'_{r_I}|^\delta =
c(\sum_{\ell \in [1,L]}S_{i_\ell})^\delta $, and that $m(I) = \sum_{i
  \in B(I)}|B_i|-|B_{i_*}|$ where $i_* = \textrm{ argmax}_{i \in
  B(I)}|B(i)|$, implying that $m(I) \le \sum_{\ell \in
  [2,L]}|B_{i_\ell}| \le \sum_{\ell \in
  [2,L]}(|B_{i_\ell}^{W_2}|+|S_{i_\ell}|) \le \sum_{\ell \in
  [2,L]}(c|S_{i_\ell}|^\delta+|S_{i_\ell}|)$.  It remains to prove
that $|B^{'W_2}_{r_I}| - f(|S'_{r_I}|) \le 0$.

Observe that 
  \begin{eqnarray*}
    |B^{'W_2}_{r_I}| - f(|S'_{r_I}|) & = & \sum_{i \in B(I)}|B^{W_2}_i| + |U|- c(\sum_{\ell \in [1,L]}|S_{i_\ell}|)^\delta \\
    & \le & \sum_{i \in B(I)}|B^{W_2}_i| + 10 \min\big\{\Sigma(I),m(I)\big\} - c(\sum_{\ell \in [1,L]}|S_{i_\ell}|)^\delta\\
    & \le & c|S_{i_1}|^\delta+c\sum_{\ell \in [2,L]}|S_{i_\ell}|^\delta+\\
    & & \phantom{c|S_{i_1}|^\delta+}~10 \min\big\{\sum_{i \in B(I)}|S_i|,\sum_{\ell \in [2,L]}(c|S_{i_\ell}|^\delta+|S_{i_\ell}|)\big\} -c(\sum_{\ell \in [1,L]}|S_{i_\ell}|)^\delta \\
    & = & \min\big\{f_1(S_{i_1}),f_2(S_{i_1})\big\}, 
  \end{eqnarray*}  
  {where}
  \begin{itemize}
  \item  $f_1(x)=cx^\delta +c\sum_{\ell \in [2,L]}|S_{i_\ell}|^\delta+10 (x+\sum_{\ell \in [2,L]}|S_{i_\ell}|)-c(x+\sum_{\ell \in [2,L]}|S_{i_\ell}|)^\delta $
  \item  $f_2(x)=cx^\delta +c\sum_{\ell \in [2,L]}|S_{i_\ell}|^\delta+10 \sum_{\ell \in [2,L]}(c|S_{i_\ell}|^\delta+|S_{i_\ell}|)-c(x+\sum_{\ell \in [2,L]}|S_{i_\ell}|)^\delta $
  \end{itemize}
Let $g:\mathbb{R}^+\to\mathbb{R}^+$ be such that $g(x)=(x+\sum_{\ell
  \in [2,L]}S_{i_\ell})^\delta - x^\delta$ and observe that it is a
non-decreasing function, as $\delta>1$.  Therefore $$f_2(x) =
c\sum_{\ell \in [2,L]}|S_{i_\ell}|^\delta+10 \sum_{\ell \in
  [2,L]}(c|S_{i_\ell}|^\delta+|S_{i_\ell}|)-cg(x),$$ and thus {$f_2$}
a non-increasing function.  Let $x_0=\sum_{\ell \in
  [2,L]}|S_{i_\ell}|^\delta$.  Notice that according to the first
property of a bucket partition, for any $i \in S^\npair,$ we have
$|S_i| \ge 1$, implying that $x_0 \ge L-1 \ge 1$.

To complete the proof we distinguish two cases according to the value
of $|S_{i_1}|$.  However, before we proceed with these two cases, we
prove a helpful claim.

\paragraph*{Claim 1:}

Let $g_2(x)= c^{\delta}x^\delta + 21x -
(cx+x^{\frac{1}{\delta}})^\delta$. For $x\ge 0$, $\delta>1$ and $c \ge
(\frac{21}{\delta})^\frac{1}{\delta-1}$, we have $g_2(x)\le 0$.

\noindent {\sl Proof of Claim 1:}
We have
\begin{eqnarray*}
  (cx+x^{\frac{1}{\delta}})^\delta & = &  c^\delta x^\delta(1+\frac{x^{\frac{1}{\delta}-1}}{c})^\delta \\
  & \ge & c^\delta x^\delta (1+\delta \frac{x^{\frac{1}{\delta}-1}}{c}) \mbox{ as $(1+u)^\alpha \ge 1+\alpha u$ for any $\alpha > 1$ and $u >0$} \\
  & = & c^\delta x^\delta + \delta c^{\delta-1} x^{\frac{1}{\delta}-1+\delta} \\
\end{eqnarray*}
Thus,
\begin{eqnarray*}
  g_2(x) & \le & 21x-\delta c^{\delta-1} x^{\frac{1}{\delta}-1+\delta} \\
  & \le & 21x-\delta c^{\delta-1} x \mbox{ as $\frac{1}{\delta}-1+\delta = \frac{\delta^2-2\delta+1}{\delta}+1 = \frac{(\delta-1)^2}{\delta}+1 \ge 1$ } \\
  & \le & 0 \mbox{~~~ as $c \ge (\frac{21}{\delta})^\frac{1}{\delta-1}$}
\end{eqnarray*}
This implies that $g_2(x) \le 0$ for any $x \ge 0$. This completes the proof of Claim 1.\medskip

Now that Claim 1 is proved, let us come back to our two cases.
 \paragraph*{Case 1 : if $|S_{i_1}| \ge cx_0$.}
 
In this case, we will prove that $f_2(|S_{i_1}|) {\le} 0$.
As $f_2$ is non-increasing, we get $f_2(|S_{i_1}|) \le f_2(cx_0)$ therefore, it remains to prove that $f_2(cx_0) \le 0$.
  Observe that as $\delta \ge 1$, $(\sum_{\ell \in [2,L]}|S_{i_\ell}|)^\delta \ge \sum_{\ell \in [2,L]}|S_{i_\ell}|^\delta$, implying that $(\sum_{\ell \in [2,L]}|S_{i_\ell}|) \ge x_0^{\frac{1}{\delta}}$. {Observe that} 
    \begin{eqnarray*}
    f_2(cx_0) & \le & cc^{\delta}x_0^\delta + cx_0+20cx_0 - c(cx_0+x_0^{\frac{1}{\delta}})^\delta \iff\\
    \frac{f_2(cx_0)}{c} & \le & c^{\delta}x_0^\delta + 21x_0 - (cx_0+x_0^{\frac{1}{\delta}})^\delta   = g_2(x_0)\\
  \end{eqnarray*}
The assertion in this case now follows from Claim 1.

\paragraph*{Case 2 : if $|S_{i_1}| < cx_0$.}

In this case, we will prove that $f_1(|S_{i_1}|) \ge 0$.
Let us first prove that $f_1$ is convex.
\begin{eqnarray*}
\frac{d^2}{dx^2}f_1(x) & = & c(\delta)(\delta-1)(x^{\delta-2}-(x+\sum_{\ell \in [2,L]}S_{i_\ell})^{\delta-2}) \\
& \ge & 0, \mbox{ as $\delta \le 2$}
\end{eqnarray*}
Thus, as $|S_{i_2}| \le |S_{i_1}| < cx_0$, by the definition of the $S_{i_j}$, and $f_1$ is convex, $f_1(|S_{i_1}|) \le \max(f_1(|S_{i_2}|),f_1(cx_0))$.
Let us prove that  both these values are lower or equal to zero.
Observe that $\frac{f_1(cx_0)}{c} \le g_2(x_0) \le 0$ by Claim 1, and thus it remains to prove that $f_1(|S_{i_2}|) \le 0$.

For any $x_i \ge 0$, let $f(x_2,\dots,x_L)= 2c^\delta x_2^\delta+c\sum_{\ell \in [3,L]}x_\ell^\delta+10 (2x_2+\sum_{\ell \in [3,L]}x_\ell)-c(2x_2+\sum_{\ell \in [3,L]}x_\ell)^\delta$.
Observe that $f_1(|S_{i_2}|)=f(|S_{i_2}|,|S_{i_3}|,\dots,|S_{i_L}|)$.

\paragraph*{Claim 2:} For any tuple of integers $(x_2,\dots,x_L)$ where $x_\ell \ge 1$ for any $\ell \in [2,L]$,
$f(x_2,\dots,x_L) \le f(x_2,0,\dots,0)$.\smallskip

\noindent {\sl Proof of Claim 2:}
Let $(x_2,\dots,x_L)$ be a tuple of integers where $x_\ell \ge 0$. Let $i_0 \in [3,L]$ such that $x_{i_0} \ge 1$, and $x'_{i_0}=x_{i_0}-1$. Let
$D  =  f(x_2,\dots,x_{i_0},\dots,x_{L})-f(x_2,\dots,x'_{i_0},\dots,x_L)$. Our goal is to prove that $D \le 0$.
Repeating the decrementation, this will imply Claim $2$.
\begin{eqnarray*}
  D & = & cx_{i_0}^\delta -c(x_{i_0}-1)^\delta+10-c(\alpha^\delta-(\alpha-1)^\delta) \mbox{ ~~~~~( where $\alpha = 2x_2+\sum_{\ell \in [3,L]}x_\ell$)} \\
  & \le  & c(x_{i_0}^\delta -(x_{i_0}-1)^\delta-(3x_{i_0})^\delta+(3x_{i_0}-1)^\delta) + 10 \mbox{  ~~~~~(as $D$ is non increasing in $\alpha$, and $\alpha \ge 3x_{i_0}$)} 
\end{eqnarray*}
For any $y \ge 1$, let $u(y)=y^\delta-(y-1)^\delta-(3y)^\delta+(3y-1)^\delta$. 
Let us first lower bound two different parts of $u$.
We have
\begin{eqnarray*}
  (y-1)^\delta & = & y^\delta(1-\frac{1}{y})^\delta \\
  & \ge & y^\delta(1-\frac{\delta}{y})  \mbox{ ~~~~~(as $(1-u)^\alpha \ge 1-\alpha u$ for any $\alpha > 1$ and $0 \le u \le 1$)} \\
  & = & y^\delta-\delta y^{\delta-1}
\end{eqnarray*}
Moreover, denoting $z=3y-1$, we have
\begin{eqnarray*}
  (3y)^\delta & = & (z+1)^\delta\\
  & = & z^\delta(1+\frac{1}{z})^\delta \\
  & \ge & z^\delta(1+\frac{\delta}{z})  \mbox{ ~~~~~ (as $(1+u)^\alpha \ge 1+\alpha u$ for any $\alpha > 1$ and $u >0$)} \\
  & = & z^\delta+\delta z^{\delta-1} \\
  & = & (3y-1)^\delta+\delta (3y-1)^{\delta-1} 
\end{eqnarray*}
Using these two lower bounds, we get
\begin{eqnarray*}
  u(y) &\le &\delta (y^{\delta-1} - (3y-1)^{\delta-1}) \\
       & \le & \delta (1-2^{\delta-1}) \mbox{ ~~~~~ (as $y^{\delta-1} - (3y-1)^{\delta-1}$ is decreasing in $y$ and $y \ge 1$)} 
\end{eqnarray*}
This implies
 \begin{eqnarray*}
   D & \le & c u(x_{i_0}) + 10 \\
   & \le &c  \delta (1-2^{\delta-1})  + 10 \\
   & \le & 0 \mbox{ ~~~~~(as $c \ge \frac{20}{(2^\delta-2)} \ge \frac{10}{\delta(2^{\delta-1}-1)}$ for $\delta \ge 1$)}
 \end{eqnarray*}
 Concluding the proof of Claim 2.\medskip

 Let us now come back to the end of proof of Case 2. We had:
 \begin{eqnarray*} 
   f_1(|S_{i_2}|)& =& f(|S_{i_2}|,|S_{i_3}|,\dots,|S_{i_L}|) \\
   & \le & f(|S_{i_2}|,0,\dots,0) \mbox{ by Claim $2$} \\
   & = & 2c^\delta |S_{i_2}|^\delta+20 |S_{i_2}|-c(2|S_{i_2}|)^\delta\\
   & = & c|S_{i_2}|^\delta(2-2^\delta)+20|S_{i_2}| \mbox{ ~~~~~( which is decreasing in $|S_{i_2}|$ as $c \ge \frac{20}{(2^\delta-2)}$ and $\delta \ge 1$)} \\
   & \le & c(2-2^\delta)+20 \mbox{ ~~~~~(as $|S_{i_2}| \ge 1$)} \\
   & \le & 0 \mbox{ ~~~~~(as $c \ge \frac{20}{(2^\delta-2)}$ and $\delta \ge 1$)} 
 \end{eqnarray*}
 This concludes the proof of Case 2.
\end{proof}

\paragraph*{Discussion on why a simpler definition of $\mu$ is not sufficient.}
We explain here why taking simply $\mu(I)=\ssigma(I)$ or $\mu(I)=m(I)$ (instead of $\mu(I)=\min(\ssigma(I),m(I))$) would not be sufficient to get a kernel in $\O(k^\delta)$ for any $\delta > 1$.
Let us assume that the current partial decomposition is $(W,B,C,B^{W_1},B^{W_2})$ where in particular $S^\npair = \{i_1,\dots,i_{k'}\}$ for $k'=\frac{k}{2}$, $B^{W_j} = \emptyset$
for $j \in \{1,2\}$, implying that buckets only contain vertices of $C^0$, and $|B_i|=1$ for any $i \in S^\npair$. Suppose that $k=2^x$ for some $x$.

Suppose first that we define $\mu(I)=\ssigma(I)$.
Notice that particular $\textsf{val}([i_1,i_2])=\mu([i_1,i_2])=2$, and thus that in the auxiliary graph we will have in particular two colors $d_1$, $d_2$ such that all vertices of $W([i_1,i_2])$ have one loop of color $d_1$
and one loop of color $d_2$.
The kernel now applies its unique reduction rule (see Definition \autoref{def_reductionrule}). Suppose that
we do not find a rainbow matching, and that the set of color we find using  \autoref{cor_rainbowvc} is $X=\{d_1,d_2\}$, and thus that the small associated vertex cover is the set $U = W([i_1,i_2])$ ($U$ must cover all these loops),
where $|U| = 5 |X| = 10$. In this case, we fall into Case 2, and the $add_2$ operation will merge $B_{i_1}$, $B_{i_2}$ and $W([i_1,i_2])$ into a new bucket $B'_{i_2}$.
If we now repeat the same scenario between $B'_{i_2}$ and $B_{i_3}$, we will now have $\mu(I)=3$ (where $I=[i_2,i_3]$), and thus $|U|\le 15$.
Thus, if we repeat again until all buckets are merged into a single one, the total number of vertices of $W$ added to buckets during these $add_2$ operations will be $5(\sum_{i=2}^{k'}i)=\O(k^2)$, and thus we could not obtain
a kernel in $\O(k^\delta)$ vertices.

Suppose now that we define $\mu(I)=m(I)$.
Notice that particular $\textsf{val}([i_\ell,i_{\ell+1}])=\mu([i_\ell,i_{\ell+1}])=1$, and thus that for any $\ell$ in the auxiliary graph we will have in particular one color $d_\ell$ such that all vertices of $W([i_\ell,i_{\ell+1}])$ have one loop of color $d_\ell$.
The kernel now applies its unique reduction rule (see Definition\autoref{def_reductionrule}). Suppose that
we do not find a rainbow matching, and that the set of color we find using  \autoref{cor_rainbowvc} is $X=\{d_1\}$, and thus that the small associated vertex cover is the set $U = W([i_1,i_2])$,
where $|U| = 5 |X| = 5$. In this case, we fall into Case 2, and the $add_2$ operation will merge $B_{i_1}$, $B_{i_2}$ and $W([i_1,i_2])$ into a new bucket $B'_{i_2}$ where $|B'_{i_2}|=7$.
If we now repeat the same scenario for $\ell=3,5,\dots$ between $B_{i_\ell}$ and $B_{i_\ell+1}$, the $add_2$ operation will merge $B_{i_\ell}$, $B_{i_{\ell+1}}$ and $W([i_\ell,i_{\ell+1}])$ into a new bucket $B'_{i_{\ell+1}}$ where $|B'_{i_{\ell+1}}|=7$.
We now have $2^{x-2}$ buckets $B'_{i_2},B'_{i_4},\dots,B'_{k'}$, each of size $7$.
The total number of vertices added to $B$ in this first part of the scenario is thus $n_1 =2^{(x-2)}5$.
If we again repeat the same scenario for $\ell=2,6,\dots$, the $add_2$ operation will merge $B'_{i_\ell}$, $B'_{i_{\ell+2}}$ and we will have $|U| = 5|X| \ge 5\times 7 \ge 5^2$.
Thus, the total number of vertices added to $B$ in this second part of the scenario is $n_2 = 2^{(x-3)}5^2$. If we continue until there is only one bucket,
the total number of vertices added in the last part of the scenario will be $n_{x-1}=5^{(x-1)}=2^{(x-1)log(5)}=k^{'log(5)}$, and thus we cannot hope for a kernel in $\O(k^\delta)$ vertices.
            
\subsection{Analysis the two cases: rainbow matching or small vertex cover}
Our kernelization algorithm takes as input a partial decomposition
$(W,B,C,B^{W_1},B^{W_2})$ of size $f$ where $f(x)=c x^\delta$ for some
constants $c$ and $\delta.$ Before going to the formal description of
the algorithm, let us sketch how it works, and why we need the
following lemmas.  \stf{At each round, the kernelization algorithm first
derives from $(W,B,C,B^{W_1},B^{W_2})$ a clean partial decomposition
with same size $f$ using Lemma~\ref{lem:clean-TPT}, then
builds the auxiliary edge-colored graph
$(G,χ)\langle W,B,C,B^{W_1},B^{W_2}\rangle$ and tries to find a
rainbow matching using \autoref{cor_rainbowvc}}.  If it finds such a
matching, then the algorithm stops and concludes using
\autoref{lemma_saferainbowT}.  Otherwise, the algorithm finds a small
vertex cover and performs an add operation in order to obtain another
tuple $(W',B',C',B^{'W_1},B^{'W_2}).$ \autoref{lemma_smallvc} will
allow us to ensure that $(W',B',C',B^{'W_1},B^{'W_2})$ is still a
partial decomposition of local size $f.$

\begin{lemma}[Case of rainbow matching]\label{lemma_saferainbowT}
  Let $(W,B,C,B^{W_1},B^{W_2})$ be a \stf{clean partial decomposition}
  of local size $f,$ where $f(x)=c x^\delta$ for some positive
  constants $c$ and $\delta.$ Suppose also that the colored multigraph
  $(G,χ)\langle W,B,C,B^{W_1},B^{W_2}\rangle$ admits a rainbow
  matching $M.$ Let $A=V(M) \cup B \cup C.$ Then, for any integer $k,$
  $(\T,k)$ is a \yes-instance of \TPT iff $(T[A],k)$ is a
  \yes-instance of \TPT.  Moreover, we have
  $|A| \le 3(c+1)(33)^\delta(k)^\delta.$
\end{lemma}

\begin{proof}
Let $\T'=T[A].$  As $\T'$ is an induced tournament of $\T,$ direction
$\Leftarrow$ is clearly true.  Let us now suppose that $(\T,k)$ is a
\yes-instance and prove that $(\T',k)$ is a \yes-instance.  Let $\P$
be a triangle-packing of $\T$ of size $k$ and let us prove that there
is a triangle-packing $\P'$ of size $k$ in $\T'.$  Let us
partition $\P = \P_1 \cup \P_2,$ where $\P_1 = \{\Delta \in \P \mid
\Delta \subseteq W \cup B\},$ and $\P_2 = \{\Delta \in \P \mid \Delta
\cap C \neq \emptyset\}.$  For any $\Delta \in \P_2,$ let $c(\Delta)$
be a vertex in $\Delta \cap C.$
  
For any $c \in C,$ let $f(c)$ be the edge of color $c$ in $M$. \stf{Notice
that $f(c)$ is well defined as $(W,B,C,B^{W_1},B^{W_2})$ is
clean.} Moreover, for any $I \in \I_{>0}$ (where $(\I,\textsf{val})$ is
the demand of $(G,χ)\langle W,B,C,B^{W_1},B^{W_2}\rangle$) and color
$d \in D_I,$ let $f(I,d)$ be the edge of color $d$ in $M.$ Let
$f(I)=\{f(I,d), d \in D_I\}.$ Observe that, for any $c \in C,$ by the
definition of edges of colors $c,$ $\{c\} \cup f(c)$ is a triangle in
$\T.$ Let $\Q_2 = \{f(c),c \in C\}$ and
$\Q_1 = \{f(I), I \in \I_{>0} \}.$ Observe that $M=\Q_1 \cup \Q_2.$ As
$M$ is a matching, for any $I \in \I_{>0},$
$|f(I)|=|D_I|=\textsf{val}(I)$; moreover notice that the $f(I)$ are
disjoint.  As, for any $d \in D_I,$ the only vertices in
$(G,χ)\langle W,B,C,B^{W_1},B^{W_2}\rangle$ that have color $d$ are
$W(I),$ we get $f(I) \subseteq W(I).$ This implies that $\Q_1$ is a
bucket allocation for $(W,B,C,B^{W_1},B^{W_2}).$ Thus, according to
\autoref{lemma_safebucketT}, there exists a packing $\P'_1$ such that
$V(\P_1') \subseteq V(\Q_1) \cup B$ and $|\P'_1|=|\P_1|.$
 
Let us now restructure triangles in $\P_2.$  For any $\Delta \in
\P_2,$ we define $g(\Delta)=\{c(\Delta)\} \cup f(c(\Delta)).$  Let
$\P'_2 = \{g(\Delta), \Delta \in \P_2\}.$ Finally, we define $\P' =
\P'_1 \cup \P'_2.$  Let us prove that $\P'$ is a triangle-packing.
For any $\Delta'_1, \Delta'_2$ in $\P'_1,$ $\Delta'_1 \cap \Delta'_2 =
\emptyset,$ as $\P'_1$ is a packing.  Let us now consider $\Delta'_1,
\Delta'_2$ in $\P'_2,$ where $\Delta'_i = g(\Delta_i).$  First, as
$\Delta_1 \cap \Delta_2 = \emptyset,$ we obtain that $c(\Delta_1) \neq
c(\Delta_2).$  Moreover, as $M$ is a matching, {$f(c(\Delta_1)) \cap
  f(c(\Delta_2)) = \emptyset.$}  As $f(c(\Delta_i)) \subseteq V(\Q_i)$
and $c(\Delta_i) \in C,$ we also have $f(c(\Delta_i)) \cap
\{c(\Delta_{3-1})\} = \emptyset,$ for $i \in [2],$ implying $\Delta'_1
\cap \Delta'_2 = \emptyset.$  Consider now the last case where
$\Delta'_i \in \P'_i$ for $i \in [2].$  As $\Delta'_1 \subseteq
V(\Q_1) \cup B$ and $\Delta'_2 \subseteq V(\Q_2) \cup C,$ and as
$V(\Q_1),$ $V(\Q_2),$ $B$ and $C$ are pairwise disjoint sets, we get
the desired result.  Finally,
$|\P'|=|\P_1'|+|\P'_2|=|\P_1|+|\P_2|=|\P| \ge k.$
  
The next step is to prove an upper bound on $|V(\T')|.$  We start by
bounding $|V(M)|.$  Recall that $V(M) = V(\Q_1) \cup V(\Q_2),$
implying that $|V(M)| = |\Q_1| + 2\cdot |C|.$  Let us now bound
$|\Q_1|.$  Let $\I_{>0}^\textrm{max}$ be the set of $\subsetint$-wise
maximal intervals of $\I_{>0}.$  Let also $\{Z_1,\dots,Z_t\}$ and
$\{I^*_1,\dots,I^*_t\}$ be the block partition and the block
intervals of $\I_{>0}^\textrm{max}.$  Observe that $r_{I^*_\ell} \le
l_{I^*_{\ell+1}},$ implying that any bucket index $i \in S^\npair$
belongs to at most two $B(I^*_\ell)$'s.  Moreover, the only case where
$i$ belongs to two $B(I^*_\ell)$'s is when $i = r_{I^*_\ell}$ and
$r_{I^*_\ell} = l_{I^*_{\ell+1}}.$

We are now ready to bound $|\Q_1|.$
  \begin{eqnarray*}
    |\Q_1| & = & \sum_{I \in \I_{>0}}|f(I)| \\
    & = & \textsf{val}(\I_{>0})\\
    & \le & \sum_{\ell \in [t]} \mu(I^*_\ell)\mbox{ by \autoref{lemma_valblock}} \\
      & \le & \sum_{\ell \in [t]} \sum_{i \in B(I^*_\ell)}|B_i| \\
      & \le & 2 \sum_{i \in S^\npair}|B_i|, \mbox{ as any $i \in S^\npair$ belongs to at most two $B(I^*_\ell)$}
  \end{eqnarray*}
 
This implies $|V(M)| \le 2|B|+2|C|.$
We are now ready to bound the size of the kernel output:
  \begin{eqnarray*}
    |V(\T')| &= & |V(M)|+|B|+|C| \\
    &\le & 3|B|+3|C| \\
    &=&3\big( |B^C|+|B^{W_1}|+|B^{W_2}|\big)+3|C|\\
    &\le&3\big(|B^C|+|B^{W_1}|+c\sum_{i \in S^\npair}(|B_i^C|+|B_i^{W_1}|)^\delta\big)+3|C|\\
    &\le&3\big(|B^C|+|B^{W_1}|+c(\sum_{i \in S^\npair}(|B_i^C|+|B_i^{W_1}|))^\delta\big)+3|C|\\
    &\le&3\big(|B^C|+|B^{W_1}|)+c(|B^C|+|B^{W_1}|)^\delta\big)+3|C|\\
    &\le&3(c+1)(|B^C|+|B^{W_1}|)^\delta+3|C|\\
    &\le&3(c+1)(11|B^C|)^\delta +3|C|, \mbox{by \autoref{buk_par_pair}}\\
     &\le&3(c+1)11^\delta(|B^C|+|C|)^\delta\\
    &=&3(c+1)11^\delta(|C^0|)^\delta\\
    &\le&3(c+1)11^\delta(3k)^\delta   
  \end{eqnarray*}

\end{proof}

\stf{The next lemma deals with the other case of the reduction rule,
  where the auxiliary edge-colored graph contains a small vertex cover
  for a set of colors. Notice that in this case, we do no necessarily
  need that the considered partial decomposition is clean.}

\begin{lemma}[Case of a small vertex cover]\label{lemma_smallvc}
Let $(W,B,C,B^{W_1},B^{W_2})$ be a partial decomposition of local size
$f,$ where $f(x)=c x^\delta$ for some constants $c>0$ and $\delta$
with $1<\delta≤2.$  Suppose that there exists a non-empty subset of
colors $X$ such that {$(G,χ)\langle W,B,C,B^{W_1},B^{W_2}\rangle[X]$}
admits a vertex cover $U$ with $|U| \le 5|X|.$  Recall that colors of
$(G,χ)\langle W,B,C,B^{W_1},B^{W_2}\rangle$ are partitioned into $C
\cup D.$  Let $X^C = X \cap C$ and $X^D = X \cap D.$
\begin{itemize}
 \item (Case 1:) If $|X^D| \le |X^C|,$ then
   $(W',B',C',B'^{W_1},B'^{W_2})=\textsf{add}_1(W,B,C,B^{W_1},B^{W_2},U
   \cup X^C)$ is still a partial decomposition of local size $f$ where
   $|W'|+|C'|< |W|+|C|.$
  \item (Case 2:) If $|X^D| > |X^C|,$ then we can find, in polynomial
    time, a non-empty $U' \subseteq U$ such that
    $(W',B',C',B'^{W_1},B'^{W_2})=\textsf{add}_2(W,B,C,B^{W_1},B^{W_2},U')$
    is a partial decomposition of local size $f,$ where $|W'|+|C'|<
    |W|+|C|.$
  \end{itemize}
\end{lemma}
\begin{proof}
Proof of \textsl{Case 1}.  Observe first that $|U|\le 5(|X^D|+|X^C|)
\le 10|X^C|,$ as we are in Case 1.  Let us now prove the two
conditions required by \autoref{lemma_add1}.  As $U$ is a vertex cover
of {$(G,χ)\langle W,B,C,B^{W_1},B^{W_2}\rangle[X]$}, for any color $c
\in X$ (and thus in $X^C$) and edge $e$ of color $c,$ $U \cap V(e)
\neq \emptyset.$  Suppose, towards a contradiction, that there exists
a triangle $\Delta \in \T[W \cup X^C]$ such that $|\Delta \cap (W
\setminus U)| \ge 2.$  Since we cannot have $\Delta \subseteq W,$ this
implies that $\Delta=\{w_1,w_2,c\}$ where $w_i \in W \setminus U$ and
$c \in X^C.$  It follows from $\Delta$ being a triangle that
$e=\{w_1,w_2\}$ is an edge of $(G,χ)\langle
W,B,C,B^{W_1},B^{W_2}\rangle$ of color $c,$ which is not intersected
by $U,$ a contradiction.  Thus, we can now apply \autoref{lemma_add1}
and obtain that $\textsf{add}_1(W,B,C,B^{W_1},B^{W_2},U \cup X^C)$ is
a partial decomposition of local size $f.$  Moreover, as $|X^D| \le
|X^C|$ and $X \neq \emptyset,$ we have $X^C \neq \emptyset,$ implying
$|C'| < |C|$ and thus $|W'|+|C'| < |W|+|C|,$ as required.\medskip
  
Proof of \textsl{Case 2}.  Observe first that $|U|\le 5(|X^D|+|X^C|)
\le 10|X^D|,$ as we are in Case 2.  Let us now find an interval $I'$
such that $|U'| \le 10\mu(I'),$ where $U'= W(I'),$ in order to apply
\autoref{lemma_add2}.  Remember that there is a partition $\{D_I\mid I
\in \I_{>0}\}$ of $D.$ For any $u$ in $X^D,$ let $I(u) \in \I_{>0}$ be
the unique bucket interval such that $u \in D_I.$  Let $\I(X^D) =
\{I(u)\mid u \in X^D\}$ and $\I^\textrm{max}(X^D)$ be the set of
$\subsetint$-wise maximal intervals of $\I(X^D).$  Let
$\{Z_1,\dots,Z_t\}$ and $\{I^*_1,\dots,I^*_\ell\}$ be the block
partition and block intervals of $\I^\textrm{max}(X^D).$  Moreover, as
$U' \neq \emptyset,$ we get $|W'| < |W|,$ implying $|W'|+|C'| <
|W|+|C|.$
 
For any $\ell \in [t],$ let $U_\ell = W(I^*_\ell).$  Remember that in
$(G,χ)\langle W,B,C,B^{W_1},B^{W_2}\rangle,$ for each color $u \in D,$
we add all $v \in W(I(u))$ as edges of color $u.$  Thus, as $U$ is a
vertex cover of {$(G,χ)\langle W,B,C,B^{W_1},B^{W_2}\rangle[X],$} for
any $u \in X^D,$ $U$ must contain $W(I(u)).$  Hence, we get
   \begin{eqnarray*}
     U &= &\bigcup_{u \in X^D}W(I(u)) \\
     & = &\bigcup_{I \in \I(X^D)}W(I) \\
     & = & \bigcup_{I \in \I^\textrm{max}(X^D)}W(I)\\
     & = &\bigcup_{\ell \in [t]}W(I^*_\ell)\\
     & = &\bigcup_{\ell \in [t]}U_\ell.
   \end{eqnarray*}
Finally, observe that the $U_\ell$ are disjoint.
   
Let us now upper bound $|X^D|.$
   \begin{eqnarray*}
     |X^D| &= &\sum_{I \in \I(X^D)}|X^D \cap D_I|\\
       & \le &\sum_{I \in \I(X^D)}|D_I|\\
     & = &\bigcup_{I \in \I(X^D)}\textsf{val}(I), \mbox{ by definition of $(G,χ)\langle W,B,C,B^{W_1},B^{W_2}\rangle$} \\
         & = & \textsf{val}(\I(X^D)) \\
         & \le & \sum_{\ell \in [t]} \mu(I^*_\ell), \mbox{ by \autoref{lemma_valblock}.}
   \end{eqnarray*}

Hence $|U| \le 10|X^D|$ implies that $\sum_{\ell \in [t]}|U_\ell| \le
10\sum_{\ell \in [t]} \mu(I^*_\ell).$  This, in turn, implies that
there exists $\ell \in [t]$ such that $|U_\ell| \le 10\mu(I^*_\ell)$
and we can find this $\ell$ in polynomial time by enumerating all
values in $[t].$  Thus, we set $U'=U_\ell$ and, as $U_\ell =
W(I^*_\ell)$ and $|U_\ell| \le 10\mu(I^*_\ell),$ it follows from
\autoref{lemma_add2} that $\textsf{add}_2(W,B,C,B^{W_1},B^{W_2},U')$
is a partial decomposition of local size $f,$ as required.
 \end{proof}

\subsection{Analysis of the overall kernel}
\stf{We are now ready to define the unique rule, except from the cleaning
phases, for our kernelization algorithm.}  In the following, we assume
that $\delta$ with $1<\delta≤2$ is fixed.

\begin{definition}[Reduction Rule for \TPT]\label{def_reductionrule}~
Given a \stf{clean partial decomposition} $(W,B,C,B^{W_1},B^{W_2})$ of local size
$f,$ where $f(x)=c(\delta)x^\delta$ with $c(\delta) =
\max(\frac{20}{(2^\delta-2)},(\frac{21}{\delta})^\frac{1}{\delta-1}),$ let us define the output $R(W,B,C,B^{W_1},B^{W_2})$
of the rule $R$ as follows:
\begin{itemize}
  \item Decide, using \autoref{cor_rainbowvc} (applied for $ε=1$),
    whether there exists a rainbow matching $M$ in the
    $p$-edge-colored mutigraph $(G,χ)\langle
    W,B,C,B^{W_1},B^{W_2}\rangle$ (where $p=|C|+\sum_{i\in
      \I_{>0}}\textsf{val}(I)$ -- recall that $p=\O(|V(\T)|^3)$).
  \begin{itemize}
  	\item If the algorithm finds a rainbow matching $M,$ then return $V(M) \cup B \cup C.$
  	\item Otherwise, let $X$ be the non-empty set of colors such
          that {$(G,χ)\langle W,B,C,B^{W_1},B^{W_2}\rangle[X]$} admits
          a vertex cover $U$ such that $|U| \le (4+ε)|X|.$ Let $X^C =
          X \cap C$ and $X^D = X \cap D.$
  	\begin{itemize}
  		\item (\textsl{Case 1}) If $|X^D| \le |X^C|,$
                  let $$(W',B',C',B'^{W_1},B'^{W_2}) =
                  \textsf{add}_1(W,B,C,B^{W_1},B^{W_2},U \cup C).$$
                  Return $(W',B',C',B'^{W_1},B'^{W_2}).$
		\item (\textsl{Case 2}) Otherwise we have $|X^D| >
                  |X^C|.$ In this case compute, in polynomial time,
                  according to \autoref{lemma_smallvc}, a non-empty
                  $U' \subseteq U$ such
                  that $$(W',B',C',B'^{W_1},B'^{W_2}) =
                  \textsf{add}_2(W,B,C,B^{W_1},B^{W_2},U')$$ is a
                  partial decomposition of local size $f.$
                  Return $(W',B',C',B'^{W_1},B'^{W_2}).$
  \end{itemize}
  \end{itemize}
  \end{itemize}
\end{definition}

Then, we obtain the following.

\begin{lemma}\label{lemma:ruleTPT}
\stf{Given a clean partial decomposition $(W,B,C,B^{W_1},B^{W_2})$ of local size
$f,$} where $f(x)=c(\delta)x^\delta$ with $c(\delta) =
\max(\frac{20}{(2^\delta-2)},(\frac{21}{\delta})^\frac{1}{\delta-1}),$ $R(W,B,C,B^{W_1},B^{W_2})$ either returns:
\begin{itemize}
  \item a set $A\subseteq V(T)$ such
that, if $T'=T[A],$ then $(T,k)$ and $(T',k)$ are equivalent
instances of \TPT\ and $|V(T')| \le
c_{0}\cdot c(δ)\cdot k^\delta,$ with $c_0 = 6534.$
 \item or a partial decomposition $(W',B',C',B'^{W_1},B'^{W_2})$ of local size
$f$, where $|W'|+|C'| < |W|+|C|.$
\end{itemize}
\end{lemma}
\begin{proof}
If $R$ finds a rainbow matching $M$ in $(G,χ)\langle
W,B,C,B^{W_1},B^{W_2}\rangle$, then \autoref{lemma_saferainbowT} immediately implies that the set $A=V(M) \cup B \cup C$ verifies the claimed properties,
as in particular $|V(T')| \le  3(c(δ)+1)33^\delta k^\delta \le 6\cdot c(δ)\cdot 33^2 k^\delta=6534\cdot c(δ)\cdot  k^\delta.$
Let us now consider that $R$ does not find a rainbow matching.

If $R$ falls into Case 1 or Case 2, then by \autoref{lemma_smallvc} we know, as we set $ε=1,$ that
$(W',B',C',B'^{W_1}$ $,B'^{W_2})$ is still a partial decomposition of
local size $f$ with $|W'|+|C'| < |W|+|C|.$
\end{proof}

Finally, we can prove the kernelization algorithm for \TPT stated in~\autoref{main_treh}.

\begin{proof}[Proof of \autoref{main_treh}]
  Let $\delta$ with $1 < \delta \le 2$,
  $c(δ) =
  \max(\frac{20}{(2^\delta-2)},(\frac{21}{\delta})^\frac{1}{\delta-1})$
  and notice that $c(δ)\geq 1$. Given an input $(T,k),$ the
  kernelization algorithm \texttt{A} starts by the greedy localization
  phase\footnote{The greedy localization is a greedy packing of
    triangles, see \autoref{greedy_loc}.}.  Assume that it does not
  find a packing of size $k,$ therefore it computes a greedy localized
  pair $(C^0,W^0)$ of $T.$ Observe that
  $(W^0,\emptyset,C^*,\emptyset,\emptyset)$ is a partial decomposition
  of local size $f,$ where $f(x)=c(\delta)x^\delta.$ \stf{Using
    Lemma~\ref{lem:clean-TPT}, we then obtain the first clean partial
    decomposition $(W,B,C,B^{W_1},B^{W_2})$ of the process, still with
    local size $f$. Now, a step of algorithm~\texttt{A} consists in
    applying the Reduction Rule $R$ for \TPT followed by a cleaning
    phase.  Algorithm~\texttt{A} exhaustively performs steps,
    obtaining a clean partial decomposition with size $f$ at the end
    of each step and stopping only when it
    falls into the matching case.}
  
  \stf{Let us prove, by induction on $|W|+|C|$, that this terminates
    in polynomial time and outputs a set $A\subseteq V(T)$ such that,
    if $T'=T[A],$ then $(T,k)$ and $(T',k)$ are equivalent instances
    of \TPT\ and $|V(T')| \le c_{0}\cdot c(δ)\cdot k^\delta,$ with
    $c_0 = 6534.$ In order to obtain this result, first notice that
    when Rule $R$ applied on a clean partial decomposition
    $(W,B,C,B^{W_1},B^{W_2})$ returns a partial decomposition
    $(W',B',C',B'^{W_1},B'^{W_2})$ with the same size and where
    $|W'|+|C'|<|W|+|C|$. When applying the cleaning phase,
    Lemma~\ref{lem:clean-TPT}, on $(W',B',C',B'^{W_1},B'^{W_2})$, we
    move vertices from $C'$ to $B'$, leaving unchanged $W'$ (and the
    sets $B'^{W_1}$ and $B'^{W_2}$, as well as the size of the
    decomposition). If we denote by
    $(W'',B'',C'',B''^{W_1},B''^{W_2})$ the clean partial
    decomposition obtain after the cleaning phase, we obtain
    $|W''|+|C''|\le |W'|+|C'|<|W|+|C|$.\\
    Now, we can  analyze the  entire process.  \\ If
    $W = C = \emptyset,$ then $W = \emptyset,$ the graph 
    {$(G,χ)\langle W,B,C,B^{W_1},B^{W_2}\rangle[X]$} is the empty
    graph, and we consider that $R$ returns the rainbow matching
    $M=\emptyset.$ Thus, according to \autoref{lemma:ruleTPT}, the
    rule $R$ outputs a set $A$ as required.  Otherwise, if
    $|W|+|C| > 0.$, it is immediate by induction, using
    \autoref{lemma:ruleTPT} and the previous remarks concerning
    cleaning phases, that \texttt{A} terminates in polynomial time and
    outputs a set $A\subseteq V(T)$ such that, if $T'=T[A],$ then
    $(T,k)$ and $(T',k)$ are equivalent instances of \TPT\ and
    $|V(T')| \le c_{0}\cdot c(δ)\cdot k^\delta,$ with $c_0 = 6534.$}
\end{proof}

Let us conclude this section by the following observation
  on bit-size.  According to the $\O(k^{2-\epsilon})$ bit-size lower
  bound of~\cite{BessyBT17trian}, our kernel of~\autoref{main_treh},
  as well as any subquadratic kernel for \TPT, necessarily ouputs a
  tournament $T'$ whose minimum feedback arc set cannot have
  $m' = \O(k^{2-\epsilon})$ arcs for any $\epsilon > 0$.  Indeed,
  suppose by contradiction that we have a kernel that outputs such a
  $T'$, where $|V(T')|=\O(k^{2-\epsilon'})$ and
  $m' = \O(k^{2-\epsilon})$.  We could now use the following classical
  encoding of $T'$.  Compute first a feedback arc set $F'$ of $T'$
  with $|F'| = \O(m')$ using any constant approximation approximation
  algorithm (for example the PTAS of~\cite{kenyon2007rank}).  Compute
  the topological ordering $\sigma$ of $T'\setminus F'$ (meaning when
  removing arcs of $F'$), and observe that in the ordering $\sigma$,
  the only backward arcs of $T'$ are $F'$.  Thus, $T'$ is encoded by
  $\sigma$ and $F'$, which requires
  $|V(T')|\log(|V(T')|)+m'(2\log(|V(T')|))=\O(k^{2-\min(\epsilon,\epsilon')})$
  bits.  This draws a parallel with the Vertex Cover problem, which
  admits a kernel with $\O(k)$ vertices~\cite{linearVc74}, but the
  resulting instance cannot have $\O(k^{2-\epsilon})$
  edges~\cite{dell2014kernelization}.

\section{Almost linear kernel for \FVST}
\label{@linearKernelFVST}

\stf{In this section we focus on the triangle hitting set probem,
  restated below. Remind that a set intesecting all the triangle of a
  tournament is also called  a feedback vertex set.}

\defparproblem{{\sc Feedback Vertex Set in Tournament}
  (\FVST)}{$(T,k)$ where $T$ is a tournament and
  $k \in \N$.}{$k.$}{Does $T$ contain a feedback vertex set of size at
  most $k$?}

\stf{We obtain a almost linear kernel for this problem, which is
  obtained by the same algorithm that the one designed in the previous
  section. The only thing that we will have to check is that, when the
  algorithm stops, we obtain an equivalence instance than the input
  instance for \FVST. The way we obtain this kernelization algorithm
  is similar to what is done~\autoref{@linearKernelIPHS}, where a
  kernel for \IPHS is derived from the one for \IPP.}

\begin{theorem}
\label{theo:kernelFVST}
\stf{There exists an algorithm that, given an instance $(T,k)$ of \FVST
outputs a set $S\subseteq V(T)$ such that $(T[S],k)$ is an equivalent
instance of $(T,k)$ where, for every $\delta$ with $1<\delta \le 2,$
we have $|S|\leq 6534\cdot c(δ)\cdot k^\delta$ (where
$c(\delta)=\max(\frac{20}{(2^\delta-2)},(\frac{21}{\delta})^\frac{1}{\delta-1})$).
In other words, for any $\delta$ with $1 < \delta \le 2$ \FVST admits a
kernel with $6534\cdot c(\delta)k^\delta$ vertices.}
\end{theorem}

\stf{As in~\autoref{@CorollaryTPT}, by setting the suitable value for
$\delta_0$, we obtain the following.}

\begin{corollary}
  \label{@CorollaryFVST}
\stf{\FVST admits a kernel with $k^{1+\frac{\O(1)}{\sqrt{\log{k}}}}$  vertices.}
\end{corollary}

\stf{Here again, the key tool to obtain the linear kernel for \FVST is
  the following lemma, analog of Lemma~\ref{lemma_saferainbowT} for
  \TPT.  All the notations and definitions follow previous section.}

\begin{lemma}[Case of rainbow matching for \FVST]\label{lem:rainbowFVST}
  \stf{Let $(W,B,C,B^{W_1},B^{W_2})$ be a clean partial decomposition
  of local size $f,$ where $f(x)=c x^\delta$ for some positive
  constants $c$ and $\delta.$ Suppose also that the colored multigraph
  $(G,χ)\langle W,B,C,B^{W_1},B^{W_2}\rangle$ admits a rainbow
  matching $M.$ Let $A=V(M) \cup B \cup C.$ Then, for any integer $k,$
  $(\T,k)$ is a \yes-instance of \FVST iff $(T[A],k)$ is a \yes-instance of \FVST.
  Moreover, we have $|A| \le 3(c+1)(33)^\delta(k)^\delta.$}
\end{lemma}

\begin{proof}
  \stf{The size requirement concerning $A$ follows from
    Lemma~\ref{lemma_saferainbowT}. Let us prove that $(T,k)$ and
    $(T[A],k)$ are equivalent instances of \FVST. As $T[A]$ is a
    subtournament of $T$, it is clear that if $T$ admits a feedback
    vertex set of size at most $k$, then it is also the case for
    $T[A]$.\\ For the converse direction, assume that $T[A]$ admits a
    feedback vertex set $X$ of size at most $k$, and let us see how to
    build one for $T$.  As in the proof of
    Lemma~\ref{lemma_saferainbowT}, we will first analyze the
    structure of $M$.  For any $c \in C,$ let $f(c)$ be the edge of
    color $c$ in $M$ ($f(c)$ is well defined as
    $(W,B,C,B^{W_1},B^{W_2})$ is clean). Moreover, for any $I \in
    \I_{>0}$ (where $(\I,\textsf{val})$ is the demand of $(G,χ)\langle
    W,B,C,B^{W_1},B^{W_2}\rangle$) and color $d \in D_I,$ let $f(I,d)$
    be the edge of color $d$ in $M.$ Let $f(I)=\{f(I,d), d \in D_I\}.$
    Observe that, for any $c \in C,$ by the definition of edges of
    colors $c,$ $\{c\} \cup f(c)$ is a triangle in $\T.$ Let $\Q_2 =
    \{f(c),c \in C\}$ and $\Q_1 = \{f(I), I \in \I_{>0} \}.$ As argued
    in the proof of Lemma~\ref{lemma_saferainbowT}, $\Q_1$ is a bucket
    allocation for $(W,B,C,B^{W_1},B^{W_2})$. So, denote by $X_1$ the
    set $X\cap (B\cup V(\Q_1))$ and by $X_2$ the set $X\setminus
    X_1$. As $X$ is a feedback vertex set of $T[A]$, it is clear that
    $X_1$ is a vertex set of $T[B\cup V({\cal Q}_1)]$. Thus by
    Lemma~\ref{lemma_safebucketT_hit}, there exists $X_1'$ a feedback
    vertex set of $T[W\cup B]$ with $|X_1'|\le|X_1|$. Finally, notice
    that $X_2$ is, in particular, a feedback vertex set of $T[V({\cal
        Q}_2)]$ which contains the $|C|$ disjoint triangles $\{f(c),c
    \in C\}$ and that $|C|\le |X_2|$. To conclude, we consider the set
    $X'=C\cup X_1'$. We have $|X'|\le |X|\le k$ and as every triangle
    of $T$ either intersects $C$ or is included in $T[W\cup B]$, $X'$
    is a feedback vertex set of $T$.\\ Therefore, the two instances
    $(T,k)$ and $(T[A],k)$ are equivalent.}
\end{proof}

\stf{Now, using Lemma~\ref{lem:rainbowFVST} instead of
  Lemma~\ref{lemma_saferainbowT} in the proof of
  Lemma~\ref{lemma:ruleTPT}, we directly obtain the analog of this
  latter one for \FVST.}

\begin{lemma}\label{lemma:ruleFVST}
\stf{Given a clean partial decomposition $(W,B,C,B^{W_1},B^{W_2})$ of local size
$f,$ where $f(x)=c(\delta)x^\delta$ with $c(\delta) =
\max(\frac{20}{(2^\delta-2)},(\frac{21}{\delta})^\frac{1}{\delta-1}),$ $R(W,B,C,B^{W_1},B^{W_2})$ either returns:
\begin{itemize}
  \item a set $A\subseteq V(T)$ such that, if $T'=T[A],$ then $(T,k)$
    and $(T',k)$ are equivalent instances of \FVST\ and $|V(T')| \le
    c_{0}\cdot c(δ)\cdot k^\delta,$ with $c_0 = 6534.$
 \item or a partial decomposition $(W',B',C',B'^{W_1},B'^{W_2})$ of local size
$f$, where $|W'|+|C'| < |W|+|C|.$
\end{itemize}}
\end{lemma}

\stf{Now, proof of Theorem~\ref{theo:kernelFVST} works the same than
  the kernelization process for \TPT that is~\autoref{main_treh}. We
  start by computing a greedy localized pair $(C^0,W^0)$. The set
  $C^0$ induces a packing of triangles of $T$. If there is more than
  $k$ triangles in the packing, then $T$ has no feedback vertex set of
  size at most $k$. Otherwise, we consider the initial partial
  decomposition $(W_0,\emptyset ,C_0)$ of $V(G)$.  Then, we
  exhaustively alternate a cleaning phase with an application of the
  rule $R$, until this last one falls into the matching case. As in
  the proof of~\autoref{main_treh}, using Lemma~\ref{lemma:ruleFVST}
  an induction on $|W|+|C|$ shows that this later case appears
  after a polynomial number of steps.  Then we conclude with
  Lemma~\ref{lem:rainbowFVST}.}

\section{Conclusion}
\label{@wissenschaftlichen}

In this paper we introduced the rainbow matching technique in
  order to derive kernelization algorithms for the
  \textsc{Triangle-Packing in Tournament} and \textsc{Feedback Vertex
    Set in Tournament} problems and the {\sc Induced 2-Path-Packing}
  and {\sc Induced 2-Paths Hitting Set} problems. {For the two first problems we derive kernels of
  $k^{1+\frac{\O(1)}{\sqrt{\log{k}}}}$ vertices, while for the two last we derive kernels of $\O(k)$ vertices}.  We stress that both kernels are producing equivalent
  instance that are sub(di)graphs of the original input (di)graphs. An
  interesting project is to investigate for which (di)graph packing or
  hitting set problems our technique can be used for the derivation of
  kernels of (almost) linear number of vertices. The general
  frameworks that encompass all such problems are the {\sc $r$-Set
    Packing} and the {\sc $r$-Hitting Set} problems where, given a
  hypergraph $H$ whose all hyperedges have $r$ vertices and a
  non-negative integer $k$, the questions respectively are wether $H$
  contains $k$ pairwise disjoint hyperedges and wether $H$ contains a
  set of size at most $k$ that intersects all the hyperedges.  In
  \cite{Abu_Khzam10animp} and \cite{Abu_Khzam10akadhs}, Abu-Khzam
  proved that both problems admit kernels of $\O(k^{r-1})$
  vertices. Any improvement of any of this to a kernel of
  $\O(k^{r-1-ε})$ vertices, {where the kernel is obtained by
    only removing vertices}, would imply equal size kernels
  for all packing problems or hitting set problems where the
  structures that are packed or hitted may be modeled by the
  hyperedges of the input of the {\sc $r$-Set Packing} or {\sc
    $r$-Hitting Set} problems. It is an open challenge whether our
  technique can be applied to these general settings.

Finally, we wish to mention that the questions of whether
  \textsc{Triangle-Packing in Tournament} or \textsc{Feedback Vertex
    Set in Tournament} admit a kernel of linear number of vertices
  remain open problems.

\newpage


\end{document}